\documentclass[11pt, oneside]{article}
\usepackage{geometry}
\usepackage{graphicx}	
\usepackage{amssymb}
\usepackage{amsfonts,latexsym,amsthm,amssymb,amsmath,amscd,euscript}
\usepackage{bbm}
\usepackage{framed}
\usepackage{fullpage}
\usepackage{mathrsfs}
\usepackage{hyperref}

\usepackage{commath}

\hypersetup{
  colorlinks=true,
  linkcolor=blue,
  filecolor=blue,
  citecolor = black,      
  urlcolor=cyan,
}
\usepackage{accents}
\usepackage{setspace}
\hypersetup{colorlinks=true,citecolor=blue,urlcolor=black,linkbordercolor={1 0 0}}
\usepackage{tikz-cd}
\newcount\Comments  %
\ifnum\Comments=1
\usepackage[inline]{showlabels}

\fi
\usepackage{turnstile}
\usepackage{mathpartir}
\usepackage{tikz}
\usepackage{xspace}

\usepackage{algorithm}
\usepackage{verbatim}
\usepackage[noend]{algpseudocode}
\newcommand{\multiline}[1]{\parbox[t]{\dimexpr\linewidth-\algorithmicindent}{#1}}
\newcommand{\nc}{\newcommand}

\makeatletter
\newtheorem*{rep@theorem}{\rep@title}
\newcommand{\newreptheorem}[2]{%
\newenvironment{rep#1}[1]{%
 \def\rep@title{#2 \ref{##1}}%
 \begin{rep@theorem}}%
 {\end{rep@theorem}}}
\makeatother

\makeatletter
\newcommand\xlabel[2][]{\phantomsection\def\@currentlabelname{#1}\label{#2}}
\makeatother

\usepackage[nameinlink,capitalize]{cleveref}
\crefformat{equation}{#2(#1)#3}
\Crefformat{equation}{#2(#1)#3}

\Crefformat{figure}{#2Figure #1#3}
\Crefname{assumption}{Assumption}{Assumptions}
\Crefformat{assumption}{#2Assumption #1#3}
\Crefname{subsubsection}{Section}{Sections}
\crefformat{subsubsection}{#2Section #1#3}
\Crefformat{subsubsection}{#2Section #1#3}

\theoremstyle{plain}
\newtheorem{theorem}{Theorem}
\newreptheorem{theorem}{Theorem}
\newreptheorem{lemma}{Lemma}
\newtheorem{lemma}[theorem]{Lemma}
\newtheorem{Conjecture}[theorem]{Conjecture}
\newtheorem{corollary}[theorem]{Corollary}

\newtheorem{proposition}[theorem]{Proposition}

\newtheorem{claim}[theorem]{Claim}

\theoremstyle{definition}
\newtheorem{definition}{Definition}
\newtheorem{defn}[definition]{Definition}

\newtheorem{remark}[definition]{Remark}

\numberwithin{theorem}{section}
\numberwithin{definition}{section}

\nc{\sups}[1]{^{\scriptscriptstyle{#1}}}
\nc{\subs}[1]{_{\scriptscriptstyle{#1}}}

\nc{\dist}{\mathrm{dist}}
\newcommand{\PPAD}{\textsf{PPAD}\xspace}
\newcommand{\TFNP}{\textsf{TFNP}\xspace}

\nc{\DMO}{\DeclareMathOperator}

\DeclareMathOperator*{\argmin}{arg\,min} %

\DMO{\prox}{prox}
\DMO{\UCB}{UCB}
\DMO{\LCB}{LCB}
\nc{\phidiff}{\phi\sups{\Delta}}
\nc{\SolveLPFeasibility}{\texttt{SolveLPFeasibility}\xspace}
\nc{\BimatrixStrongSmooth}{\texttt{BimatrixStrongSmooth}\xspace}
\nc{\GeneralStrongSmooth}{\texttt{GeneralStrongSmooth}\xspace}
\nc{\VerifyDeviation}{\texttt{VerifyDeviation}\xspace}
\nc{\OptimizeDeviation}{\texttt{OptimizeDeviation}\xspace}
\nc{\QueryEquilibrium}{\texttt{QueryEquilibrium}\xspace}
\nc{\pexp}{q_{\mathrm{exp}}}
\nc{\nn}{\nonumber}
\nc{\rk}{\mathrm{rk}}
\nc{\brk}[3]{{\rm br}_{#1}^{#2}({#3})}
\nc{\co}{{\rm co}}
\nc{\Bin}{\mathrm{Bin}}
\nc{\br}[2]{{\rm br}^{#1}({#2})}
\nc{\depth}[1]{{\rm d}({#1})}
\nc{\tA}{\textsc{A}}
\nc{\child}[2]{{\rm ch}_{#1}({#2})}
\nc{\parent}[1]{{\rm pa}({#1})}
\nc{\dg}{\dagger}
\nc{\bB}{\mathbf{B}}
\nc{\Span}{{\rm Span}}
\nc{\unif}{\mathsf{unif}}
\nc{\indsig}[2]{\mathcal{I}_{#1}({#2})}
\nc{\total}{{\rm fin}}
\nc{\early}{{\rm pre}}
\nc{\zsink}{z_{\rm sink}}
\nc{\lowv}{{\rm low}}
\nc{\ol}{\overline}
\nc{\ul}{\underline}
\nc{\madec}[3]{\texttt{ma-dec}_{#1}({#2}, {#3})}
\nc{\madeco}[1]{\texttt{ma-dec}_{#1}}
\nc{\madecd}[3]{\texttt{ma-dec}^{\texttt{d}}_{#1}({#2}, {#3})}
\nc{\SF}{\mathscr{F}}
\nc{\SP}{\mathscr{P}}
\nc{\PiMarkov}{\Pi^{\rm markov}}
\nc{\Critic}{\texttt{Critic}}
\nc{\trunc}[2]{\mathsf{trunc}_{#2}({#1})}
\nc{\sbl}{of strong Bellman type\xspace}

\nc{\gamvec}{\gamma}
\nc{\til}{\widetilde}
\nc{\td}{\tilde}
\nc{\wh}{\widehat}
\nc{\todo}[1]{\ifnum\Comments=1 {\color{red}  [TODO: #1]}\fi}
\nc{\old}[1]{\ifnum\Comments=1 {\color{brown}  [OLD: #1]}\fi}

\nc{\noah}[1]{\ifnum\Comments=1 {\color{purple} [ng: #1]}\fi}
\nc{\abhishek}[1]{\ifnum\Comments=1 {\color{blue} [as: #1]}\fi}
\nc{\nika}[1]{\ifnum\Comments=1 {\color{red} [nh: #1]}\fi}

\nc{\dhruv}[1]{\ifnum\Comments=1 {\color{magenta} [dr: #1]}\fi}
\nc{\BP}{\mathbb{P}}
\nc{\BI}{\mathbb{I}}

\nc{\fools}[3]{\MF_{#3}({#1}, {#2})}
\nc{\fool}[2]{\MF({#1},{#2})}
\nc{\clip}[2]{{\rm clip}\left[ \left. {#1} \right| {#2} \right]}
\nc{\imax}{\omega}
\DMO{\conv}{conv}
\nc{\MH}{\mathcal{H}}
\nc{\MV}{\mathcal{V}}
\nc{\MC}{\mathcal{C}}
\nc{\MI}{\mathcal{I}}
\nc{\st}{\star}
\nc{\lng}{\langle}
\nc{\rng}{\rangle}
\DMO{\OOPT}{opt}
\nc{\dopt}[2]{\ell_{\OOPT}({#1},{#2})}
\nc{\grad}{\nabla}
\nc{\MG}{\mathcal{G}}
\nc{\MP}{\mathcal{P}}
\nc{\PP}{\mathbb{P}}
\nc{\TT}{\mathbb{T}}
\nc{\TTmax}{\TT_{\max}}
\DMO{\REG}{Reg}
\DMO{\WREG}{wReg}
\nc{\reg}[2]{{\Delta}_{{#1}}({#2})}
\nc{\wreg}[2]{{\Delta}^{\rm w}_{{#1}}({#2})}
\nc{\Reg}[2]{{\REG}_{{#1}}({#2})}
\nc{\wReg}[2]{{\WREG}_{{#1}}({#2})}
\DMO{\Ham}{Ham}
\DMO{\Gap}{Gap}
\DMO{\GD}{GD}
\DMO{\GDA}{GDA}
\DMO{\EG}{EG}
\nc{\TE}{\til{\E}}
\nc{\Var}{\mathbb{V}}
\DMO{\Cov}{Cov}
\DMO{\OGDA}{OGDA}
\DMO{\Unif}{\mathsf{Unif}}
\DMO{\Tr}{Tr}
\nc{\Qu}{\ul{Q}}
\nc{\Qo}{\ol{Q}}
\nc{\Ro}{\ol{R}}
\nc{\Vu}{\ul{V}}
\nc{\Vo}{\ol{V}}
\nc{\RanQ}{\Delta Q}
\nc{\RanV}{\Delta V}
\nc{\clipQ}{\Delta \breve{Q}}
\nc{\frzQ}{\Delta \mathring{Q}}
\nc{\clipV}{\Delta \breve{V}}
\nc{\clipdelta}{\breve{\delta}}
\nc{\cliptheta}{\breve{\theta}}
\nc{\delmin}{\Delta_{{\rm min}}}
\nc{\delmins}[1]{\Delta_{{\rm min},{#1}}}
\nc{\gapfinal}[1]{\max \left\{ \frac{\frzQ_{{#1}}^{k^\st}(x,a)}{2H}, \frac{\delmin}{4H} \right\}}
\nc{\post}[2]{R({#1}; {#2})}
\nc{\posts}[3]{R_{#3}({#1}; {#2})}

\nc{\algnst}[1]{\begin{align*}#1\end{align*}}
\nc{\algn}[1]{\begin{align}#1\end{align}}
\nc{\matx}[1]{\left(\begin{matrix}#1\end{matrix}\right)}

\nc{\nuu}{\nu}

\nc{\bel}[1]{\mathbf{b}({#1})}
\nc{\nbel}[1]{\bar{\mathbf{b}}({#1})}
\nc{\sbel}[2]{\mathbf{b}'_{#1}({#2})}
\nc{\nsbel}[2]{\bar{\mathbf{b}}'_{#1}({#2})}

\nc{\bv}{\mathbf{v}}
\nc{\bone}{\mathbf{1}}
\nc{\bX}{\mathbf{X}}
\nc{\bY}{\mathbf{Y}}
\nc{\bG}{\mathbf{G}}
\nc{\bz}{\mathbf{z}}
\nc{\bw}{\mathbf{w}}
\nc{\bA}{\mathbf{A}}
\nc{\bJ}{\mathbf{J}}
\nc{\bK}{\mathbf{K}}
\nc{\bb}{\mathbf{b}}
\nc{\ba}{\mathbf{a}}
\nc{\bc}{\mathbf{c}}
\nc{\bC}{\mathbf{C}}
\nc{\BR}{\mathbb R}
\nc{\BA}{\mathbb{A}}
\nc{\BC}{\mathbb C}
\nc{\bx}{\mathbf{x}}
\nc{\bS}{\mathbf{S}}
\nc{\bM}{\mathbf{M}}
\nc{\bR}{\mathbf{R}}
\nc{\bN}{\mathbf{N}}
\nc{\NN}{\mathbb{N}}
\nc{\by}{\mathbf{y}}
\nc{\sy}{y}
\nc{\sx}{x}

\nc{\MO}{\mathcal O}
\nc{\MU}{\mathcal{U}}
\nc{\ME}{\mathcal{E}}
\nc{\MN}{\mathcal{N}}
\nc{\MK}{\mathcal{K}}
\nc{\MM}{\mathcal{M}}
\nc{\MS}{\mathcal{S}}
\nc{\MT}{\mathcal{T}}
\nc{\BF}{\mathbb F}
\nc{\BQ}{\mathbb Q}
\nc{\MX}{\mathcal{X}}
\nc{\MA}{\mathcal{A}}
\nc{\MD}{\mathcal{D}}
\nc{\MB}{\mathcal{B}}
\nc{\MZ}{\mathcal{Z}}
\nc{\MJ}{\mathcal{J}}
\nc{\MW}{\mathcal{W}}
\nc{\MR}{\mathcal{R}}
\nc{\MY}{\mathcal{Y}}
\nc{\BZ}{\mathbb Z}
\nc{\BN}{\mathbb N}
\nc{\ep}{\epsilon}
\nc{\epbe}{\varepsilon_{\mathsf{BE}}}
\nc{\epout}{\varepsilon_{\mathsf{outlier}}}
\nc{\bellc}[1][h]{\MT_{#1}^\circ}
\nc{\vep}{\varepsilon}
\nc{\gapfn}[1]{\varepsilon_{#1}}
\nc{\ggapfn}[2]{\varphi_{#1}({#2})}
\nc{\epsahk}{\gapfn{0}}
\nc{\BH}{\mathbb H}
\nc{\BG}{\mathbb{G}}
\nc{\D}{\Delta}
\nc{\MF}{\mathcal{F}}
\nc{\One}[1]{\mathbbm{1}\{{#1}\}}
\nc{\bOne}{\mathbf{1}}
\nc{\Aopt}{\mathcal{A}^{\rm opt}}
\nc{\Amul}{\mathcal{A}^{\rm mul}}

\nc{\SQ}{\mathsf Q}

\nc{\DO}{\accentset{\circ}{\D}}
\nc{\mf}{\mathfrak}
\nc{\mfp}{\mathfrak{p}}
\nc{\mfq}{\mf{q}}
\nc{\Sp}{\mbox{Spec}}
\nc{\Spm}{\mbox{Specm}}
\nc{\hookuparrow}{\mathrel{\rotatebox[origin=c]{90}{$\hookrightarrow$}}}
\nc{\hookdownarrow}{\mathrel{\rotatebox[origin=c]{-90}{$\hookrightarrow$}}}
\nc{\hra}{\hookrightarrow}
\nc{\tra}{\twoheadrightarrow}
\nc{\sgn}{{\rm sgn}}
\nc{\aut}{{\rm Aut}}
\nc{\Hom}{{\rm Hom}}
\nc{\img}{{\rm Im}}
\DMO{\id}{Id}
\nc{\supp}{\mathsf{supp}}
\DMO{\KL}{KL}
\nc{\kld}[2]{\KL({#1}||{#2})}
\nc{\ren}[2]{D_2({#1}||{#2})}
\nc{\chisq}[2]{\chi^2({#1}||{#2})}
\nc{\tvd}[2]{D_{\mathsf{TV}}({#1}, {#2})}
\nc{\hell}[2]{D_{\mathsf{H}}^2({#1}, {#2})}
\nc{\dbi}[3][\pi]{D_{\mathsf{bi}}^{#1}({#2} \| {#3})}
\DMO{\BSS}{BSS}
\DMO{\BES}{BES}
\DMO{\BGS}{BGS}
\DMO{\poly}{poly}
\nc{\indep}{\perp}
\DMO{\sink}{sink}
\nc{\fp}[1]{\MP_1({#1})}
\nc{\BO}{\mathbb{O}}
\nc{\BT}{\mathbb{T}}

\nc{\RR}{\mathbb{R}}
\nc{\Gradient}{\nabla}
\DMO{\diag}{diag}
\nc{\EE}{\mathbb{E}}
\nc{\MQ}{\mathcal{Q}}

\DMO{\PR}{Pr}
\renewcommand{\Pr}{\PR}
\nc{\E}{\mathbb{E}}
\nc{\ra}{\rightarrow}
\renewcommand{\t}{\top}
\nc{\hc}{\{0,1\}^n}
\nc{\pmhc}[1]{\{-1,1\}^{#1}}
\nc{\uad}{\quad}

\newcommand{\ip}[2]{ \left\langle {#1}, {#2} \right\rangle}

\usepackage{multirow}

\title{Smooth Nash Equilibria: Algorithms and Complexity}

\author{ Constantinos Daskalakis\thanks{Email: \texttt{costis@csail.mit.edu}. Supported by NSF Awards CCF-1901292, DMS-2022448, and DMS2134108, a Simons Investigator Award, the Simons Collaboration on the Theory of Algorithmic Fairness.} \\ MIT \and Noah Golowich\thanks{Email: \texttt{nzg@mit.edu}. Supported by a Fannie \& John Hertz Foundation Fellowship and an NSF Graduate Fellowship.} \\ MIT  \and Nika Haghtalab \thanks{Email: \texttt{nika@berkeley.edu}. This work was supported in part by the National Science Foundation under grant CCF-2145898,
a C3.AI Digital Transformation Institute grant, Google faculty scholar award, and Berkeley AI Research Commons grants.} \\ UC Berkeley \and Abhishek Shetty\thanks{Email: \texttt{shetty@berkeley.edu}.  Supported by an Apple AI/ML PhD Fellowship.} \\ UC Berkeley  }

\date{September 19, 2023}

\begin{document}
\maketitle

\begin{abstract}
  A fundamental shortcoming of the concept of Nash equilibrium is its computational intractability: approximating Nash equilibria in normal-form games is \PPAD-hard.
  In this paper, inspired by the ideas of smoothed analysis, we introduce a relaxed variant of Nash equilibrium called \emph{$\sigma$-smooth Nash equilibrium}, for a {smoothness parameter} $\sigma$. 
  In a $\sigma$-smooth Nash equilibrium, players only need to achieve utility at least as high as their best deviation to a \emph{$\sigma$-smooth strategy}, which is a distribution that does not put too much mass (as parametrized by $\sigma$) on any fixed action. We distinguish two variants of $\sigma$-smooth Nash equilibria: \emph{strong} $\sigma$-smooth Nash equilibria, in which players are required to play $\sigma$-smooth strategies under equilibrium play, and \emph{weak} $\sigma$-smooth Nash equilibria, where there is no such requirement.
  
  We show that both weak and strong $\sigma$-smooth Nash equilibria have superior computational properties to Nash equilibria: when $\sigma$ as well as an approximation parameter $\ep$ and the number of players are all constants, there is a \emph{{constant-time}} randomized algorithm to find a weak $\ep$-approximate  $\sigma$-smooth Nash equilibrium in normal-form games. In the same parameter regime, there is a \emph{polynomial-time} deterministic algorithm to find a strong $\ep$-approximate $\sigma$-smooth Nash equilibrium in a normal-form game. These results stand in contrast to the optimal algorithm for computing $\ep$-approximate Nash equilibria, which cannot run in faster than quasipolynomial-time. 
   We complement our upper bounds by showing that when either $\sigma$ or $\ep$ is an inverse polynomial, finding a weak $\ep$-approximate $\sigma$-smooth Nash equilibria becomes computationally intractable. 
   
   Our results are the first to propose a variant of Nash equilibrium which is computationally tractable, allows players to act independently, and which, as we discuss, is justified by an extensive line of work on individual choice behavior in the economics literature.
\end{abstract}

\section{Introduction}
The notion of \emph{Nash equilibrium}, which is a strategy profile of a game in which each player acts independently of the other players in a way that best responds to them, has been a mainstay in the study of game theory and economics over the last several decades \cite{10.1257/jel.37.3.1067}, with a wide range of applications in related areas. 
For instance, it provides a model that predicts agents' behavioral patterns \cite{doi:10.1073/pnas.0308738101}, though with varying degrees of success, serves as a universal solution concept in multiagent learning settings \cite{filar2007competitive,zhang2019multi}, and furnishes a rich set of problems for study of the computational aspects of equilibria \cite{DBLP:journals/talg/Daskalakis13,babichenko2020informational}, among other applications.

Most known results pertaining to the computation of a Nash equilibrium in general-sum normal-form games are negative. In particular, it is known to be \PPAD-hard to compute Nash equilibria, even approximately \cite{DBLP:journals/siamcomp/DaskalakisGP09,DBLP:journals/jacm/ChenDT09,  DBLP:journals/talg/Daskalakis13,rubinstein2015inapproximability}, meaning that there is unlikely to be a polynomial-time algorithm. 
Similar hardness results abound in other natural models of computation, such as the \emph{query complexity} model. 
In particular, even in 2-player normal-form games, one has to query a nearly constant fraction of the payoff matrices' entries to find an approximate Nash equilibrium \cite{fearnley2014learning,DBLP:journals/teco/FearnleyS16,DBLP:conf/focs/GoosR18}.
The computational intractability of Nash equilibria casts doubt over whether it is actually a reasonable solution concept to model agents' play: real-world agents can only implement efficient algorithms, so should not be expected to play Nash equilibria in all games \cite{DBLP:journals/siamcomp/DaskalakisGP09,DBLP:journals/talg/Daskalakis13,DBLP:journals/geb/HartM10}.\footnote{This perspective is well summarized by a quote by Kamal Jain: ``If your laptop can’t find it then neither can the market''.} %
Moreover, for applications in which Nash equilibria do yield high-utility strategies for agents, its computation nevertheless presents significant challenges \cite{marris2022turbocharging}.

To address these issues, it is popular to consider relaxations of Nash equilibria. Given the competitive nature of many games, one fudamental property of Nash equilibria we would like such a relaxation to preserve is that players choose their actions independently.  Second, of course, we would like such a relaxation to be computationally tractable.
We remark that the well-known equilibrium concepts of \emph{coarse correlated equilibria} \cite{moulin1978strategically} and \emph{correlated equilibria} \cite{aumann1987correlated}, which are relaxations of Nash equilibria, can be computed in polynomial time, but generally involve correlation in the actions of the players.
In particular, implementing them requires a \emph{correlation scheme}, which draws a random variable and sends each player a private signal depending on it. 
Such schemes with guaranteed privacy of the signals may not be available in certain applications, %
motivating our requirement that players' strategies be independent. %

\paragraph{Smooth Nash equilibria.} In this paper we introduce \emph{smooth Nash equilibria}, which meet both of the above requirements. 
The inspiration for smooth Nash equilibria comes from  smoothed analysis \cite{DBLP:journals/cacm/SpielmanT09}, which aims to circumvent computational hardness barriers by postulating that inputs to algorithms come from \emph{smooth distributions}, which do not place too much mass on any single value. Alternatively, one may think of smooth distributions as being perturbations of arbitrary distributions. There are two ways to apply smoothed analysis in the context of game theory: first, one can assume the game itself (i.e., the payoff matrices) comes from a smooth distribution, and second, one could assume the game is worst-case but that the players choose their strategies according to smooth distributions. The former approach has been well-studied, but it fails to sidestep the \PPAD-hardness barrier~\cite{DBLP:journals/jacm/ChenDT09,DBLP:conf/focs/BoodaghiansBHR20}.
We focus on the latter approach, which has attracted more attention in recent years, wherein players' strategies (usually the adversarial player's strategies) are smoothed~\cite{DBLP:conf/nips/HaghtalabRS20,DBLP:conf/focs/HaghtalabRS21,DBLP:conf/nips/HaghtalabHSY22,DBLP:conf/colt/BlockDGR22,bhatt2023smoothed,DBLP:conf/colt/BlockSR23,DBLP:conf/colt/BlockP23,DBLP:journals/corr/abs-2301-11187,rakhlin2011online}. %

In particular, given a normal-form game $G$ where each player has $n$ actions and a \emph{smoothness parameter} $\sigma \in (0,1)$, we say that a strategy profile is a \emph{$\sigma$-smooth Nash equilibrium} if each player cannot improve its utility by switching to any \emph{$\sigma$-smooth} distribution over actions. A $\sigma$-smooth distribution is one which does not place too much mass, namely, more than $1/(\sigma n)$, on any single action. Further, given $\ep \in (0,1)$, an \emph{$\ep$-approximate $\sigma$-smooth Nash equilibrium} is a strategy profile in which each player cannot improve its utility by more than $\ep$ by switching to a $\sigma$-smooth distribution. We refer the reader to \cref{sec:smoothed-equilibria} for the formal definitions.

The computational properties of smooth Nash equilibrium (as well as its approximate counterpart) depend heavily on the smoothness parameter $\sigma\in (0,1)$. It is easy to see that when $\sigma = 1/n$, the definition of smooth Nash equilibrium coincides with that of Nash equilibrium, whereas when $\sigma = 1$, the concept is vacuous: the uniform distribution over $[n]$ is a smooth Nash equilibrium. Our focus is on the intermediate regime $1/n < \sigma < 1$. We also make a distinction between two versions of (approximate) smooth equilibrium: we say that a \emph{strong} $\ep$-approximate $\sigma$-smooth Nash equilibrium is one in which players are required to play $\sigma$-smooth distributions under equilibrium play; a \emph{weak} $\ep$-approximate $\sigma$-smooth Nash equilibrium does not have this additional requirement. 

\paragraph{Our contributions.} We study the complexity of computing $\ep$-approximate $\sigma$-smooth Nash equilibria in $m$-player, $n$-action normal-form games, for various values of $\ep, \sigma \in (0,1)$. Our focus is on the setting when the number of players $m$ is a constant, though our bounds apply more generally. In this setting, our main results are summarized below, where we let $n$ denote the number of actions of each player: %
\begin{itemize}
\item When $\ep, \sigma$ are constants, there is a $\poly(n)$-time deterministic algorithm which finds a \emph{strong} $\ep$-approximate $\sigma$-smooth Nash equilibrium (\cref{thm:alg-strong}). 
\item When $\ep, \sigma$ are constants, there is a \textbf{\emph{constant-time}} randomized algorithm which finds a \emph{weak} $\ep$-approximate $\sigma$-smooth Nash equilibrium (\cref{thm:weak-nash-constant}). Furthermore, the number of queries this algorithm makes to the payoff matrices is bounded above by $\poly(\ep^{-1}, \sigma^{-1})$ (\cref{thm:query-equilibrium}).
\item There is a constant $\sigma_0$ so that, for $\ep = 1/\poly(n)$, it is \PPAD-hard to compute weak $\ep$-approximate $\sigma_0$-smooth Nash equilibrium in 2-player games (\cref{thm:ppad-hard-const-sig}). Moreover, there is a constant $\ep_0$ so that, for $\sigma = 1/\poly(n)$, there is no algorithm running in time $n^{o(\log n)}$ which computes weak $\ep_0$-approximate $\sigma$-smooth Nash equilibrium in 2-player games under the Exponential Time Hypothesis for \PPAD (\cref{cor:smooth-nash-eth}). 
\end{itemize}

In the case that $\sigma = 1/n$, in which case weak and strong $\sigma$-smooth Nash equilibria are both equivalent to Nash equilibria, our results on time complexity recover well-known results on the complexity of Nash equilibria (assuming that $m$ is constant). 
In particular, in this regime, our upper bound of $n^{O(\log(1/\sigma)/\ep^2)}$ becomes $n^{O(\log(n)/\ep^2)}$, which recovers the classical result of \cite{lmn}.
In fact, under ETH for \PPAD, this is the best possible upper bound for computing approximate Nash equilibria in $2$-player games \cite{DBLP:journals/sigecom/Rubinstein17}.
Furthermore, any algorithm for computing approximate Nash equilibria in $2$-player games must make $ n^{2 - o(1)}  $ queries \cite{DBLP:conf/focs/GoosR18} which is in contrast with our \emph{constant-time} randomized algorithm for finding weak $\sigma$-smooth Nash equilibria when $\sigma$ is constant.

Thus, the parameters $\ep, \sigma$ act as dials which allow us to trade off computational efficiency for a stronger equilibrium notion. 
Our results are summarized in \hyperref[tab:results]{Table \ref*{tab:results}}.

\begin{table}[htbp]    \label{tab:results}
    \centering
    \begin{tabular}{|c|c|c|c|}
        \hline
        Model & Upper Bound  &  \multicolumn{2}{|c|}{Lower Bounds}  \\
        \hline
        \multirow{2}{*}{Query Complexity} & $  m \cdot \left( \frac{m \log^2(m/(\delta \sigma \ep))}{\ep^2 \sigma} \right)^{m+1} $  & {$\Omega(\exp(m))$} & \multirow{2}{*}{$\epsilon = \Omega (1)$, $\sigma = \Omega(1)$}         \\
        & \cref{thm:query-equilibrium}$^\dagger$ & \cref{lem:query-lb} &  \\ 
        \hline
        \multirow{6}{*}{Time Complexity} & \multirow{3}{*}{ $ n^{O( m^4 \log(m/ \sigma) / \epsilon^2 )  }  $}   &  {\PPAD-hard} &    \multirow{2}{*}{$\epsilon = n^{-c}$, $\sigma = \Omega(1)$}   \\ 
        & & \cref{thm:ppad-hard-const-sig} &     \\  
        \cline{3-4}
        & \cref{thm:alg-strong}$^\star$ & $n^{\log n  }$ under ETH-\PPAD & \multirow{2}{*}{$\epsilon = \Omega (1) $, $\sigma = n^{-c} $}   \\ 
        \cline{2-2}
        &  \multirow{2}{*}{$\left( \frac{m \log(1/\delta)}{\sigma \ep} \right)^{O\left( \frac{m^2 \log(m/\delta \sigma)}{\ep^2} \right)}$} & \cref{cor:smooth-nash-eth} &  \\ 
        \cline{3-4}
        & & {$\Omega(\exp(m))$} & \multirow{2}{*}{$\epsilon = \Omega (1)$, $\sigma = \Omega(1)$}    \\ 
        & \cref{thm:weak-nash-constant}$^\dagger$  & \cref{lem:query-lb} & \\ 
        \hline     
    \end{tabular}
    \caption{Summary of our results. {\small $\dagger$: These upper bounds are obtained by randomized algorithms, and compute \emph{weak} smooth Nash equilibria with probability $1-\delta$. $\star$: This upper bound is obtained by a deterministic algorithm, and computes \emph{strong} smooth Nash equilibria. All lower bounds hold with respect to weak (and therefore also strong) Nash equilibria, and hold with respect to randomized algorithms.}}
\end{table}

\subsection{Further Motivation \& Connections to Other Equilibrium Notions}
\label{sec:other-equilibria} 

Implicit in the definition of $\sigma$-smooth Nash equilibria is the assumption that agents only aim to compete with their best $\sigma$-smooth response. It is natural to question how reasonable this assumption is. In this section, we discuss an extensive line of work on individual choice behavior in the economics literature which has yielded choice models that help to justify $\sigma$-smooth Nash equilibria. In doing so, we also discuss connections between $\sigma$-smooth Nash equilibria and  other variants of Nash equilibria which have been previously proposed.

\paragraph{Quantal Response Equilibria \& Logit Equilibria.} A fundamental difference between $\sigma$-smooth Nash equilibrium and the standard notion of Nash equilibrium is that in the former, agents do not place all mass of their strategy on their utility-maximizing action(s), given the strategies of others. One may wonder: \emph{Is this realistic behavior?} The theory of \emph{random utility models} \cite{mcfadden1976quantal} proposes a justification for this type of behavior. In particular, suppose that, for fixed strategies of other agents, each agent $j$'s computation of its expected utilities for each action is perturbed by a random noise vector $\vep_j$. This noise could represent calculation errors due to imperfect information an agent has about the game's payoffs or other agents' strategies. As a result of this noise, conditioned on other agents' strategies an agent will generally play some actions which are suboptimal. For a fixed distribution over the agents' noise vectors $(\vep_1, \ldots, \vep_m)$, a strategy profile in which each agent $j$'s strategy is obtained as above by perturbing its (true) expected payoff by $\vep_j$ and then best-responding, is known as a \emph{quantal response equilibrium} (\cref{def:qre}) \cite{mckelvey1995quantal,goeree2016quantal}. Moreover, the mapping from an agent's expected payoff vector $u$ to its strategy obtained by best-responding to the perturbed vector $u + \vep_j$ is known as its \emph{quantal response function}. 

 We remark that an agent's best $\sigma$-smooth response yields larger utility than its quantal response function for many choices of the perturbations $\vep_j$. This observation provides further justification for the set of $\sigma$-smooth deviations as representing a reasonable deviation set with which to define equilibria. To illustrate, consider the common setting where the perturbations $\vep_j$ are drawn according to the extreme value distribution. In this case, given an expected payoff vector $u \in \BR^n$, an agent's quantal response function is given by a softmax function of $u$. %
 This class of response functions, known as the \emph{logistic response functions}, has its origins in the seminal work of Luce \cite{luce1959individual}, and the corresponding special case of quantal response equilibrium, called \emph{logit equilibrium}, has been extensively studied since (e.g., \cite{mcfadden1976quantal,anderson2002logit,haile2008empirical,goeree2002quantal,goeree2000asymmetric}). We show in \cref{prop:logit-smooth} that, under mild conditions on $\sigma$, the utility obtained by the logistic response function is no better than that obtained by the best $\sigma$-smooth distribution.

In light of the above connections, one might wonder if $\sigma$-smooth Nash equilibria can in fact be obtained as a special case of quantal response equilibria, for some distribution of the noise vectors $\vep_j$. As we discuss in \cref{sec:qre-se}, this is not the case: for any distribution of the noise vectors $\vep_j$, there are games for which the set of quantal response equilibria is separated in total variation distance from the set of strong $\sigma$-smooth Nash equilibria (\cref{prop:qre-sne}).

Finally, we mention that a key property of $\sigma$-smooth Nash equilibrium which is essential to obtaining efficient algorithms is that there is a natural notion of $\ep$-approximate smooth equilibrium (\cref{def:sne}). All of our efficient algorithms require $1/\ep$ to be bounded. There does not appear to be a correspondingly natural notion of approximation for quantal response equilibria.

\paragraph{Polyhedral games.} The set of strong $\sigma$-smooth Nash equilibria for a given normal-form game may be viewed as the set of Nash equilibria in the \emph{concave polyhedral game} in which each player's action set is the polyhedron consisting of all $\sigma$-smooth distributions, and the payoffs are defined by linearity. The study of (polyhedral) concave games was initiated by \cite{rosen1965existence}, and has since been studied in various contexts, including for analyzing no-regret learning algorithms \cite{gordon2008noregret} and modeling extensive-form games \cite{farina2022kernelized}, amongst others \cite{mansour2022strategizing}. Our results on strong $\sigma$-smooth equilibria can be viewed as showing that there are efficient algorithms for computation of approximate Nash equilibria in a certain natural class of polyhedral concave games.

\paragraph{Additional equilibrium notions.} \cite{rosenthal1989bounded} introduced a notion of \emph{bounded rationality equilibrium} in which the probability that each player plays an action is related linearly to the suboptimality of that action. The resulting equilibria have similar behavior to quantal response equilibria \cite{mckelvey1995quantal,goeree2016quantal}. \cite{beja1992imperfect} introduced \emph{imperfect performance equilibria}, which are Nash equilibria in an augmented game in which each agent $j$ intends to play a certain strategy $\tilde x_j$, but the strategy it actually plays, $x_j$, is the result of passing $\tilde x_j$ through a \emph{performance mapping}. This performance mapping could simulate the fact that ``agents'' in a game can represent organizations, whose actions are the agglomeration of those of multiple individuals. As these individuals may be self-interested, they may not always act according to the organization's best interests (i.e., its best response). 

\emph{Trembling-hand perfect equilibria}, introduced by \cite{selten1975reexamination}, are those which are obtained as limit points of a sequence of polyhedron games, parametrized by $\delta \searrow 0$, in which each player's polyhedron is the set of strategies which put mass at least $\delta$ on each action. Note that this polyhedron is distinct from the polyhedron of smooth strategies, which require that each action receives mass \emph{no more than} $1/(\sigma n)$. Trembling-hand perfect equilibria are further refined by \emph{proper equilibria} \cite{myerson1978refinements}. We remark that trembling-hand perfect equilibria are related to certain limit points of \emph{control-cost equilibria} \cite[Chapter 4]{vandamme1987stability}, which are equilibria in which each player pays a cost for playing each action, which depends on the probability with which the action is played. The above notions of equilibria are inspired by various justifications as to why agents might play suboptimal actions with positive probability. However, unlike smooth Nash equilibria, none of them can be computed efficiently in general normal-form games. In fact, trembling-hand perfect equilibria and proper equilibria are refinements of Nash equilibria, so an efficient algorithm for computing them would imply an efficient algorithm for computing Nash equilibria. 

Finally, we mention the works \cite{awasthi2012nash,balcan2017nash}, which show that (the standard notion of) Nash equilibrium may be computed more efficiently in games which satisfy certain stability properties. In particular, \cite{awasthi2012nash} establishes an improvement to the quasipolynomial-time algorithm of \cite{lmn} if it is guaranteed that all $\ep$-approximate Nash equilibria lie in a small ball centered around an exact Nash equilibrium. \cite{balcan2017nash} show a similar quantitative improvement under the more general condition that under any small perturbation of the payoff matrices of the game, any Nash equilibrium of the perturbed game is close to some Nash equilibrium of the original game. We emphasize that our results are orthogonal to \cite{awasthi2012nash,balcan2017nash} in that we focus on computing a \emph{relaxed} notion of Nash equilibrium in \emph{worst-case} games.

\paragraph{Regularization and equilibrium solving.} Quantal response equilibria can  be interpreted as Nash equilibria in a modified game where each player's utility has an additional regularization term that depends on the entropy of its strategy. The idea of using entropy regularization, as well as generalizations of it, has been exploited in recent years \cite{DBLP:conf/icml/JacobWFLHBAB22,DBLP:conf/iclr/Bakhtin0LGJFMB23,DBLP:journals/corr/abs-2206-15378,doi:10.1126/science.ade9097,sokota2023unified} to instill desirable properties of solutions to multi-agent reinforcement learning problems. For instance, to find strategies which are ``human-like'',  \cite{DBLP:conf/icml/JacobWFLHBAB22} does the following:  for a fixed, known human strategy $\tilde x$, they introduce a regularizer which minimizes the KL divergence to $\tilde x$ (a generalization of entropy), thus ensuring that learned strategies will be similar to $\tilde x$, while being close to a best response in some sense. 
In a similar spirit, smooth Nash equilibria can be seen as equilibria induced by regularization with respect to the \emph{min-entropy}. Moreover, appropriate generalizations of the min-entropy, as per \cref{lem:generalize-smooth-polytope}, could be used to model the task of finding equilibria which are close to arbitrary (non-uniform) distributions, such as ones played by humans. Whereas most of the prior work on regularization in this context (e.g., \cite{sokota2023unified}) focuses on two-player zero-sum games, our concept of smooth Nash equilibrium leads to efficient algorithms in multi-player, general-sum games. It is an exciting question to understand if our ideas have implications in such general-sum versions of these multi-agent RL problems.

\paragraph{Boosting.} 
The best-response condition of smooth Nash equilibrium is also motivated by a connection to boosting algorithms. 
The link between two-player zero-sum games and boosting is well-known \cite{10.5555/2207821} and can be seen as a fundamental application of equilibrium solving.  
One of the key aspects of the analysis of boosting algorithms (such as ADABoost) which allows for rates that do not depend on the number of samples is the fact that one needs to only compete against subsets of the dataset of size $\eta n$ where $n$ is the size of the dataset and $\eta$ is the desired accuracy \cite[Sections 3.6 and 3.7]{v008a006}.
This can be viewed as the best-response condition of weak $\eta$-smooth Nash equilibrium, where the ``dataset'' player who places weights on the sample set need only  compete with smooth distributions. 
Furthermore, in several applications such as the construction of hard-core sets \cite{DBLP:conf/focs/Impagliazzo95,DBLP:journals/eccc/Kale07,DBLP:conf/soda/BarakHK09}, noise resilient (smooth) boosting \cite{DBLP:conf/colt/Servedio01} and construction of private datasets \cite{DBLP:conf/nips/HaghtalabRS20}, it is important that the distribution output by the ``dataset'' player  does not concentrate too much mass on any single point (see \cite[Section 3.7]{v008a006} for a further discussion). 
This can be viewed as the best-response condition of the strong smooth Nash equilibrium.

\paragraph{Organization.}
In \cref{sec:smoothed-equilibria}, we formally define smooth Nash equilibria and present an important structural result in \cref{sec:sampling}. 
In \cref{sec:query}, we present our query-efficient algorithms for finding weak smooth equilibria.
In \cref{sec:eff-algs}, we discuss efficient algorithms for finding smooth equilibria, including both the weak (\cref{sec:weak-ub}) and strong (\cref{sec:strong-ub}) variants. %
\cref{sec:hardness} contains our hardness results, and \cref{sec:discussion} contains concluding remarks and open questions. %
\cref{sec:zero-sum} contains discussion of smooth equilibria in zero-sum games.

\section{Preliminaries}
\label{sec:prelim}

Given a tuple of elements $y =  (y_1 , \dots , y_m) $, we will denote by $ y_{-i} $ the tuple $ (y_1 , \dots , y_{i-1} , y_{i+1} , \dots , y_m ) $.
Further, for an element $y_i'$, we will use the standard game theoretic notation of $ (y_i' , y_{-i})  $ to refer to the tuple $ \left( y_1, \dots , y_{i-1} , y_i', y_i , \dots , y_m \right)  $. 
In particular, $  y = (y_i , y_{-i})  $. %
For $i \leq i'$, we write $y_{i:i'} = (y_i, y_{i+1}, \ldots, y_{i'})$. 
For $n \in \BN$, we let $\Delta^n$ denote the set of distributions over $[n]$. For $i \in [n]$, let $e_i \in \Delta^n$ denote the corresponding basis vector. For $y \in \Delta^n$, we let $\supp(y) := \{ i :\ y_i > 0 \}$.

\begin{definition}[Normal Form Games]
  Let $n$ and $m$ be positive integers. 
  A normal form game with $m$ players each with $n$ actions is specified by a tuple of \emph{payoff mappings} $A_1 , \dots , A_m  : \left[ n \right]^m \to \left[ 0,1 \right] $, with the interpretation that player $j$ receives utility $A_j \left( a_1 , \dots , a_m \right) $ when the players play the actions $a_1 , \dots , a_m $ respectively.
\end{definition}

  A \emph{mixed strategy} (or simply \emph{strategy}) of player $j$ is an element $x_j \in \Delta^n$. A \emph{strategy profile} is a tuple $x = (x_1, \ldots, x_m) \in (\Delta^n)^m$. 
  We will extend the notation $A_j(a_1, \ldots, a_m)$ by linearity: for a strategy profile $x = (x_1, \ldots, x_m)$, we write
  \begin{align}
    A_j (x ) = A_j( x_1 , \dots , x_m ) = \mathbb{E}_{ a_1 \sim x_1  } \dots \mathbb{E}_{ a_m \sim x_m } \left[ A_j \left( a_1 , \dots , a_m \right) \right]. \nonumber
  \end{align}    
 For an action $a_j \in [n]$ and a strategy profile $x$, we will write $A_j(a_j, x_{-j})$ to denote $A_j(e_{a_j}, x_{-j})$, namely player $j$'s expected utility when it plays $a_j$ and others play according to $x$.

 \begin{definition}[Nash equilibrium]
   \label{def:ne}
  Given $\ep \in (0,1)$ and a normal-form game specified by $A_1, \ldots, A_m$, a strategy profile $x = (x_1, \ldots, x_m) \in (\Delta^n)^m$ is an \emph{$\ep$-approximate Nash equilibrium} if for each $j \in [m]$,
  \begin{align}
\max_{x_j' \in  \Delta^n } A_j(x_j', x_{-j}) - A_j(x) \leq \ep\nonumber.
  \end{align}
  We will refer to a $0$-approximate Nash equilibrium as simply a \emph{Nash equilibrium}. 
\end{definition}

In the special case of 2-player games, we will denote the payoff functions of the game by matrices $A,B \in \BR^{n \times n}$, and the players' respective strategy profiles by $x,y \in \Delta^n$.

\section{Smooth Equilibria} \label{sec:smoothed-equilibria}

First, we define the polytope of smooth distributions, which is the set of distributions that do not concentrate too much mass on any one coordinate. %

\begin{definition}[Smooth Distribution Polytope]
  \label{def:smooth-polytope}
  Suppose $n \in \BN$ and $\sigma \in \BR$ satisfy $ 1/n \leq   \sigma \leq 1 $. 
  The \emph{smooth distribution polytope} $\MK_{\sigma, n}$ is the set %
  \begin{align} \label{eq:smooth-distribution-polytope}
    \MK_{\sigma,n} = \left\{ x \in \Delta^n : 0 \leq  x_i \leq \frac{1}{n\sigma}  \ \ \forall i \in [n] \right\}. 
  \end{align}
\end{definition}
We call the parameter $\sigma $ the \emph{smoothness parameter}; %
for larger $\sigma$, the constraints defining $\MK_{\sigma, n}$ are more stringent, i.e., distributions in $\MK_{\sigma, n}$ are ``smoother''. %
In particular, for $\sigma = 1/n$ we have $ \MK_{\sigma,n} = \Delta^n   $ and for $\sigma = 1$ we have $ \MK_{\sigma,n} = \left\{ \frac{1}{n} \mathbf{1} \right\} $ where $\mathbf{1}$ is the all ones vector.
With this definition in hand, we now define smooth Nash equilibria, the key concept that we study in this paper.
\begin{definition}[Smooth Nash equilibria]
  \label{def:sne}
  Fix an $m$-player normal-form game $G$ with payoff mappings $A_1, \ldots, A_m : [n]^m \ra [0,1]$, and $\ep, \sigma \in [0,1]$. A strategy profile $x = (x_1, \ldots, x_m) \in (\Delta^n)^m$ is an \emph{weak $\ep$-approximate $\sigma$-smooth Nash equilibrium} if, for each $j \in [m]$,
  \begin{align}
\max_{x_j' \in \MK_{\sigma, n}} A_j(x_j', x_{-j}) - A_j(x) \leq \ep\nonumber.
  \end{align}
  It is moreover a \emph{strong $\ep$-approximate $\sigma$-smooth Nash equilibrium} if $x_j \in \MK_{\sigma, n}$ for each $j \in [m]$. 
\end{definition}
We refer to weak (resp., strong) 0-approximate $\sigma$-smooth Nash equilibria as simply weak (resp., strong) $\sigma$-smooth Nash equilibria. 
Intuitively, \cref{def:sne} says that in $\sigma$-smooth Nash equilibria, no player can gain more than $\ep$ utility by deviating to a distribution that is $ \sigma $-smooth. 
Furthermore, in strong $\sigma$-smooth Nash equilibria, the mixed strategies of each player are themselves $ \sigma $-smooth.

The set of strong $\sigma$-smooth Nash equilibria of a normal-form game is equal to the set of Nash equilibria in the polyhedral concave game \cite{rosen1965existence} where each player's action set is $\MK_{\sigma, n}$ and the payoffs are given by $x \mapsto A_j(x)$, for $j \in [m],\ x \in (\MK_{\sigma, n})^m$. Since the sets $\MK_{\sigma, n}$ are compact and convex, \cite[Theorem 1]{rosen1965existence} implies that strong $\sigma$-smooth Nash equilibria always exist. 
\begin{proposition}[Existence of Equilibria; \cite{rosen1965existence}] \label{thm:existence} 
For any $\sigma \in (0,1)$, every normal-form game has a strong $\sigma$-smooth Nash equilibrium. 
\end{proposition}

\begin{remark}[Generalizing the smooth distribution polytope] \label{lem:generalize-smooth-polytope} 
  The smooth polytope $\MK_{\sigma, n}$ may alternatively be defined as the set of distributions on $[n]$ whose Radon-Nikodym derivative with repsect to the uniform distribution on $[n]$ is bounded by $1/\sigma$. In particular, 
$ 
  \MK_{\sigma,n} = \left\{ x \in \Delta^n :  \left\|\frac{d x }{ d \mu_n }\right\|_\infty \leq \frac{1}{\sigma} \right\},
$
where $ \mu_n $ is the uniform measure on $ [n] $. 
As is common in smoothed analysis (e.g., \cite{DBLP:conf/colt/BlockDGR22,DBLP:conf/colt/BlockSR23,DBLP:conf/colt/BlockP23,DBLP:journals/corr/abs-2301-11187}), one may generalize the notion of smooth distributions and thereby smooth equilibria by allowing $\mu_n$ to be an arbitrary distribution on $[n]$ or even a more general probability space. %
Essentially all of our results remain unchanged under appropriate compactness assumptions. %
\end{remark}

\subsection{Structure of Smooth Nash Equilibria}
\label{sec:sampling}
Next, we discuss a key structural result underlying our upper bounds, which states that there are weak approximate smooth Nash equilibria for which the players' strategy profiles are supported on a small number of actions (\cref{lem:existence_k_unif}). In fact, we prove a slightly stronger statement, namely that %
a smooth Nash equilibrium can be approximated by sampling. 
\begin{lemma}
    \label{lem:mplayer-approx}
    Let $A_1, \ldots, A_m : [n]^m \ra [0,1]$ denote the payoff matrices of an $m$-player normal-form game $G$. Let $\ep, \sigma \in (0,1)$ be fixed. Let 
$ 
      k = \frac{C_{\ref{lem:mplayer-approx}}m \log(8m/ \delta \sigma)}{\ep^2}
$ 
    where $C_{\ref{lem:mplayer-approx}}$ is a sufficiently large constant. 
    Given a weak $\sigma$-smooth Nash equilibrium $x = (x_1, \ldots, x_m)$, let $ \hat x = (\hat x_1, \ldots, \hat x_m)  $ be the random strategy profile where $\hat x_j$ is the uniform average over $k$ actions $B_{j,1}, \ldots, B_{j,k}$ sampled i.i.d.~from $ x_j$. %
    Then, with probability $ 1 - \delta $, we have, for all $j \in [m]$, 
    \begin{align} \label{eq:s-unif-close}
    |A_j(x_1, \ldots, x_m) - A_j(\hat x_1, \ldots, \hat x_m)| \leq \frac{\ep}{2} , \qquad \sup_{x_j' \in \MK_{\sigma, n}} \left| A_j(x_j', \hat x_{-j}) - A_j(x_j',  x_{-j})\right| \leq \frac{\ep}{2},
    \end{align}
      and $\hat x$ satisfying \cref{eq:s-unif-close} is a weak $\ep$-approximate $\sigma$-smooth Nash equilibrium.
    \end{lemma}
The main technical lemma driving the sampling result of \cref{lem:mplayer-approx} is a generalization of Massart's finite class lemma to the case when the supremum is over the set of smooth distributions on $[n]$ (namely, \cref{lem:smooth-class-lemma}). 
At a high level, \cref{lem:smooth-class-lemma} allows us to bound the sampling error in each player's deviation to a best-response smooth strategy without incurring a $\log(n)$ factor in the number $k$ of samples. 

We capture the sparsity of strategy profiles such as $\hat x$ produced by the sampling procedure in \cref{lem:mplayer-approx}  with the following definition.
  \begin{definition}
    \label{def:k-uniform}
    For $k \in \BN$, we say that a distribution $y \in \Delta^n$ is \emph{$k$-uniform} if for each $i \in [n]$, $y_i$ is an integral multiple of $1/k$. We say that a strategy profile $x = (x_1, \ldots, x_m) \in (\Delta^n)^m$ is $k$-uniform if each of its constituent strategies $x_j$ is.
    \end{definition}

Thus, we see from \cref{lem:mplayer-approx} that $k$-uniform weak approximate smooth Nash equilibria exist for $k$ that is independent of the number of actions of the game. 
\begin{corollary}[Existence of $k$-uniform weak smooth Nash] \label{lem:existence_k_unif}
Let $m \in \BN$, $\ep, \sigma \in (0,1)$, be given and set  $k = \frac{C_{\ref{lem:mplayer-approx}}m \log(8m/ \sigma)}{\ep^2}$. Then any $m$-player normal-form game $G$ has a weak $k$-uniform $\ep$-approximate $\sigma$-smooth Nash equilibrium.
\end{corollary}

\cref{lem:existence_k_unif} generalizes the result of \cite{lmn}, as well as the follow-up of \cite{hemon2008approximate}, which treat the case of approximate Nash equilibrium,  corresponding to the case of $\sigma = 1/n$. The parameter $k$ needed by \cite{lmn,hemon2008approximate}, which governs the sparsity of the equilibrium, thus grows logarithmically in $n$. %
As we show in \cref{sec:eff-algs}, the fact that the sparsity is only logarithmic in $1/\sigma$ is key to getting a {polynomial}-time algorithm for computing smooth Nash equilibria when $1/\sigma = O(1)$.

\section{Query Complexity of Smooth Nash Equilibria}
\label{sec:query}
In this section, we prove an upper bound on the randomized query complexity of computing weak smooth Nash equilibria. This result comprises the bulk of our proof in \cref{sec:weak-ub} that there is a randomized constant-time algorithm for computing weak $\ep$-approximate  $\sigma$-smooth Nash equilibria. Beside the application to computational complexity, query complexity of equilibria is of interest in its own right,
as a tool to understand the amount of information that needs to be shared in order to find an equilibrium \cite{babichenko2020informational}. 
It is also closely related to the analysis of uncoupled dynamics converging to equilibrium \cite{DBLP:conf/icml/ConitzerS04,DBLP:journals/geb/HartM10,DBLP:conf/focs/GoosR18}.

We briefly review the query complexity model. 
 The payoff mappings $A_1, \ldots, A_m : [n]^m \ra [0,1]$ are assumed to be unknown, but the algorithm can repeatedly query single entries $A_j(a_1, \ldots, a_m)$ (for $j \in [m], a_1, \ldots, a_m \in [n]$) of these matrices. Randomness is allowed in choosing the queries. After making at most $Q$ queries, for $Q \in \BN$ denoting the \emph{query complexity}, the algorithm is required to output a strategy profile. 
Our main result in this model is \cref{thm:query-equilibrium} below, which states that we can find a smooth Nash equilibrium with a number of queries that is independent of the number of actions $n$. 
\begin{theorem}
  \label{thm:query-equilibrium}
  Let $A_1, \ldots, A_m : [n]^m \ra [0,1]$ denote the payoff matrices of an $m$-player normal-form game $G$. Let $\ep, \sigma, \delta \in (0,1)$ be fixed. Then $\QueryEquilibrium((A_1, \ldots, A_m), \sigma, \ep, \delta)$ (\cref{alg:query-equilibrium}) makes  $ O \left( m \cdot \left( \frac{m \log^2(m/(\delta \sigma \ep))}{\ep^2 \sigma} \right)^{m+1} \right)$ %
  queries to entries of payoff matrices of $G$. Moreover, it outputs a strategy profile $\hat x$ which is a weak $\ep$-approximate $\sigma$-smooth Nash equilibrium of $G$ with probability at least $1-\delta$.
\end{theorem}

Note that when $m, \sigma, \ep, \delta$ are constant, the query bound obtained by \cref{thm:query-equilibrium} is also a constant. 
This is in constrast to the query complexity of $\ep$-approximate Nash equilibria, even for two players, which requires at least $n^{2 - o(1) }$ queries even for constant $\epsilon$ and constant probability of error $\delta$ \cite{fearnley2014learning,DBLP:journals/teco/FearnleyS16,DBLP:conf/focs/GoosR18}. 

The main idea behind the algorithm \QueryEquilibrium can be summarized in two steps.
The first is the following strengthening of the existence of $k$-uniform approximate smooth Nash equilibria from \cref{lem:existence_k_unif}.  In particular, \cref{lem:unif-sampling} below shows that not only do $k$-uniform approximate smooth Nash equilibria exist for small $k$, but also that they can found in small subsets that are uniformly sampled from the set of actions.

\begin{lemma}
  \label{lem:unif-sampling}
  Let $A_1, \ldots, A_m : [n]^m \ra [0,1]$ denote the payoff matrices of an $m$-player normal-form game $G$. Let $\ep, \sigma \in (0,1)$ be fixed, and  $t =  \frac{C_{\ref{lem:mplayer-approx}}m \log(16 m/ \delta \sigma)}{\ep^2}, \ell = \frac{\log(2tm / \delta )  }{\sigma} $ and set $k=t\ell$. 
  For $j \in[m]$, %
   let $\hat X^j\in \Delta^n$ denote the uniform measure over $k$ uniformly random (with replacement) elements selected from $[n]$. 
  Then, with probability $1-  \delta $, there is a $t$-uniform $\ep$-approximate $\sigma$-smooth Nash equilibrium suported on $ \supp(\hat X^1) \times \dots \supp(\hat X^m)$.
\end{lemma}
Given \cref{lem:mplayer-approx}, the main ingredient in the proof of \cref{lem:unif-sampling} is the coupling lemma from \cite{DBLP:conf/focs/HaghtalabRS21} (\cref{lem:coupling}), which implies that, for any $k \in \BN, \sigma \in (0,1)$ and any $\sigma$-smooth distribution $p$, the following holds: there is a coupling between an i.i.d.~sample $S$ of size $k$ from $p$ and an i.i.d.~sample $S'$ of size roughly $k/\sigma$ from the uniform distribution, so that with high probability under the coupling we have $S \subset S'$. To prove \cref{lem:unif-sampling}, this coupling lemma is applied for each player $j$, with the sample $S$ corresponding to a sample from a $\sigma$-smooth Nash equilibrium as in \cref{lem:mplayer-approx}, and the sample $S'$ corresponding to $\hat X^j$. The full proof is in \cref{sec:use-coupling}.

  \begin{algorithm}[t]
  \caption{$\OptimizeDeviation(\MR, A, j, x, \sigma)$: compute optimal smooth deviation using few queries}
  \label{alg:optimize-deviation}
  \begin{algorithmic}[1]\onehalfspacing
    \State Write $\MR = \{ r_{j,1}, \ldots, r_{j,N} \}$, for $r_{j,1}, \ldots, r_{j,N} \in [n]$. 
    \For{$i \in [N]$}
    \State Set $\hat v_i = A(r_{j,i}, x_{-j})$.
    \EndFor
    \State \multiline{Let $\tau_1, \ldots, \tau_N \in [N]$ denote a permutation of $[N]$ so that $\hat v_{\tau_1} \geq \cdots \geq \hat v_{\tau_N}$.}
    \State \Return $\hat v := \frac{1}{\sigma N} \sum_{k=1}^{\sigma N} \hat v_{\tau_k}$.
  \end{algorithmic}
\end{algorithm}

Given \cref{lem:unif-sampling}, \QueryEquilibrium proceeds in the natural way. 
 It samples random $k$-uniform measures $\hat X^j$ as in \cref{lem:unif-sampling} (\cref{line:draw-hatx}), and iterates over all $t$-uniform strategy profiles $\hat x$ supported on them (\cref{line:s-hatx}).
The key challenge, however, is to test whether each such strategy profile $\hat x$ is in fact an $\ep$-approximate $\sigma$-smooth equilibrium, in a query-efficient way.
 In particular, for each player $j$, we must estimate the value of its best $\sigma$-smooth response to $\hat x_{-j}$.
The naive way to compute this best response is too query-expensive, requiring more than $\Omega(n)$ queries: it would, for each player $j \in [m]$, iterate over all $n$ actions of player $j$, and average the best $\sigma n$. %
Instead, we settle for an approximation to the best $\sigma$-smooth deviation, computed as follows: we sample a subset $\MR_j$ of sufficiently large size $N $ (which is nevertheless \emph{independent} of $n$) in \cref{line:sample-rj} uniformly at random. 
Then, we average estimates of player $j$'s values for the best $\sigma N$ sampled actions in $\MR_j$ (\cref{line:call-optimize-deviation}, which calls \OptimizeDeviation, \cref{alg:optimize-deviation}). This subset $\MR_j$ is reused amongst all $t$-uniform strategy profiles $\hat x$ that are considered in \QueryEquilibrium; reusing the samples $\MR_j$ as such allows us to avoid having the query complexity scale exponentially in $\log(1/(\delta \sigma))/\ep^2$. If no player $j$ can improve its utility by $\Omega(\ep)$, then the algorithm returns the strategy profile $\hat x$ (\cref{line:test-eq-query}).

  \begin{algorithm}[t]
  \caption{$\QueryEquilibrium((A_1, \ldots, A_m), \sigma, \ep, \delta)$: compute weak smooth Nash equilibrium using few queries}
  \label{alg:query-equilibrium}
  \begin{algorithmic}[1]\onehalfspacing
    \State Set  $t =  \frac{C_{\ref{lem:mplayer-approx}}m \log(32 m/ \delta \sigma)}{(\ep/4)^2}, \ell = \frac{\log(4tm / \delta )  }{\sigma} $, and $k = t\ell$.\label{line:set-tl}
    \State %
    Set $N:= \frac{16 C_{\ref{lem:highprob-good}} \cdot tm \log (k/\delta)}{\ep^2 \sigma^2}$. %
    \label{line:set-K}
    \State For $j \in [m]$, initialize $\hat A_j, \tilde A_j : [n]^m \ra [0,1]$ arbitrarily.
    \For{$j \in [m]$}
    \State \multiline{Let $\hat X_j \in \Delta^n$ denote the uniform measure over $k$  elements of $[n]$, chosen uniformly at random with replacement.}\label{line:draw-hatx}
    \State Let $\MS_j := \supp(\hat X_j) \subset [n]$.
    \EndFor
    \For{$j \in [m]$}
    \State Let $\MR_j$ be a set consisting of $N$ uniformly random elements of $[n]$, chosen with replacement.\label{line:sample-rj}
    \EndFor 
    \For{$j \in [m]$ and each $(b_j, b_{-j}) \in \MR_j \times \prod_{j' \neq j} \MS_{j'}$}
    \State Query $A_j(b_j, b_{-j})$ and let $\hat A_j(b_j, b_{-j})$ be the result of the query.\label{line:set-ahat}
    \EndFor
    \For{$j \in [m]$ and $b \in \prod_{j' \in [m]} \MS_{j'}$}
    \State Query $A_j(b)$ and let $\tilde A_j(b)$ be the result of the query. \label{line:set-atilde}
    \EndFor 
    \For{each $t$-uniform strategy profile $\hat x = (\hat x_1, \ldots, \hat x_m)$ supported on $\MS_1 \times \cdots \times \MS_m$}\label{line:s-hatx}
    \For{$j \in [m]$}
    \State Set $\hat v_j \gets \OptimizeDeviation(\MR_j, \hat A_j, \hat x, \sigma)$.\Comment{\emph{\cref{alg:optimize-deviation}}}\label{line:call-optimize-deviation}
    \State Set $\tilde v_j \gets \tilde A_j(\hat x)$.
       \EndFor
       \If{$\max_{j \in [m]}\{ \hat v_j - \tilde v_j\} \leq \ep/2$}\label{line:test-eq-query}
       \State \Return the strategy profile $\hat x$. 
       \EndIf
       \EndFor
       \State \Return an arbitrary strategy profile. \Comment{\emph{We will show this happens w.p.~$\leq \delta$.}}\label{line:arbitrary-sp}
  \end{algorithmic}
\end{algorithm}

\paragraph{Strong smooth Nash equilibria.} It is straightforward to see that an analogue of \cref{thm:query-equilibrium} cannot hold for strong $\sigma$-smooth Nash equilibria: even in the case $m=1$, there is a constant $c > 0$ so that $cn$ queries are needed to find a strong $c$-approximate $1/4$-smooth Nash equilibrium.%
\footnote{In particular, a 1-player game is described by a vector $A_1 \in \BR^n$: suppose that $A_1$ is chosen randomly with each entry drawn independently from $\mathrm{Ber(1/2)}$. Consider any randomized algorithm which makes at most $n/8$ queries to $A_1$ and outputs a distribution $\hat x \in \MK_{1/4, n}$. Conditioned on the set $\MQ \subset [n]$ of entries queried by the algorithm, all values of $(A_1)_i$, for $i \not \in \MQ$, are uniform and independent bits. Thus, conditioned on any set $\MQ$ with $|\MQ| \leq n/8$, with at least constant probability, a constant fraction of the mass of $\hat x$ is on entries of $A_1$ which are 0, and thus the suboptimality of $\hat x$ with respect to the best $1/4$-smooth deviation (which has value 1 with high probability) is $\Omega(1)$.} %
That said, sampling variants of a constant query strong Nash equilibria algorithm are not ruled out and we discuss some these questions in \cref{sec:discussion}.

\begin{remark}[Query Complexity Lower Bounds] \label{lem:query-lb} 
  As mentioned earlier, the query and communication complexity of finding Nash equilibria in games is well-studied. 
In particular, it is known that finding an $\ep_0$ approximate Nash equilibrium, for some constant $\ep_0$, in a 2-player game requires $\Omega(n^{2-o(1)})$ queries \cite{fearnley2014learning,DBLP:journals/teco/FearnleyS16,DBLP:conf/focs/GoosR18}. 
Further, for $m$-player games with two actions per player the query complexity of finding an $\ep_0$-approximate Nash equilibrium is known to be $  2^{ \Omega(m) } $ (see \cite{DBLP:journals/geb/BabichenkoR22} and references therein). 
This implies that for constant $\sigma$ and $\epsilon$, the result of  \cref{thm:query-equilibrium} is tight up to logarithmic factors in the exponent.
Getting tighter lower bounds that fully elucidate the required dependence on $\sigma, \epsilon$ is an interesting avenue for future work. 
\end{remark}

\section{Efficient Algorithms for Finding Smooth Nash Equilibria}
\label{sec:eff-algs}
In this section, we introduce algorithms to compute $\ep$-approximate $\sigma$-smooth Nash equilibria. First, in \cref{sec:weak-ub}, we show that, as a relatively straightforward consequence of the results of \cref{sec:sampling,sec:query}, weak smooth Nash equilibria can be efficiently computed when the approximation and smoothness parameters are constants.  Then, we show that a similar conclusion also applies to strong smooth Nash equilibria in \cref{sec:strong-ub}, though the proof requires some new ideas. %

\subsection{Finding Weak Smooth Equilibria in Games}
\label{sec:weak-ub}
Recall that, in \cref{lem:existence_k_unif}, we showed that, in any normal-form $m$-player game, there exists an $\ep$-approximate $\sigma$-smooth Nash equilibrium which is $k$-uniform (per \cref{def:k-uniform}) for $k = O \left( \frac{m \log (m/\sigma)}{\ep^2}\right)$, which is \emph{constant} when $\ep, \sigma, m$ are constants. Since the number of $k$-uniform strategies of any player can be enumerated in time $n^k$, and since it can be efficiently tested whether a $k$-uniform strategy profile is an $\ep$-approximate $\sigma$-smooth Nash equilibrium, it follows that we can compute such an equilibrium in time $n^{mk}$, which is $\poly(n)$ when $k,m,\sigma$ are constants. 
\begin{theorem}[Polynomial-time algorithm for weak smooth Nash]
  \label{thm:eff-weak}
Let $A_1, \ldots, A_m : [n]^m \ra [0,1]$ denote the payoff matrices of an $m$-player normal-form game $G$. Then, for any $\sigma, \ep \in (0,1)$, there is an algorithm running in time $n^{O \left(\frac{m^2 \log(m/\sigma)}{\ep^2}\right)}$ which finds a weak $\ep$-approximate $\sigma$-smooth Nash equilibrium.
\end{theorem}

The algorithm and proof for \cref{thm:eff-weak} (presented formally in \cref{sec:eff-weak-proof}) are analogous to the well-known result of \cite{lmn}, where a \emph{quasipolynomial}-time algorithm was established for computing $\ep$-approximate Nash equilibrium in $m$-player games, when $\ep, m$ are constants. \cref{thm:eff-weak} improves this result to \emph{polynomial-time} for $\Omega(1)$-smooth equilibria since the parameter $k$ from \cref{lem:existence_k_unif} as discussed above is constant when $\sigma = \Omega(1)$, whereas %
for (non-smooth) equilibria, \cite{lmn} requires $k = O(\log n)$.

\subsubsection{Sublinear-time algorithms} \label{sec:sublinear}

Is it possible to beat polynomial-time, i.e., obtain \emph{sublinear-time} algorithms for approximating $\sigma$-smooth equilibria? It is straightforward to see that if we restrict our attention to deterministic algorithms, this is not possible, %
even for $\ep = 1/2, \sigma = 1/4, m = 1$. To see this, note that in the case $m=1$, the game is described simply by a vector $A_1 \in \BR^n$.  For any deterministic algorithm running in time at most $n/2$, it must make at most $n/2$ queries to $A_1$. Then there is a fixed sequence $i_1, i_2, \ldots, i_{n/2}$ of the $n/2$ indices at which it queries $A_1$ when all queries return 0. No matter which fixed distribution in $\Delta^n$ the algorithm outputs, there is a subset $\MS \subset [n]\backslash \{ i_1, \ldots, i_{n/2} \}$ of size $|\MS| = \sigma n = n/4$, so that, setting $(A_1)_{i} = 1$ for $i \in \MS$ (and $(A_1)_i = 0$ otherwise) ensures that the algorithm's output distribution yields utility at most $1/2$, whereas deviating to play uniformly on $\MS$ yields utility of 1.

One may nevertheless wonder about randomized algorithms; surprisingly, the answer turns out to be vastly different. In particular, it follows from the techniques used to prove our query complexity upper bounds in \cref{sec:query} that, in the setting where $m, \sigma, \ep$ are constants, there is a randomized \emph{constant time}\footnote{Note that here we work with the  $O(\log(n/\epsilon) )$-word RAM model. This model allows the algorithm to specify the indices of actions and to access entries of the payoffs in constant time. Note that the payoffs are only needed up to precision $\poly(\ep)$ if our goal is to find $\ep$-approximate smooth equilibria.} algorithm that computes $\ep$-approximate $\sigma$-smooth equilibria in $m$-player games, with arbitrarily small constant failure probability. 
 
\begin{theorem}[Constant-time randomized algorithm for weak smooth Nash]
  \label{thm:weak-nash-constant}
Let $A_1, \ldots, A_m : [n]^m \ra [0,1]$ denote the payoff matrices of an $m$-player normal-form game $G$. Then for any $\sigma, \ep, \delta \in (0,1)$, there is an algorithm running in time $ \left( \frac{m \log(1/\delta)}{\sigma \ep} \right)^{O\left( \frac{m^2 \log(m/\delta \sigma)}{\ep^2} \right)}$ which outputs a strategy profile which is a weak $\ep$-approximate $\sigma$-smooth equilibrium with probability at least $1-\delta$. 
\end{theorem}
The proof of \cref{thm:weak-nash-constant} (presented in \cref{sec:wn-constant-proof}) follows from \cref{thm:query-equilibrium} by bounding the running time of \QueryEquilibrium.

\subsection{Finding strong smooth Nash equilibria}
\label{sec:strong-ub}

The strategy profiles in the equilibria found by the search procedures used to establish \cref{thm:eff-weak,thm:weak-nash-constant} are $k$-sparse with $k=\poly(m, \sigma^{-1}, \ep^{-1})$. For small values of $m, \ep, \sigma$ (which is the main regime of interest), these strategies are certainly not $\sigma$-smooth,  and thus the equilibria found are not \emph{strong} smooth equilibria. In order to find strong equilibria, we need to implement an additional ``smoothening'' step. As we shall see, we can do so using linear programming, leading to the following result. 
\begin{theorem}[Polynomial-time algorithm for strong smooth Nash, multi-player]
  \label{thm:alg-strong}
Let $A_1, \ldots, A_m : [n]^m \ra [0,1]$ denote the payoff matrices of an $m$-player normal-form game $G$. Then for any $\sigma > 0$, there is an algorithm running in time $ n^{O \left( \frac{m^4 \log (m/\sigma)}{\ep^2} \right)}$  which finds a \textbf{strong} $\ep$-approximate $\sigma$-smooth Nash equilibrium.
\end{theorem}
We first describe the proof of \cref{thm:alg-strong} in the case of $m=2$ players, for which the claimed algorithm is given by \BimatrixStrongSmooth (\cref{alg:bimatrix-strong}). In this case, for ease of notation, we denote the payoff matrices by $A, B \in [0,1]^{n \times n}$ and the two players' strategies by $x,y \in \Delta^n$. 
Set $\ep_0 = \ep/4$, and $k = \frac{2C_{\ref{lem:mplayer-approx}} \log(2/\sigma)}{\ep_0^2}$, where $C_{\ref{lem:mplayer-approx}}$ is the constant of \cref{lem:mplayer-approx}. \BimatrixStrongSmooth proceeds as follows. 
 It iterates over all $k$-uniform strategy profiles $(\hat x, \hat y)$. For each one which is a weak $\ep_0$-approximate $\sigma$-smooth Nash equilibrium, it solves a linear program (namely, \cref{eq:2player-program}) which aims to find a $\sigma$-smooth strategy profile $(x,y)$ which approximates $(\hat x, \hat y)$ in terms of each player's utility with respect to all $\sigma$-smooth deviations of the other player. If such a program is feasible, realized by $(x,y)$, then the program returns such $(x,y)$. Intuitively, the constraints of this program ensure that the fact that the smooth Nash equilibrium constraints are satisfied for $(\hat x, \hat y)$ implies that the smooth Nash equilibrium constraints are satisfied for $(x,y)$.

To complete the proof of the theorem (for $m=2$), we need to establish two facts. 
First, that for some $k$-uniform strategy profile $(\hat x, \hat y)$ which is a weak smooth Nash equilibrium, the program \cref{eq:2player-program} will be feasible. 
 Second, that any feasible solution of \cref{eq:2player-program} (given that $(\hat x, \hat y)$ is a weak smooth Nash equilibrium) is in fact a strong smooth Nash equilibrium. %
The proof of the first fact uses \cref{lem:mplayer-approx} with input a \emph{strong} $\sigma$-smooth Nash equilibrium $(x,y)$ (which exists by \cref{thm:existence}). 
The conclusion \cref{eq:s-unif-close} of \cref{lem:mplayer-approx} can be used to show that $(x,y)$ is a feasible solution to \cref{eq:2player-program} for an appropriate choice of $(\hat x, \hat y)$. The proof of the second fact follows by using the constraints of \cref{eq:2player-program} to derive that any feasible solution must be an approximate strong smooth Nash equilibrium. Finally, we need to ensure that \BimatrixStrongSmooth can be implemented in the claimed time. The most nontrivial part of this claim is ensuring that the ellipsoid algorithm can efficiently solve \cref{eq:2player-program}, in light of the fact that \cref{eq:2player-program} has exponentially many constraints. We discuss this issue in detail in the appendix. 
The full proof for the 2-player case can be found in \cref{sec:strong-2player-proof}.

\begin{algorithm}[t]
  \caption{$\BimatrixStrongSmooth((A,B), n, \sigma, \ep)$: compute strong smooth equilibria of 2-player games}
  \label{alg:bimatrix-strong}
  \begin{algorithmic}[1]\onehalfspacing
    \State Set $\ep_0 = \ep/4$ and $k = \frac{2C_{\ref{lem:mplayer-approx}} \log(2/\sigma)}{\ep_0^2}$.
    \For{Each $k$-uniform strategy profile $(\hat x, \hat y)$}
    \If{$(\hat x, \hat y)$ is a weak $\ep_0$-approximate $\sigma$-smooth Nash equilibrium}
    \State Solve the following feasilibity linear program for $x,y \in \BR^n$, using the ellipsoid algorithm:%
      \begin{subequations}
        \label{eq:2player-program}
      \begin{align}
        \mbox{ Find } x,y \in \BR^n:\quad 
        & x,y \in \MK_{\sigma, n} \label{eq:2pp-0}\\
        & |x^\t A \hat y - \hat x^\t A \hat y| \leq \ep_0 \quad \label{eq:2pp-1}\\
        & |(x')^\t Ay - (x')^\t A \hat y | \leq \ep_0 \quad \forall  x' \in \MK_{\sigma, n}\label{eq:2pp-4}  \\      
        & |x^\t By' - \hat x^\t By'| \leq \ep_0 \quad \forall y' \in \MK_{\sigma, n}\label{eq:2pp-2}\\
        & |\hat x^\t By - \hat x^\t B \hat y| \leq \ep_0\label{eq:2pp-3}.
      \end{align}
      \end{subequations}
      \If{\SolveLPFeasibility outputs that above program is feasible, realized by $(x,y)$}
      \State \Return $(x,y)$. %
      \EndIf
      \EndIf
      \EndFor
  \end{algorithmic}
\end{algorithm}

\paragraph{Multiplayer games.} 
The main challenge in extending the above arguments to the case of $m$-player games, for general $m > 2$, is generalizing the program \cref{eq:2player-program} to the $m$-player case. At first glance this may seem problematic if it turns out to be necessary to have, say, constraints of the form $|A_j(x_j', x_{-j}) - A_j(x_j', \hat x_{-j})| \leq \ep_0$ for all $x_j' \in \MK_{\sigma, n}$. Such a constraint is not linear in the program variables $x = (x_1, \ldots, x_m)$ when $m > 2$, since $A_j(x_j', x_{-j})$ is a polynomial of degree $(m-1)$ in $x$. 
Fortunately, such constraints are avoidable. 
 We will use a hybrid argument to show, over the course of $m$ steps, that satisfiability of a certain \emph{linear} program yields a strong smooth equilibrium from a weak smooth equilibrium $\hat x$. 
 To apply this hybrid argument, we will need a stronger 
version of \cref{lem:mplayer-approx}, stated below as \cref{lem:mplayer-approx-strong}. 
In the lemma statement, we generalize the notation $\MK_{\sigma, n}$.
For any finite set $\MS$ and $\sigma > 0$, we let $\MK_{\sigma, \MS}$ denote the set of distributions $P \in \Delta(\MS)$ so that $P(s) \leq \frac{1}{\sigma |\MS|}$ for all $s \in \MS$. In particular, below we  have $\MS = [n]^\ell$ for $\ell \in \BN$. We denote elements of $\MK_{\sigma, [n]^\ell}$ by $x_{1:\ell}'$; note that such $x_{1:\ell}'$ is in general \emph{not} a product distribution; neverthelss, we will slightly abuse notation by writing, for a fixed sequence $x_{\ell+1:m} = (x_{\ell+1}, \ldots, x_m) \in (\Delta^n)^{m-\ell}$, 
\begin{align}
A_j(x_{1:\ell}', x_{\ell+1:m}) := \E_{(b_1, \ldots, b_\ell) \sim x_{1:\ell}'} \E_{b_i \sim x_i \ \forall i \geq \ell+1} [A_j(b_1, \ldots, b_m)]\nonumber.
\end{align}
Similarly, if $p \geq \ell+1$, we write $A_j(x_{1:\ell,p}', x_{\ell+1:m,-p})$ to denote the corresponding expectation where the $p$th coordinate is included in the distribution $x_{1:\ell,p}' \in \Delta([n]^{\ell+1})$ and excluded from $x_{\ell+1:m}$. 
\begin{lemma}
  \label{lem:mplayer-approx-strong}
  Let $A_1, \ldots, A_m : [n]^m \ra [0,1]$ denote the payoff matrices of an $m$-player normal-form game $G$. Let $\ep,\sigma \in (0,1)$ be fixed. Set $k = \frac{C_{\ref{lem:mplayer-approx-strong}}m \log(m/\sigma)}{\ep^2}$, for a sufficiently large constant $C_{\ref{lem:mplayer-approx-strong}}$. Given a strong 0-approximate $\sigma$-smooth Nash equilibrium $x = (x_1, \ldots, x_m)$, there is a strategy profile $\hat x = (\hat x_1, \ldots, \hat x_m)$ satisfying the following properties:
  \begin{enumerate}
  \item \label{it:kuniform} All entries of each $\hat x_j$, $j \in [m]$, are $k$-uniform;
  \item \label{it:weaknash} $\hat x$ is a weak $\ep$-approximate $\sigma$-smooth Nash equilibrium.
    \item \label{it:approx-strong}The following inequalities hold for each $j \in [m]$:
  \begin{align}
    \max_{x_{1:\ell-1}' \in \MK_{\sigma^{\ell-1}, [n]^{\ell-1}}} \left| A_j(x_{1:\ell-1}', \hat x_{\ell}, \hat x_{\ell+1:m}) - A_j(x_{1:\ell-1}', x_\ell, \hat x_{\ell+1:m}) \right| \leq \ep \qquad \forall \ell \in [m]\label{eq:all-l-approx}\\
       \max_{x_{1:\ell-1,j}' \in \MK_{\sigma^{\ell}, [n]^{\ell}}} \left| A_j(x_{1:\ell-1,j}', \hat x_{\ell}, \hat x_{\ell+1:m,-j}) - A_j(x_{1:\ell-1,j}', x_\ell, \hat x_{\ell+1:m,-j}) \right| \leq \ep \qquad \forall \ell \in [j-1]\label{eq:one-more-approx}.
  \end{align}
\end{enumerate}
\end{lemma}
The inequalities \cref{eq:all-l-approx,eq:one-more-approx} generalize \cref{eq:s-unif-close} in that a maximum is taken over all $\sigma^\ell$-smooth distributions on $[n]^\ell$, for various $\ell \in [m]$. To prove \cref{thm:alg-strong} given \cref{lem:mplayer-approx-strong}, we use \GeneralStrongSmooth (\cref{alg:general-strong}), which is similar to \BimatrixStrongSmooth, except that the program \cref{eq:2player-program} is replaced by \cref{eq:multiplayer-program}, whose constraints mirror those in \cref{eq:all-l-approx,eq:one-more-approx}. The proof uses these constraints together with a hybrid argument to show that a solution of \cref{eq:multiplayer-program} is an approximate smooth Nash equilibrium. %
The full proofs of \cref{lem:mplayer-approx-strong,thm:alg-strong} are in \cref{sec:proof-mplayer-approx-strong}. %

\section{Hardness Results for Smooth Equilibria}
\label{sec:hardness}
In the previous sections, we showed that computing $\ep$-approximate $\sigma$-smooth Nash equilibria is tractable when both $\ep, \sigma$ are constants. In this section, we show that this result is optimal in the sense that when either $\ep$ or $\sigma$ is an inverse polynomial, then an efficient algorithm is ruled out by standard hardness assumptions for the complexity class \PPAD. %
\PPAD is a subclass of the class \TFNP of total search problems consisting of those problems in \TFNP which have a polynomial-time reduction to the \texttt{End-of-the-Line} problem (\cref{def:ppad}).\footnote{Recall that \emph{total} search problems are those for which a solution exists given any input.} %
\PPAD-hardness of a problem is typically taken as strong evidence that a problem is intractable; this perspective is supported by the fact that  under cryptographic assumptions, \PPAD does not have polynomial-time algorithms (see \cite{DBLP:conf/tcc/BitanskyCHKLPR22,DBLP:conf/crypto/GargPS16} and references therein).  We defer a formal discussion of \PPAD to \cref{sec:ppad}.

Throughout the section, we restrict to 2-player, $n$-action games, and show \PPAD hardness for various regimes of $\ep, \sigma$. (Analogous lower bounds for games with a constant number $m$ of players immediately follow, being a generalization of the former.)

 \subsection{Polynomially small smoothness parameter}
 \label{sec:harness-polynomial}
First, we establish lower bounds for the setting that $\sigma$ is an inverse polynomial in $n$ and $\ep$ is a constant. The main ingredient for such lower bounds is the following lemma, which gives a polynomial-time reduction between the problems of finding $\ep$-approximate $\sigma$-smooth Nash equilibria in 2-player games for differing values of $\ep, \sigma$ whose dependence on $n$ differs by a polynomial.
\begin{lemma}
  \label{thm:smooth-nash-nash}
  Let $\ep : \BN \ra (0,1), \sigma : \BN \ra (0,1)$ be non-increasing real-valued functions of natural numbers. Then for any $c \in (0,1)$, the problem of finding weak $\ep(n)$-approximate $\sigma(n)$-smooth Nash equilibrium in 2-player, $n$-action games has a polynomial-time reduction (in the sense of \cref{def:poly-reducible}) to the problem of finding weak $\ep(n^c)$-approximate $\sigma(n^c)$-smooth Nash equilibrium in 2-player, $n$-action games.

  Moreover, the same conclusion holds if ``weak'' is replaced by ``strong''. 
\end{lemma}
The proof of \cref{thm:smooth-nash-nash} proceeds via a padding argument: given an $n$-action normal-form game $G$, it constructs a game $G'$ with $N := n^{1/c}$ actions which consists of $n^{(1-c)/c}$ copies of $G$. Given a $\sigma(N)$-smooth $\ep(N)$-approximate Nash equilibrium of $G'$, the proof shows how to construct (roughly speaking) a $\sigma(2N)$-smooth $\ep(N)$-approximate Nash equilibrium of $G$. Recalling that $N = n^{1/c}$, we get the result after reparametrizing the functions $\ep(\cdot), \sigma(\cdot)$.

\paragraph{Lower bounds for constant $\ep$ under ETH for \PPAD.} %
We observe that $\MK_{1/n,n} = \Delta^n$; thus, comparing \cref{def:sne,def:ne}, we see that weak (and strong) $\ep$-approximate $1/n$-smooth Nash equilibria are equivalent to $\ep$-approximate Nash equilibria. 
Moreover, we remark that under the exponential time hypothesis for \PPAD (\cref{conj:ppad-eth}), it is known \cite{DBLP:journals/sigecom/Rubinstein17} that, for some constant $\ep_0$, there is no algorithm that computes $\ep_0$-approximate Nash equilibria in 2-player $n$-action games in time $n^{\log^{1-\delta} n}$, for any constant $\delta > 0$ (\cref{thm:quasipoly-hardness}). \cref{thm:smooth-nash-nash} tells us that computing $\ep_0$-approximate Nash equilibria (or, equivalently, $\ep_0$-approximate $1/n$-smooth Nash equilibria) is polynomial-time reducible to computing $\ep_0$-approximate $1/n^c$-smooth Nash equilibria, for any constant $c \in (0,1)$. Thus, we obtain the following as a corollary:
\begin{corollary}
  \label{cor:smooth-nash-eth}
For some $\ep_0 \in (0,1)$, the following holds assuming the ETH for \PPAD: for any $\delta, c \in (0,1)$ there is no algorithm that computes weak $\ep_0$-approximate $n^{-c}$-smooth Nash equilibrium in 2-player $n$-action games in time $n^{\log^{1-\delta}n}$. 
\end{corollary}

\subsection{Constant smoothness parameter} 
\label{sec:harness-constant} 
Next, we establish lower bounds for the setting that $\sigma$ is a constant and $\ep$ is an inverse polynomial in $n$. 
One might hope that the same technique used to establish \cref{cor:smooth-nash-eth} might apply in this setting, by using \cref{thm:smooth-nash-nash} for $c = 0$. However, this approach fails since in the limit $c \ra 0$, the size of the game $G'$ used to prove \cref{thm:smooth-nash-nash}, which is $\Theta(n^{1/c})$, diverges. Indeed, to establish hardness of finding weak $\sigma_0$-approximate $n^{-c}$-smooth Nash equilibrium in 2-player, $n$-action games, for \emph{constant} $\sigma_0$ (stated formally in \cref{thm:ppad-hard-const-sig} below), we require additional techniques.

\paragraph{Generalized matching pennies game.} Our approach proceeds by showing that it is hard to find approximate smooth Nash equilibria in a certain class of games known as approximate generalized matching pennies games. To do so, we first show that all approximate smooth Nash equilibria in generalized matching pennies games have a certain structure (formalized in \cref{lem:validity}) that allows us to relate them to approximate Nash equilibria in such games. We will then apply a result of \cite{DBLP:journals/jacm/ChenDT09} which shows that it is \PPAD-hard to compute approximate Nash equilibria in this subclass of games.

To define generalized matching pennies games, fix an integer $K \in \BN$ and write $N = 2K$. Let $M := {2K^3}$. Let $A^\st \in \BR^{N \times N}$ denote the $K \times K$ block matrix where the $2 \times 2$ blocks along the main diagonal have all entries equal to $M$, and all other entries are 0. (Note that we have $A^\st = M \cdot I_K \otimes J_{2,2}$, where $J_{2,2}$ denotes the $2 \times 2$ matrix of 1s and $\otimes$ denotes the Kronecker product.) Let $B^\st = -A^\st$. 
\begin{definition}
Given $K \in \BN$ and $N := 2K$, a 2-player $N$-action game $(A,B)$ is defined to be an \emph{approximate $K$-generalized matching pennies ($K$-GMP) game} if all entries of $A- A^\st$ and of $B-B^\st$ are in $[0,1]$, where $A^\st, B^\st \in \BR^{N \times N}$ are defined in terms of $K$ above.
\end{definition}
For ease of notation, we have departed from our convention that all payoffs of the game are in $[0,1]$; this setting is equivalent to the former by rescaling. Moreover, we have suppressed the dependence of $A^\st, B^\st$ on $K$ in our notation.

Let $(A,B)$ denote an approximate $K$-GMP game and let $x,y \in \Delta^N$ be strategy vectors. We use the following convention: $\bar x, \bar y \in \Delta^K$ denote  vectors defined by $\bar x_k = x_{2k-1} + x_{2k}$ and $\bar y_k = y_{2k-1} + y_{2k}$, for $k \in [K]$. 
\begin{lemma}%
  \label{lem:validity}
  Suppose $K, M \in \BN,\ \sigma \in (0,1)$, and $\ep \geq \frac{8\sigma K^2}{M}$ are given. If $G = (A,B)$ is a $K$ approximate $K$-GMP game, then %
  a weak 1-approximate $\sigma$-smooth Nash equilibrium $(x,y)$ must satisfy the following for all $k \in [K]$:
  \begin{align}
\bar x_k = 1/K \pm \ep, \qquad \bar y_k = 1/K \pm \ep\nonumber.
  \end{align}
\end{lemma}
As a result of \cref{lem:validity}, any $\ep$-approximate $\sigma$-smooth Nash equilibrium $(x,y)$ in an approximate $K$-GMP game must in fact be an $\ep K$-approximate Nash equilibrium. The intuition behind this implication is as follows: if some player (say the $x$-player) had a useful deviation from their strategy $x$ to any fixed action $i$, they could instead deviate to the strategy whereby they play $i$ with probability $1/K$ and $x$ with the remaining probability. Since $\max_i x_i \leq \max_k \bar x_k \leq 1/K + \ep$ (by \cref{lem:validity}), as long as $2/K + \ep \leq \frac{1}{\sigma N} = \frac{1}{2 \sigma K}$ (which will hold if $\sigma \leq 1/6$), the resulting strategy is $\sigma$-smooth. Moreover, the player would gain a $1/K$ fraction of their utility gain for deviating to $i$. By combining this observation and the fact that computing approximate Nash equilibria in $K$-GMP games is \PPAD-hard (\cref{thm:gmp-hardness}, from \cite{DBLP:journals/jacm/ChenDT09}), we can show the following theorem. 
The full proofs of all results from this section are in \cref{sec:const-sig-lb-proofs}. 

\begin{theorem}[\PPAD-hardness for constant $\sigma$]
  \label{thm:ppad-hard-const-sig}
For any constant $c_1 \in (0,1)$, the problem of computing weak $n^{c_1}$-approximate $1/6$-smooth Nash equilibria in 2-player $n$-action games is \PPAD-hard.
\end{theorem}

\section{Discussions and Open Problems}
\label{sec:discussion}

In this paper we introduced the notion of smooth Nash equilibria and showed that they satisfy many desirable computational properties, namely that they have polynomial time and query algorithms. Our results open numerous avenues for future work, listed below.
\begin{itemize}
    \item \textbf{Improving Dependence on the number of players:} Note that the dependence on $m$ in the exponent of the running time of \cref{thm:alg-strong}, namely $\tilde O(m^4)$, is worse than that in the corresponding result for finding \emph{weak} smooth Nash equilibrium (\cref{thm:eff-weak}), where the dependence is $\tilde O(m^2)$. It is an interesting question if this dependence can be improved. 
    In fact, we expect even the $\tilde O(m^2)$ dependence can be improved to $ O(m \log m ) $ using the techniques of \cite{DBLP:journals/mor/BabichenkoBP17}.  
    \item \textbf{Sampling from strong smooth equilibrium:} In \cref{sec:query}, we gave a constant time algorithm for computing a weak $\sigma$-smooth Nash equilibrium and argued that an analogous algorithm for strong $\sigma$-smooth Nash equilibrium could not run in sublinear time. One could ask if, instead of outputting a strong $\sigma$-smooth Nash equilibrium, we can output a sample from a strong $\sigma$-smooth Nash equilibrium in sublinear time.
    \item \textbf{General Notions of Smoothness:}
    In this paper, we focus on smoothness defined using the $L_\infty$ norm of the Radon-Nikodym derivative (see \cref{lem:generalize-smooth-polytope}). 
    However, one could consider other notions of smoothness corresponding to other notions of distance from a fixed measure, such as $\chi^2$ divergence or KL divergence. 
    We expect that our techniques can be extended to general settings with the appropriate changes (for example, \ref{lem:coupling} would need to be generalized as in \cite{DBLP:conf/colt/BlockP23}). 
    Exploring this direction is an interesting avenue for future work.
    \item \textbf{General Complexity of Polyhedral and Concave Games:}
    One perspective on our result obtaining algorithms for smooth Nash equilibria in $2$-player games is through the lens of covering numbers  of matrices \cite{alon2014cover}. 
    That is, the proof of \cref{lem:mplayer-approx} via \ref{lem:smooth-class-lemma} can be viewed as bounding the appropriate covering numbers of the polytope induced by the game matrix acting on the set of smooth distributions.
    A natural question is whether this notion of covering number can be used to prove a general characterization of complexity of finding Nash equilibria in polyhedral games. In particular, can the lower bound from \cite{DBLP:journals/sigecom/Rubinstein17} be generalized to provide a characterization of the complexity of finding Nash equilibria in general settings? 
    \item \textbf{Sharp Lower Bounds:} A related question to the above is whether we can improve our lower bounds in \cref{sec:hardness} to match the upper bounds in the full regime of parameters (i.e., for all $\ep, \sigma \in (0,1)$ and all  $m \in \mathbb{N} $) both in the case of time complexity and query complexity.
    Understanding this question even for the case of $\sigma = 1/n$ and $m=2$ seems to be a challenging open problem. 
    \item \textbf{Extensions to Other Equilibrium Concepts:} Given that the notion of smooth Nash equilibria leads significant computational advantages relative to Nash equilibria, it is natural to ask if similar advantages can be obtained for other equilibrium concepts such as Arrow-Debreu market equilibria. 
    Developing a notion of smoothness in these settings which simultaneously has a natural and economically meaningful interpretation while also leading to computational advantages would be ideal.
    \item \textbf{Applications:} As discussed earlier, one can view our notion of smooth Nash equilibria as a strengthening of quantal response equilibria. 
    Given the extensive usage of this notion in modern applications of equilibrium solving to multi-agent reinforcement learning, it would be interesting to see if the notion of smooth Nash equilibria  can be applied in multiagent settings to obtain computational advantages.
    
    \end{itemize}

\section*{Acknowledgements}
This work was done in part while the authors were visiting the Learning and Games program at the Simons Institute for the Theory of Computing. 

\bibliographystyle{alpha}
\bibliography{refs.bib}

\appendix

\section{Proofs from \cref{sec:sampling}} \label{sec:proofs-sampling}

\subsection{Smooth Finite Class Lemma} \label{sec:smooth-class-lemma}

\begin{lemma}[Smooth finite class lemma]
  \label{lem:smooth-class-lemma}
Suppose $\{ U_i \}_{i \in [n]}$ is a finite collection of random variables which satisfy the following sub-Gaussianity condition: there is a constant $c > 0$ so that for all $\lambda > 0$ and all $i \in [n]$, $\E[\exp(\lambda U_i)] \leq e^{c^2 \lambda^2/2}$. Then
for any $\delta > 0$,
\begin{align}
\BP \left( \max_{w \in \MK_{\sigma,n}} \lng w, U \rng > c{\sqrt{2 \log 1/(\delta \sigma)}} \right) \leq \delta\nonumber.
\end{align}
\end{lemma}

\begin{proof}
  Using Jensen's inequality (twice) as well as the fact that $\exp(\cdot)$ is monotone, we see that, for any $\lambda > 0$,
  \begin{align}
  \exp \left( \lambda \E \max_{w \in \MK_{\sigma,N}} \langle w, U \rangle \right) \leq & \E \exp \left( \max_{w \in \MK_{\sigma,N}} \langle w, \lambda U \rangle \right) \nonumber\\
  =& \E \max_{w \in \MK_{\sigma,N}} \exp \left( \langle w, \lambda U \rangle \right)\nonumber\\
  \leq & \E \max_{w \in \MK_{\sigma,N}} \sum_{i=1}^N w_i \exp(\lambda U_i) \nonumber\\
  \leq & \E \sum_{i=1}^N \frac{1}{N \sigma} \exp(\lambda U_i)\nonumber\\
  \leq & \frac{1}{\sigma} \cdot \exp(c^2 \lambda^2/2)\nonumber,
  \end{align}
  where the final inequality uses the assumed sub-Gaussianity condition.
  
  By Markov's inequality, for any $\delta \in (0,1)$, we have that
  \begin{align}
  \BP \left( \exp \left( \max_{w \in \MK_{\sigma, N}} \lng w, \lambda U \rng \right)> \frac{1}{\delta \sigma}\cdot \exp(c^2 \lambda^2/2) \right) \leq \delta\nonumber.
  \end{align}
  Rearranging, we obtain that
  \begin{align}
  \BP \left( \max_{w \in \MK_{\sigma,N}} \lng w, U \rng > \frac{\log 1/(\delta \sigma)}{\lambda} + \frac{c^2 \lambda}{2}\right) \leq \delta\nonumber.
  \end{align}
  Choosing $\lambda = \sqrt{2 \log 1/(\delta \sigma)}/c$ establishes the  statement.
\end{proof}

\subsection{Proof of \cref{lem:mplayer-approx}}

We recall the standard bounded differences inequality. 

\begin{lemma}[Bounded differences inequality]
  \label{lem:bd}
  Fix $n \in \BN$, a set $\MX$, and let $f : \MX^n \ra \BR$ and suppose non-negative constants $c_1, \ldots, c_n > 0$ are given so that, for each $i \in [n]$,
  \begin{align}
\sup_{x_1, \ldots, x_n, x_i'\in \MX} | f(x_1, \ldots, x_n) - f(x_1, \ldots, x_i', \ldots, x_n)| \leq c_i\label{eq:bd-asm}.
  \end{align}
  Write $v := \frac 14 \sum_{i=1}^n c_i^2$. Suppose $X_1, \ldots, X_n$ are independent $\MX$-valued random variables, and that $Z = f(X_1, \ldots, X_n)$. Then for all $\lambda > 0$, $\E[ \exp(\lambda \cdot (Z - \E[Z]))] \leq \exp(\lambda^2 v/2)$, and, for all $\delta \in (0,1)$,
  \begin{align}
\BP\left( |Z - \E[Z]| > \sqrt{2v \cdot \log(2/\delta)} \right) \leq \delta. 
  \end{align}
\end{lemma}

\begin{proof}[Proof of \cref{lem:mplayer-approx}]
  For each $j \in [m]$ and $s \in [k]$, let $B_{j,s} \in [n]$ denote the random variable defined by $\BP(B_{j,s} = i) = x_{j,i}$. For $b \in [n]$, let $e_b \in \BR^n$ denote the corresponding standard basis vector. For each $j \in [m]$ and $i \in [n]$, we define the following functions $f_j : [n]^{km} \ra \BR$ and $f_{j,i} :[n]^{k(m-1)} \ra \BR$:
  \begin{align}
    f_j(b_{1,1}, \ldots, b_{1,k}, \ldots, b_{m,1}, \ldots, b_{m,k}) :=& A_j \left( \frac 1k \sum_{s=1}^k e_{b_{1,s}}, \ldots, \frac 1k \sum_{s=1}^k e_{b_{m,s}} \right)\nonumber\\
    f_{j,i}(b_{1,1}, \ldots, b_{1,k}, \ldots, b_{m,1}, \ldots, b_{m,k}) :=& A_j \left( \frac 1k \sum_{s=1}^k e_{b_{1,s}}, \ldots, e_i, \ldots, \frac 1k \sum_{s=1}^k e_{b_{m,s}} \right)\nonumber.
  \end{align}
 Note that  $f_{j,i}$ does not depend on $b_{j,s}$ for any $s \in [k]$. 
Note also that each of the functions $f_j, f_{j,i}$ satisfies the bounded differences assumption \cref{eq:bd-asm} with upper bound (denoted $c_i$ in \cref{eq:bd-asm}) bounded above by $1/k$. 
  
We now define random $k$-uniform vectors $\hat X_1, \ldots, \hat X_m$, by $\hat X_j := \frac 1k \sum_{s=1}^k e_{B_{j,s}}$. Write $Z_j := f_j((B_{j,s})_{j \in [m], s \in [k]})$ and $Z_{j,i} = f_{j,i}((B_{j,s})_{j \in[m], s\in [k]})$. By \cref{lem:bd} applied to the function $f_j$ and independent random variables $B_{j,s}$ (so that we may take $v = \frac 14 \frac{km}{k^2} = \frac 14 \frac{m}{k}$), we have that, for any $\delta \in (0,1)$,
\begin{align}
\BP\left( |Z_j - \E[Z_j]| >  \sqrt{2m/k \cdot \log(2/\delta)} \right) \leq \delta\nonumber.
\end{align}
Observe that $Z_j = A_j(\hat X_1, \ldots, \hat X_m)$ and so $\E[Z_j] = A_j(x_1, \ldots, x_m)$. 
Noting that $\sqrt{2m/k \cdot \log( 8m / \delta  )} \leq \ep/2$ by our choice of $k$ (as long as $C_{\ref{lem:mplayer-approx}}$ is sufficiently large), we see that, for each $j \in [m]$,
\begin{align}
\BP\left( |A_j(\hat X_1, \ldots, \hat X_m) - A_j(x_1, \ldots, x_m)| > \ep/2 \right) \leq \frac{\delta}{4 m } \label{eq:zj-deviation}.
\end{align}

Next, for each $j \in [m], i \in [n]$, we apply \cref{lem:bd} to the function $f_{j,i}$ and independent random variables $B_{j,s}$, which yields, for $\lambda > 0$, $\E[\exp(\lambda \cdot (Z_j - \E[Z_j]))] \leq \exp(\lambda^2m/(2k))$. Using that $Z_{j,i} = A_j(e_i, \hat X_{-j})$ and so $\E[Z_{j,i}] = A_j(e_i, x_{-j})$, it follows that, for $\lambda > 0$,
\begin{align}
\E[\exp(\lambda \cdot (A_j(e_i, \hat X_{-j}) - A_j(e_i, x_{-j})))] \leq \exp(\lambda^2 m/(2k))\label{eq:ai-subgaussian}.
\end{align}
We next apply \cref{lem:smooth-class-lemma} with $N = n$, $U_i = A_j(e_i, \hat X_{-j}) - A_j(e_i, x_{-j})$, and $c = \sqrt{m/k}$, as well as to the negation of $U_i$. Then, the lemma gives that for any $\delta \in (0,1)$,
\begin{align}
  & \BP \left( \max_{w \in \MK_{\sigma, n}}     \left|A_j(w, \hat X_{-j}) - A_j(w, x_{-j})\right| > \sqrt{\frac{2m \log(2/(\delta \sigma))}{k} }\right)\nonumber\\
  = & \BP \left( \max_{w \in \MK_{\sigma, n}} \left|\sum_{i=1}^n w_i \cdot \left( A_j(e_i, \hat X_{-j}) - A_j(e_i, x_{-j})\right)\right| > \sqrt{\frac{2m \log(2/(\delta \sigma))}{k} }\right) \leq \delta\nonumber.
\end{align}
Noting that $\sqrt{2m/k \cdot \log(8 / \delta \sigma)} \le \ep/2$ by our choice of $k$ (as long as $C_{\ref{lem:mplayer-approx}}$ is sufficiently large), we see that, for each $j \in [m]$,
\begin{align}
\BP \left( \max_{w \in \MK_{\sigma, n}} A_j(w, \hat X_{-j}) - \max_{w \in \MK_{\sigma, n}} A_j(w, x_{-j}) > \ep/2 \right) \leq \BP \left( \max_{w \in \MK_{\sigma, n}}\left| A_j(w, \hat X_{-j}) - A_j(w, x_{-j}) \right|> \ep/2 \right) \leq \frac{\delta}{ 4 m } \label{eq:zij-deviation}.
\end{align}
 Let $\ME$ denote the event under which none of the $2m$ low-probability events in \cref{eq:zj-deviation,eq:zij-deviation} hold. 
 By the union bound, we $\BP(\ME) \geq 1 - \delta $. 
 For any $\hat X$ for which $\ME$ holds, we have that \cref{eq:s-unif-close} holds. 
 To establish that such $\hat X$ is in fact the desired equilibrium, note that under $\ME$, for each $j \in [m]$,
\begin{align}
  & \max_{w_j \in \MK_{\sigma, n}} A_j(w_j, \hat X_{-j}) - A_j(\hat X_1, \ldots, \hat X_m)
  \leq \ep + \max_{w_j \in \MK_{\sigma, n}} A_j(w_j, x_{-j}) - A_j(x_1, \ldots, x_m) \leq \ep\nonumber,
\end{align}
where the first inequality uses \cref{eq:zj-deviation,eq:zij-deviation}, and the second inequality uses the fact that $x_1, \ldots, x_m$ is a 0-approximate $\sigma$-smooth Nash equilibrium. 
\end{proof}

\section{Proofs for \cref{sec:query}} \label{sec:proof-query}

\subsection{Proof of \cref{lem:unif-sampling}}
\label{sec:use-coupling}
To prove \cref{lem:unif-sampling}, we require the following variant of the coupling lemma from \cite{DBLP:conf/focs/HaghtalabRS21} (simplified and generalized by \cite{DBLP:conf/colt/BlockDGR22,DBLP:conf/nips/HaghtalabHSY22}). 
\begin{lemma}[Coupling; \cite{DBLP:conf/focs/HaghtalabRS21}]
  \label{lem:coupling}
  Fix $n \in \BN$, $\sigma \in (0,1)$, and let $p \in \Delta^n$ be a $\sigma$-smooth distribution. For each $t,\ell \in \BN$, there is a coupling $\Pi$ so that $(B_1, D_{1,1}, \ldots, D_{1,\ell}, \ldots, B_t, D_{t,1}, \ldots, D_{t,\ell}) \sim \Pi$ (with $B_i, D_{i,j} \in [n]$ for $i \in [t], j \in [\ell]$) satisfies the following:
  \begin{enumerate}
  \item $B_1, \ldots, B_t$ are distributed independently from $p$.
  \item $D_{i,j}$ (for $i \in [t], j \in [\ell]$) are distributed independently and uniformly from $[n]$.
  \item With probability at least $1-t(1-\sigma)^\ell$, $B_i \in \{ D_{i,1}, \ldots, D_{i,\ell} \}$ for each $i \in [t]$. 
  \end{enumerate}
\end{lemma}

\begin{proof}[Proof of \cref{lem:unif-sampling}]
    Let $(x_1 , \dots , x_m)$ be a strong $\sigma$-smooth Nash equilibrium of $G$. 
    Note that such an equilibrium exists by \cref{thm:existence}.
    Recall from the definition of strong $\sigma$-smooth Nash equilibrium $x_i$ is a $\sigma$-smooth distribution for each $i \in [m]$. 
Fix $t =  \frac{C_{\ref{lem:mplayer-approx}}m \log(16 m/ \delta \sigma)}{\ep^2}, \ell = \frac{\log(2tm / \delta )  }{\sigma} $ and set $k=t\ell$ as in the statement of the theorem.
Let $(B_1^j, D^j_{1,1}, \ldots, B^j_{1,\ell}, \ldots, B^j_t, D^j_{t,1}, \ldots, D^j_{t,\ell})$ be distributed according to the coupling of \cref{lem:coupling} instantiated with the smooth distribution $x_j$. 
In particular, $D_{a,b}^j \sim \Unif([n])$ are independent for $a \in [t], b \in [\ell]$. 
For any $j \in [m]$, $  \hat X^j $ has distribution which is identical to that of $\sum_{a=1}^t \sum_{b=1}^\ell e\subs{D_{a,b}^j}$. Moreover, noting that $B_1^j, \ldots, B_t^j \sim x_j$ are independent samples from $x_j$, if we let $\hat Y^j := \frac{1}{t} \sum_{i=1}^t e\subs{B_i^j}$, then \cref{lem:mplayer-approx} gives that with probability $1-\delta/2$, $\hat Y := (\hat Y^1, \ldots, \hat Y^m)$ is a weak $\ep$-approximate $\sigma$-smooth Nash equilibrium. 
Moreover, by \cref{lem:coupling} and a union bound over $j$, with probability $1-\delta/2$, it holds that $B_i^j \in \{ D_{i,1}^j, \ldots, D_{i,\ell}^j\}$ for each $j \in [m]$ and $i \in [t]$. 
In particular, with probability at least $1- \delta$, it holds that there exists a weak $\ep$-approximate $\sigma$-smooth Nash equilibrium, namely $\hat Y$, for which $\hat Y^j$ is $t$-uniform over $\supp(\hat X^j) $. 
\end{proof}

\subsection{Proof of  \cref{lem:test-deviations}}
To establish \cref{lem:test-deviations}, the sets $\MR_j \subset [n]$ of actions sampled in \cref{alg:query-equilibrium} need to satisfy certain properties, which will be shown to hold with high probability. \cref{def:good-sample} below specifies these properties.
\begin{definition}[Good collection]
  \label{def:good-sample}
  Given a game with payoffs $(A_1, \ldots, A_m)$, consider a collection of multisets $ \MR_j \subset [n]$ (for $j \in [m]$), where $|\MR_j| = N \in \BN$ for each $j$. To simplify notation below, we will write $\MR_j = \{ r_{j,1}, \ldots, r_{j, N} \}$. 
  For each strategy profile $ x = ( x_1, \ldots,  x_m)$ and $j \in [m]$, let $\pi_1(j,  x), \ldots, \pi_n(j,  x) \in [n]$ denote a permutation of $[n]$ so that $A_j(\pi_1(j,  x),  x_{-j})\geq \ldots \geq A_j(\pi_n(j,  x),  x_{-j})$. Then define
  \begin{align}
\MT(j,  x) := \left\{ i \in [N] \ : \ \exists i' \leq \sigma n \mbox{ s.t. } \pi_{i'}(j,  x) = r_{j,i} \right\}\nonumber.
  \end{align}
  In words, $\MT(j,  x)$ contains the indices of the actions in $\MR_j$ which are amongst the best $\sigma$-fraction of actions for player $j$ when other players play according to $ x_{-j}$. For $\ep \in (0,1)$ and a strategy profile $x \in (\Delta^n)^m$, we say that the collection $(\MR_j)_{j \in [m]}$ is a \emph{$\ep$-good collection with respect to $x$} if the following properties hold:
  \begin{enumerate}
  \item $||\MT(j,  x)| - N\sigma| \leq \ep \sigma N$ for each $j \in [m]$.\label{it:tnsigma}
  \item \label{it:topa-approx} For each $j \in [m]$,
    \begin{align}
\left| \frac{1}{|\MT(j,  x)|} \sum_{i \in \MT(j,  x)} A_j(r_{j,i},  x_{-j}) - \frac{1}{\sigma n} \sum_{i'=1}^{\sigma n} A_j(\pi_{i'}(j,  x),  x_{-j}) \right| \leq \ep\nonumber.
    \end{align}
  \end{enumerate}
\end{definition}

\begin{lemma}[Testing deviations to smooth strategies]
  \label{lem:test-deviations}
  Suppose we are given payoff matrices $A_1, \ldots, A_m : [n]^m \ra [0,1]$, a strategy profile $x \in (\Delta^n)^m$, and a collection of multisets $\MR_j \subset [n]$ which is an $\ep$-good collection with respect to $x$. Moreover suppose that $\hat A_j : [n]^m \ra [0,1]$ agrees with $A_j$ on $\MR_j \times \supp(x_{-j})$. Then for any $\delta \in (0,1), j \in [m]$, $\OptimizeDeviation(\MR, \hat A_j, j, x, \sigma)$ (\cref{alg:optimize-deviation}) outputs a real number $\hat v$ so that
  \begin{align}
\left| \hat v - \max_{x_j' \in \MK_{\sigma, n}} A_j(x_j', x_{-j}) \right| \leq \ep\nonumber.
  \end{align}
\end{lemma}

\begin{proof}[Proof of \cref{lem:test-deviations}]
        Recalling that $x$ is fixed, let us write $\bar A_j(r_j) := A_j(r_j, x_{-j})$ for each $r_j \in [n]$. We write $\MR_j = \{ r_{j,1}, \ldots, r_{j,N} \}$. Moreover, we abbreviate $\MT := \MT(j,x)$, $\pi := \pi(j,x)$, and write $T := |\MT|$. The fact that $(\MR_j)_{j \in [m]}$ is an $\ep$-good collection with respect to $x$ (in particular, \cref{it:topa-approx} of \cref{def:good-sample}) gives that
      \begin{align}
\left| \frac{1}{T} \sum_{i \in \MT} \bar A_j(r_{j,i}) - \frac{1}{\sigma n} \sum_{i'=1}^{\sigma n} \bar A_j(\pi_{i'}) \right| \leq \ep\label{eq:tsum-sigman}. 
      \end{align}

      Let $\nu_1, \ldots, \nu_N \in [N]$ denote a permutation of $[N]$ so that $\bar A_j(r_{j, \nu_1}) \geq \cdots \geq \bar A_j(r_{j, \nu_N})$. The definitions of $\MT$ and $\pi$ give that, as multisets, we have $\{ \bar A_j(r_{j,i}) :\ i \in \MT \} = \{ \bar A_j(r_{j,\nu_1}), \ldots, \bar A_j(r_{j, \nu_T}) \}$. Therefore,
      \begin{align}
      \left|\frac{1}{T} \sum_{i \in \MT} \bar A_j(r_{j,i}) - \frac{1}{\sigma N} \sum_{i=1}^{\sigma N} \bar A_j(r_{j, \nu_i}) \right| = \left| \frac 1T \sum_{i=1}^T \bar A_j(r_{j,\nu_i}) - \frac{1}{\sigma N} \sum_{i=1}^{\sigma N} \bar A_j(r_{j, \nu_i}) \right| \leq \frac{|T-\sigma N|}{\sigma N} \leq {\ep}\label{eq:nusum},
      \end{align}
      where the first inequality uses \cref{lem:sorted-diff} and the second inequality uses the fact that, since $(\MR_j)_{j\in [m]}$ is an $\ep$-good collection with respect to $x$, we have $|T - N \sigma | \leq \ep \sigma N $ (\cref{it:tnsigma} of \cref{def:good-sample}). 

      Since we have assumed that $\hat A_j(r) = A_j(r)$ for all $r \in \MR_j \times \supp(x_{-j}) \subset [n]^m$, it holds that, for each $i \in [N]$, $\bar A_j(r_{j,i}) = A_j(r_{j,i}, x_{-j}) = \hat A_j(r_{j,i}, x_{-j}) = \hat v_i$, where $\hat v_i$ is as defined in \cref{alg:optimize-deviation}. Thus, $(\tau_1, \ldots, \tau_N)$, as defined in \cref{alg:optimize-deviation}, is equal to $(\nu_1, \ldots, \nu_N)$, and we have
      \begin{align}
\frac{1}{\sigma N} \sum_{i=1}^{\sigma N} \bar A_j(r_{j, \nu_i}) = \frac{1}{\sigma N} \sum_{i=1}^{\sigma N} \hat \nu_{\tau_i}\label{eq:baraj-equal-hatv}.
      \end{align}
      
      Note that, by definition of $\pi$, we have $\frac{1}{\sigma n} \sum_{i'=1}^{\sigma n} \bar A_j(\pi_{i'}) = \max_{x_j' \in \MK_{\sigma, n}} \bar A_j(x_j') = \max_{x_j' \in \MK_{\sigma ,n}} A(x_j', x_{-j})$.
      Thus, combining \cref{eq:tsum-sigman,eq:nusum,eq:baraj-equal-hatv} yields that $\left| \hat v - \max_{x_j' \in \MK_{\sigma ,n}} A(x_j', x_{-j}) \right| \leq 2\ep$, which is what we wanted to show. %
    \end{proof}

\begin{lemma}[High-probability guarantee for good collection]
  \label{lem:highprob-good}
  There is a constant $C_{\ref{lem:highprob-good}}$ so that the following holds. 
Fix $k, t \in \BN$, $\ep \in (0,1/2)$, $N \geq \frac{C_{\ref{lem:highprob-good}} \cdot tm \log (k/\delta)}{\ep^2 \sigma^2}$, and suppose that $\MS_1, \ldots, \MS_m \subset[ n]$ satisfy $|\MS_j| \leq k$ for each $j \in [m]$. For each $j \in [m]$, suppose that $\MR_j$ is a set consisting of $N$ uniformly random elements of $[n]$, chosen with replacement. With probability at least $1-\delta$ over the choice of $\MR_j$, the collection $(\MR_j)_{j \in [m]}$ is an $\ep$-good collection with respect to each $t$-uniform strategy profile $ x$ which is supported on $\MS_1 \times \cdots \times \MS_m$.
\end{lemma}
\begin{proof}[Proof of \cref{lem:highprob-good}]
  Let us fix any strategy profile $x$ and player  $j \in [m]$. Note that, over the randomness in $\MR_j$, we have $|\MT(j, x)| \sim \Bin(N, \sigma)$. By Hoeffding's inequality, as long as $N \geq \frac{C \log k^{tm}/\delta}{\ep^2 \sigma^2}$ for a sufficiently large constant $C$, it holds that $||\MT(j, x)| - N \sigma| \leq \ep \sigma N$ under some event $\ME^1_{j,x}$ which occurs with probability at least $1-\delta/(2k^{tm})$.

  For each $j \in [m]$, let us write $\MR_j = \{ r_{j,1}, \ldots, r_{j,N} \}$, where $r_{j,i} \in [n]$ for $i \in [N]$. For any $s \in [N]$, conditioned on the event $\{ |\MT(j, x)| = s \}$, the values $A_j(r_{j,i},  x_{-j})$, for $i \in \MT(j, x)$, are distributed i.i.d.~according to $\Unif( \{ A_j(\pi_1(j, x),  x_{-j}), \ldots, A_j(\pi_{\sigma n}(j,  x),  x_{-j}) \} )$. Thus, another application of Hoeffding's inequality yields that  that, conditioned on $\ME^1_{j,x}$ (which implies that $|\MT(j, x)| \geq \sigma N / 2$), with probability at least $1-\delta/ (2k^{tm})$, we have that
  \begin{align}
\left|\frac{1}{|\MT(j,x)|} \sum_{i \in \MT(j,x)} A_j(r_{j,i},  x_{-j}) - \frac{1}{\sigma n} \sum_{i'=1}^{\sigma n} A_j(\pi_{i'}(j,  x),  x_{-j}) \right| \leq \ep\label{eq:mtsum-hoeffding},
  \end{align}
  where we have used that $|\MT(j,x)| \geq \sigma N / 2 \geq \frac{C \log k^{tm}/\delta}{\ep^2}$, for a sufficiently large constant $C$. In particular, there is an event $\ME^2_{j,x}$ with $\Pr(\ME^2_{j,x}) \geq 1 - \delta/(2k^{tm})$, so that under $\ME^1_{j,x} \cap \ME^2_{j,x}$, \cref{eq:mtsum-hoeffding} holds.

  Note that the number of $t$-uniform strategy profiles $ x$ which are supported on $\MS_1 \times \cdots \times \MS_m$ is bounded above by $k^{tm}$. Let $\ME$ denote the intersection of $\ME^1_{j,x} \cap \ME^2_{j,x}$, over all $j \in [m]$, and $t$-sparse strategy profiles $x$ supported on $\MS_1 \times \cdots \times \MS_m$. By a union bound, we have that $\Pr(\ME) \geq 1-\delta$. By definition of $\ME$, it is evident that under $\ME$, we have that $(\MR_j)_{j \in [m]}$ is an $\ep$-good collection with respect to each $t$-uniform strategy profile $x$ supported on $\MS_1 \times \cdots \times \MS_m$. 
\end{proof}

\begin{lemma}
  \label{lem:sorted-diff}
  Suppose that $y_1 \geq \cdots \geq y_N$ is a non-increasing sequence of real numbers in $[0,1]$, of length $N \in \BN$. Then for any $S, S' \in [N]$, it holds that
  \begin{align}
\left| \frac 1S \sum_{i=1}^S y_i - \frac{1}{S'} \sum_{i=1}^{S'} y_i \right| \leq \frac{|S' - S|}{\max\{S, S' \}}.\nonumber
  \end{align}
\end{lemma}
\begin{proof}
  Without loss of generality we may assume that $S \leq S'$. Then
  \begin{align}
   \left|  \frac 1S \sum_{i=1}^S y_i - \frac{1}{S'} \sum_{i=1}^{S'} y_i \right|  = & \left| \left( \frac{1}{S} - \frac{1}{S'} \right) \cdot \sum_{i=1}^S y_i - \frac{1}{S'} \sum_{i=S+1}^{S'} y_i \right|
    \leq  \frac{S' - S}{S'} \nonumber,
  \end{align}
  where the second inequality uses that $y_i \in [0,1]$ for all $i \in [N]$. 
\end{proof}

\subsection{Proof of \cref{thm:query-equilibrium}}

\begin{proof}[Proof of \cref{thm:query-equilibrium}]
  Let $t, \ell, k$ be defined as in \cref{line:set-tl} of \cref{alg:query-equilibrium}. Then by \cref{lem:unif-sampling} with error probability $\delta/2$ and approximation parameter $\ep/4$, it holds that, under some event $\ME_1$ that occurs with probability $1-\delta/2$ over the choices of $\hat X_1, \ldots, \hat X_m$ in \cref{line:draw-hatx}, there is a $t$-uniform $\ep/4$-approximate $\sigma$-smooth Nash equilibrium supported on $\supp(\hat X_1) \times \cdots \times \supp(\hat X_m) = \MS_1 \times \cdots \times \MS_m$.

  Let $\ME_2$ be the event of \cref{lem:highprob-good} given parameters $k,t,\ep/4, N$, as well as the sets $\MS_1, \ldots, \MS_m \subset [n]$. In particular, $\Pr(\ME_2) \geq 1-\delta$ and, under $\ME_2$, the collection $(\MR_j)_{j \in [m]}$ is an $\ep/4$-good collection with respect to each $t$-uniform strategy profile $\hat x$ supported on $\MS_1 \times \cdots \times \MS_m$ (in particular, this is the set of strategy profiles in the for loop on \cref{line:s-hatx} of \cref{alg:query-equilibrium}). Thus, \cref{lem:test-deviations} guarantees that, under $\ME$, for each $\hat x$ in the for loop on \cref{line:s-hatx}, the output $\hat v_j$ of $\OptimizeDeviation(\MR_j, \hat A_j, x, \sigma)$ satisfies $| \hat v_j - \max_{x_j' \in \MK_{\sigma, n}} A_j(x_j, \hat x_{-j}) | \leq \ep/4$. Note that we hae used here that $\hat A_j$ agrees with $A_j$ on $\MR_j \times \supp(\hat x_{-j}) \subset \MR_j \times \prod_{j' \neq j} \MS_{j'}$, by \cref{line:set-ahat}. 
  Moreover, by \cref{line:set-atilde}, we have that $\tilde v_j = \tilde A_j(\hat x) = A_j(\hat x)$. Thus, under $\ME_2$, we have that, for each $j \in [m]$ and each strategy profile $\hat x$ considered in \cref{line:s-hatx}, 
  \begin{align}
\left| (\hat v_j - \tilde v_j) - \left( \max_{x_j' \in \MK_{\sigma, n}} A_j(x_j, \hat x_{-j}) - A_j(\hat x) \right) \right| \leq \ep/4\label{eq:hatv-close-deviation}.
  \end{align}
  Thus, under $\ME_1 \cap \ME_2$, there will be some $t$-uniform strategy profile $\hat x$ considered in the for loop on \cref{line:s-hatx} which is an $\ep/4$-approximate $\sigma$-smooth Nash equilibrium, and if such $\hat x$ is ever reached in the algorithm, we will have $\hat v_j - \tilde v_j \leq \ep/2 < \ep$. In particular, under $\ME_1 \cap \ME_2$, \cref{line:arbitrary-sp} will not be reached. Moreover, under $\ME_1 \cap \ME_2$, for any strategy profile $\hat x$ which is returned, by \cref{eq:hatv-close-deviation}, it must be a $3\ep/4$-approximate $\sigma$-smooth Nash equilibrium.

  It remains to  bound the query complexity of $\QueryEquilibrium$. The only queries are made in \cref{line:set-ahat,line:set-atilde}, for a total of $m \cdot \max\{k, N \} \cdot k^{m-1}$ queries. Using the values of $k,N$ set in \cref{line:set-tl,line:set-K}, the number of queries is therefore at most 
  \begin{align}
O \left( \frac{m^2 t \log(k/\delta)}{\ep^2 \sigma^2} \cdot (t\ell)^{m-1} \right) \leq O \left( m \cdot \left( \frac{m \log^2(m/(\delta \sigma \ep))}{\ep^2 \sigma} \right)^{m+1} \right)\nonumber.
  \end{align}
  \end{proof}

\section{Proofs for \cref{sec:eff-algs}}
\label{sec:proofs-weak-ub}

\subsection{Proofs from \cref{sec:weak-ub}}

\subsubsection{Proof of \cref{thm:eff-weak}}   \label{sec:eff-weak-proof}

\begin{proof}[Proof of \cref{thm:eff-weak}]
  Set $k = \frac{Cm \log(m/\sigma)}{\ep^2}$, where $C$ is the constant of \cref{lem:existence_k_unif}. The algorithm is straightforward: simply iterate over all possible $k$-uniform strategy profiles, check if each one is a weak $\ep$-approximate $\sigma$-smooth Nash equilibrium, and terminate once one is found. \cref{lem:existence_k_unif} guarantees that this algorithm will eventually succeed. Moreover, the number of $k$-uniform strategies for each player is bounded above by $n^k$, so the number of $k$-uniform strategy profiles is at most $n^{mk}$.

  To bound the running time of the algorithm, it suffices to show that, given a $k$-uniform strategy profile $\hat x = (\hat x_1, \ldots, \hat x_m)$, it may be efficient checked whether $\hat x$ is a weak $\ep$-approximate $\sigma$-smooth Nash equilibrium, i.e., whether
  \begin{align}
\max_{j \in [m]} \left\{ \max_{x_j' \in \MK_{\sigma, n}} A_j(x_j', \hat x_{-j}) - A_j(\hat x) \right\} \leq \ep\label{eq:eq-ver}.
  \end{align}
  But note that for each $j \in [m]$ and $i \in [n]$, the value $A_j(e_i, \hat x_{-j})$ can be computed in time $\poly(k, n^{O(m)}) \leq \poly(1/\ep, \log 1/\sigma, n^m)$. Moreover, the maximizer $x_j' \in \MK_{\sigma, n}$ of $A_j(x_j', \hat x_{-j})$ is simply the distribution which splits its mass uniformly over the $\sigma n$ largest indices $i \in [n]$ of $A_j(e_i, \hat x_{-j})$. Thus $\max_{x_j' \in \MK_{\sigma, n}} A_j(x_j', \hat x_{-j})$ can be computed in time $\poly(1/\ep, \log 1/\sigma, n^m)$, and clearly the same holds for $A_j(\hat x)$. Thus the left-hand side of \cref{eq:eq-ver} can be efficiently computed, and it follows that the overall algorithm runs in time $n^{O(mk)}$, as desired.
\end{proof}

\subsubsection{Proof of \cref{thm:weak-nash-constant}}
\label{sec:wn-constant-proof}

\begin{proof}[Proof of \cref{thm:weak-nash-constant}]
  \cref{thm:query-equilibrium} gives that $\QueryEquilibrium((A_1, \ldots, A_m), \sigma, \ep, \delta)$ outputs a weak $\ep$-approximate $\sigma$-smooth Nash equilibrium of the game $G$ with probability at least $1-\delta$. It only remains to bound the running time of $\QueryEquilibrium$, which is dominated by the for loop in \cref{line:s-hatx}. The number of iterations of this loop is $k^{tm}$, and each iteration calls \OptimizeDeviation, which requires $O(N \log N)$ time. (The parameters $k,t,N$ are defined on \cref{line:set-tl,line:set-K} of \QueryEquilibrium.) Altogether, the computational complexity may be bounded above by
  \begin{align}
O \left( N \log N \cdot k^{tm} \right) \leq \left( \frac{m \log(1/\delta)}{\sigma \ep} \right)^{O\left( \frac{m^2 \log(m/\delta \sigma)}{\ep^2} \right)}\nonumber.
  \end{align}
  
\end{proof}

\subsection{Proofs from \cref{sec:strong-ub}}
\label{sec:strong-proofs}
In this section, our main objective is to prove \cref{thm:alg-strong}.

\subsubsection{Proof for Strong Smooth Nash Equilibria in 2-player Games}
\label{sec:strong-2player-proof} 
As a warm-up, we begin with the case of $m=2$ players, for which we can establish that strong $\sigma$-smooth Nash equilibria may be found in essentially the same time as is guaranteed by \cref{thm:eff-weak} for weak equilibria:
\begin{theorem}[Polynomial-time algorithm for strong smooth Nash: 2-player version]
  \label{thm:strong-2player}
Let $A, B \in [0,1]^{n \times n}$ denote the payoff matrices of an 2-player normal-form game $G$. Then for any $\sigma, \ep \in (0,1)$, the algorithm $ \BimatrixStrongSmooth(n, \sigma, \ep)$ (\cref{alg:bimatrix-strong}) runs in time $n^{O \left( \frac{\log 1/\sigma}{\ep^2}\right)}$ and outputs a \textbf{strong} $\ep$-approximate $\sigma$-smooth Nash equilibrium.
\end{theorem}
\begin{proof}[Proof of \cref{thm:strong-2player}]
  Set $\ep_0 = \ep/4$, and $k = \frac{2C_{\ref{lem:mplayer-approx}} \log(2/\sigma)}{\ep_0^2}$, where $C_{\ref{lem:mplayer-approx}}$ is the constant of \cref{lem:mplayer-approx}. The algorithm \BimatrixStrongSmooth is given in \cref{alg:bimatrix-strong}.  %
To prove the theorem, we need to establish three facts: first, that for some $k$-uniform strategy profile $(\hat x, \hat y)$ which is a weak smooth Nash equilibrium, the program \cref{eq:2player-program} in \BimatrixStrongSmooth will be feasible, and second, that any feasible solution of \cref{eq:2player-program} (given that $(\hat x, \hat y)$ is a weak smooth Nash equilibrium) is an fact a strong smooth Nash equilibrium. Third, we need to show that the algorithm's running time is bounded appropriately. 

  \paragraph{Feasility of the program \cref{eq:2player-program}.} By \cref{lem:mplayer-approx} applied to the present game $G$ with $\ep = \ep_0$ and given as input a \emph{strong} 0-approximate $\sigma$-smooth Nash equilibrium $(x,y) \in (\Delta^n)^2$ (which exists by \cref{thm:existence}), there is some $k$-uniform strategy profile $(\hat x, \hat y)$ satisfying the following inequalities:
  \begin{align}
    | \hat x^\t A \hat y - x^\t A y | \leq \frac{\ep_0}{2}, \qquad  \sup_{x' \in \MK_{\sigma, n}} | (x')^\t A \hat y - (x')^\t Ay | \leq \frac{\ep_0}{2} \label{eq:a-strong}\\
    |\hat x^\t B \hat y -x^\t B y | \leq \frac{\ep_0}{2}, \qquad \sup_{y' \in \MK_{\sigma, n}} | \hat x^\t A y' - x^\t A y' | \leq \frac{\ep_0}{2}\label{eq:b-strong},
  \end{align}
  which must therefore be a weak $\ep_0$-approximate $\sigma$-smooth Nash equilibrium. We will show that the strategy profile $(x,y)$, satisfies all constraints of \cref{eq:2player-program}. First, since $(x,y)$ was chosen above to be a strong smooth Nash equilibrium, we certainly have $x,y \in \MK_{\sigma, n}$. Next, we compute
  \begin{align}
|x^\t A \hat y -\hat x^\t A \hat y| \leq |\hat x^\t A \hat y - x^\t A y | + |x^\t A \hat y - x^\t Ay| \leq \ep_0\nonumber,
  \end{align}
  where the second inequality uses \cref{eq:a-strong}; this establishes that $(x,y)$ satisfies \cref{eq:2pp-1}. Next, the second inequality in \cref{eq:a-strong} establishes that \cref{eq:2pp-4} is also satisfied by $(x,y)$. That \cref{eq:2pp-2,eq:2pp-3} are satisfied as well follows from symmetric arguments. This completes the proof that the program \cref{eq:2player-program} is feasible for some weak $\ep_0$-approximate $\sigma$-smooth Nash equilibrium $(\hat x, \hat y)$ (namely, the satisfying solution is strong equilibrium $(x,y)$ that we have chosen). 

  \paragraph{Correctness of solutions of the program \cref{eq:2player-program}.} Suppose $(x,y)$ solves \cref{eq:2player-program} for some $(\hat x, \hat y)$ which is a weak $\ep_0$-approximate $\sigma$-smooth Nash equilibrium. Then we have
  \begin{align}
| x^\t Ay - \hat x^\t A \hat y| \leq |x^\t Ay - x^\t A \hat y| + |x^\t A \hat y - \hat x^\t A \hat y| \leq 2\ep_0\label{eq:xy-hatxy-correct},
  \end{align}
  where the second inequality uses \cref{eq:2pp-4} with $ x' = x \in \MK_{\sigma, n}$ as well as \cref{eq:2pp-1}. Then we have
  \begin{align}
    \max_{\bar x \in \MK_{\sigma, n}} \bar x^\t Ay - x^\t Ay \le & \max_{\bar x \in \MK_{\sigma, n}}\left\{\bar x^\t A \hat y - x^\t Ay\right\} + \max_{\bar x \in \MK_{\sigma, n}} |\bar x^\t Ay - \bar x^\t A \hat y | \nonumber\\
    \leq & \max_{\bar x \in \MK_{\sigma, n}} \bar x^\t A \hat y - \hat x^\t A \hat y + 3\ep_0\leq 4\ep_0\nonumber,
  \end{align}
  where the second inequality uses \cref{eq:2pp-4} and \cref{eq:xy-hatxy-correct}, and the final inequality uses the fact that $(\hat x, \hat y)$ is a weak $\ep_0$-approximate $\sigma$-smooth Nash equilibrium. 
  In a symmetric manner it holds that $\max_{\bar y \in \MK_{\sigma, n}} x^\t A \bar y - x^\t Ay \le 3\ep_0$, thus establishing that $(x, y)$ is a $4\ep_0$-approximate $\sigma$-smooth Nash equilibrium.

  \paragraph{Bounding the running time.} Note that the outer loop of the algorithm will iterate through at most $n^k$ strategy profiles $(\hat x, \hat y)$. The argument in the proof of \cref{thm:eff-weak} establishes that whether $(\hat x, \hat y)$ is a weak $\ep_0$-approximate $\sigma$-smooth Nash equilibrium may be checked in time  $\poly(1/\ep_0, \log 1/\sigma, n)$. Moreover, feasibility of the linear program \cref{eq:2player-program} can be checked in time $\poly( n)$ using the ellipsoid algorithm, assuming that $\ep_0 \geq 1/n$. We refer to the proof of \cref{thm:alg-strong} for a detailed explanation of how this is done.  %
  Altogether, the algorithm runs in time $n^{O(k)} = n^{O \left( \frac{\log 1/\sigma}{\ep^2} \right)}$, as desired.
\end{proof}

\subsubsection{Proof for strong smooth Nash equilibria in $m$-player games}
\label{sec:proof-mplayer-approx-strong}

\begin{proof}[Proof of \cref{lem:mplayer-approx-strong}]
  For each $j \in [m]$ and $s \in [k]$, let $B_{j,s} \in [n]$ denote the random variable defined by $\BP(B_{j,s} = i) = x_{j,i}$. For $b \in [n]$, let $e_b \in \BR^n$ denote the corresponding standard basis vector. We define random $k$-uniform vectors $\hat X_1, \ldots, \hat X_m$, by $\hat X_j = \frac 1k \sum_{s=1}^k e_{B_{j,s}}$, for $j \in [m]$.

  Fix any values of $j,\ell \in [m]$. For any values of $i_1, \ldots, i_{\ell-1} \in [n]$, as well as any values of $\hat x_{\ell+1}, \ldots, \hat x_{m} \in \BR^n$, Hoeffding's inequality gives that the random variable $Z := A_j(i_{1:\ell-1}, \hat X_\ell, \hat x_{\ell+1:m})$ satisfies $\E[\exp(\lambda (Z - \E[Z]))] \leq \exp(\lambda^2/(8k))$ for all $\lambda \in \BR$. Since $\E[Z] = A_j(i_{1:\ell-1}, x_\ell, \hat x_{\ell+1:m})$, it follows that, for $\lambda \in \BR$,
  \begin{align}
\E[\exp ( \lambda \cdot (A_j(i_{1:\ell-1}, \hat X_\ell, \hat x_{\ell+1:m}) - A_j(i_{1:\ell-1}, x_\ell, \hat x_{\ell+1:m})))] \leq \exp(\lambda^2/(8k))\nonumber.
  \end{align}
  For each possible value of $\hat x_{\ell+1:m}$, we now apply \cref{lem:smooth-class-lemma} with $N = n^{\ell-1}, U_{i_{1:\ell-1}} = A_j(i_{1:\ell-1}, \hat X_\ell, \hat x_{\ell+1:m}) - A_j(i_{1:\ell-1}, x_\ell, \hat x_{\ell+1:m})$, and $c = 1/\sqrt{2k}$, as well as to the negation of $U_{i_{1:\ell-1}}$. Then the lemma gives that for any fixed $\hat x_{\ell+1:m}$, for  $\delta \in (0,1)$,
  \begin{align}
\BP \left( \max_{x_{1:\ell-1}' \in \MK_{\sigma^{\ell-1}, [n]^{\ell-1}}} \left| A_j(x'_{1:\ell-1}, \hat X_\ell, \hat x_{\ell+1:m}) - A_j(x'_{1:\ell-1}, x_\ell, \hat x_{\ell+1:m})\right| > \sqrt{\frac{2m \log(2/(\delta \sigma))}{k}} \right) \leq \delta\nonumber.
  \end{align}
  Using the fact that the random variables $\hat X_1, \ldots, \hat X_m$ are independent, by averaging the above over $\hat x_{\ell+1} \sim \hat X_{\ell+1}, \ldots, \hat x_m \sim \hat X_m$, we obtain
    \begin{align}
\BP \left( \max_{x_{1:\ell-1}' \in \MK_{\sigma^{\ell-1}, [n]^{\ell-1}}} \left| A_j(x'_{1:\ell-1}, \hat X_\ell, \hat X_{\ell+1:m}) - A_j(x'_{1:\ell-1}, x_\ell, \hat X_{\ell+1:m})\right| > \sqrt{\frac{2m \log(2/(\delta \sigma))}{k}} \right) \leq \delta\label{eq:all-l-highprob},
    \end{align}
    where the probability is now over the (independent) draws of the random variables $\hat X_\ell, \ldots, \hat X_m$.

    An analogous argument applied to the random variables $A_j(i_{1:\ell-1,j}, \hat X_\ell, \hat x_{\ell+1:m,-j})$, for each possible choice of $i_{1:\ell-1,j} \in [n]^\ell$ and $\hat x_{\ell+1:m,-j} \in (\Delta^n)^{m-\ell-1}$, gives that
    \begin{align}
\BP \left( \max_{x_{1:\ell-1,j}' \in \MK_{\sigma^{\ell-1}, [n]^{\ell-1}}} \left| A_j(x'_{1:\ell-1,j}, \hat X_\ell, \hat X_{\ell+1:m,-j}) - A_j(x'_{1:\ell-1,j}, x_\ell, \hat X_{\ell+1:m,-j})\right| > \sqrt{\frac{2m \log(2/(\delta \sigma))}{k}} \right) \leq \delta\label{eq:one-more-highprob}.
    \end{align}
    We now choose $\delta = 1/(8m^2)$. Let $C_0 > 0$ be a sufficiently large constant so that choosing $k = \frac{C_0 m \log(m/\sigma)}{\ep^2}$ yields $\sqrt{\frac{2m \log(2/(\delta \sigma))}{k}} \leq \ep$. Let $C_{\ref{lem:mplayer-approx-strong}} := \max\{ C_0, C_{\ref{lem:mplayer-approx}} \}$, where $C_{\ref{lem:mplayer-approx}}$ was defined in \cref{lem:mplayer-approx}. 

    By a union bound, we have that \cref{eq:all-l-highprob,eq:one-more-highprob} are satisfied for all $j,\ell \in [m]$ with probability at least $3/4$. Moreover, \cref{lem:mplayer-approx} and our choice of $C_{\ref{lem:mplayer-approx-strong}}$ gives that with probability at least $3/4$ over the draw of $\hat X_1, \ldots, \hat X_m$, the strategy profile $\hat X := (\hat X_1, \ldots, \hat X_m)$ is a weak $\ep$-approximate $\sigma$-smooth Nash equilibrium of $G$. Thus, by another union bound, with probability at least $1/2$ over the draw of $\hat X$, choosing $\hat x = \hat X$ satisfies the requirements of \cref{it:weaknash,it:approx-strong} of the lemma statement. Moreover, it is evident that \cref{it:kuniform} is satisfied with probability 1, thus completing the proof of the lemma.
\end{proof}

\begin{proof}[Proof of \cref{thm:alg-strong}]
  Set $\ep_0 = \ep/(2m)$ and $k = \frac{4 C_{\ref{lem:mplayer-approx-strong}}m \log(m/\sigma)}{\ep_0^2}$, where $C_{\ref{lem:mplayer-approx-strong}}$ is the constant of \cref{lem:mplayer-approx-strong}. The claimed algorithm is given by \GeneralStrongSmooth (\cref{alg:general-strong}). The algorithm iterates through all $k$-uniform strategy profiles $\hat x$ (\cref{line:multiplayer-kuniform}). For each one which is a weak $\ep_0$-approximate $\sigma$-smooth Nash equilibrium of the game $G$ specified by $(A_1, \ldots, A_m)$, the algorithm solves a linear program for each $\ell \in [m]$, specified by \cref{eq:multiplayer-program} (recall that the notation used in \cref{eq:multiplayer-program} was defined in \cref{sec:strong-ub}). This linear program searches for $x_\ell \in \MK_{\sigma, n}$ for each $\ell \in [m]$ which ``behaves similarly'' to $\hat x_\ell$ with respect to each player $j$'s payoff function (constraint \cref{eq:mpp-1}) as well as with respect to certain sets of distributions over certain players' strategies (constraint \cref{eq:mpp-2}). The particular structure of these deviations is chosen to ensure that a hybrid argument, detailed below, carries through. If \GeneralStrongSmooth ever finds a feasible solution $x_\ell^\st$ to all $m$ linear programs, it returns it. 
  \begin{algorithm}[t]
  \caption{$\GeneralStrongSmooth((A_1, \ldots, A_m), n, \sigma, \ep)$: compute strong smooth equilibria of $m$-player game}
  \label{alg:general-strong}
  \begin{algorithmic}[1]\onehalfspacing
    \State Set $\ep_0 = \ep/(2m)$ and $k = \frac{4 C_{\ref{lem:mplayer-approx-strong}}m \log(m/\sigma)}{\ep_0^2}$.
    \For{Each $k$-uniform strategy profile $\hat x = (\hat x_1, \ldots, \hat x_m)$}\label{line:multiplayer-kuniform}
    \If{$\hat x$ is a weak $\ep_0$-approximate $\sigma$-smooth Nash equilibrium of $(A_1, \ldots, A_m)$:}
    \State Set $x_1^\st, \ldots, x_m^\st \gets \perp$. 
    \State \multiline{For $\ell = 1,2,\ldots, m$, 
      solve the following feasiblity linear program for $x_\ell \in \BR^n$, using the ellipsoid algorithm:} %
        \begin{subequations}
          \label{eq:multiplayer-program}
          \begin{align}
           & \mbox{ Find $x_\ell \in \BR^n$}:\
             x_\ell \in \MK_{\sigma, n}\label{eq:mpp-0}\\
            & |A_j(x_{1:\ell-1}', x_\ell, \hat x_{\ell+1:m}) - A_j(x_{1:\ell-1}', \hat x_\ell, \hat x_{\ell+1:m}) | \leq \ep_0 \quad \forall j \in [m],\ x_{1:\ell-1}' \in \MK_{\sigma^{\ell-1}, [n]^{\ell-1}} \label{eq:mpp-1}\\
            & |A_j(x_{1:\ell-1,j}', x_\ell, \hat x_{\ell+1:m, -j}) - A_j(x_{1:\ell-1,j}', \hat x_\ell, \hat x_{\ell+1:m,-j})| \leq \ep_0\quad  \forall j > \ell, x_{1:\ell-1,j}' \in \MK_{\sigma^\ell, [n]^\ell}\label{eq:mpp-2}.
          \end{align}
        \end{subequations}
        \If{the program \cref{eq:multiplayer-program} is feasible for each $\ell$}
        \State Let $x_\ell^\st$ denote a feasible solution for each $\ell \in [m]$.
        \State \Return $(x_1^\st, \ldots, x_m^\st)$.
        \EndIf
      \EndIf
      \EndFor
\end{algorithmic}
\end{algorithm}

  To complete the proof of the theorem, we need to establish three facts: first, that for some $k$-uniform strategy profile $\hat x$ which is a weak smooth Nash equilibrium, the program \cref{eq:multiplayer-program} will be feasible for each $\ell \in [m]$; second, that if $x_\ell^\st$ are feasible solutions of \cref{eq:multiplayer-program} for the respective values of $\ell \in [m]$ (given that $\hat x$ is a weak smooth Nash equilibrium), then $(x_1^\st, \ldots, x_m^\st)$ is in fact a strong smooth Nash equilibrium. Finally, we need to bound the running time of the algorithm.

  \paragraph{Feasibility of the program \cref{eq:multiplayer-program}.} \cref{lem:mplayer-approx-strong}, applied with $\ep = \ep_0$, guarantees that there is some $k$-uniform strategy profile $\hat x$ which is a weak $(\ep_0/2)$-approximate $\sigma$-smooth Nash equilibrium of $G$, so that the program \cref{eq:multiplayer-program} is feasible for each $\ell \in [m]$. In particular, the feasible solution guaranteed by \cref{lem:mplayer-approx-strong} is the strong $\sigma$-smooth Nash equilibrium $x$ passed as input to the lemma (which always exists by \cref{thm:existence}), and the particular $k$-uniform strategy profile $\hat x$ for which such $x$ exists is the one whose existence is guaranteed by \cref{lem:mplayer-approx-strong}. 

  \paragraph{Correctness of solutions of the program \cref{eq:multiplayer-program}.} Suppose that $x_1^\st, \ldots, x_m^\st$ are feasible solutions of \cref{eq:multiplayer-program} for $\ell = 1, 2, \ldots, m$, given that $\hat x$ is a weak $\ep_0$-approximate $\sigma$-smooth Nash equilibrium. We may now compute, for each $j \in [m]$, 
  \begin{align}
     |A_j(x_1^\st, \ldots, x_m^\st) - A_j(\hat x_1, \ldots, \hat x_m)| 
    \leq & \sum_{\ell=1}^m | A_j(x_{1:\ell-1}^\st, \hat x_{\ell:m}) - A_j(x_{1:\ell}^\st, \hat x_{\ell+1:m})|\nonumber\\
    = & \sum_{\ell=1}^m |A_j(x_{1:\ell-1}^\st, \hat x_\ell, \hat x_{\ell+1:m}) - A_j(x_{1:\ell-1}^\st, x_\ell^\st, \hat x_{\ell+1:m})|\leq m\ep_0\label{eq:aj-xstar-xhat},
  \end{align}
  where the first inequality uses the triangle inequality and the second inequality follows from the fact that $x_\ell^\st$ satisfies \cref{eq:mpp-1}, together with the fact that $\prod_{i=1}^{\ell-1} x_i^\st \in \prod_{i=1}^{\ell-1} \MK_{\sigma, n} \subset \MK_{\sigma^{\ell-1}, [n]^{\ell-1}}$. %

  In a similar manner, for each $j \in [m]$ and each $\bar x_j \in \MK_{\sigma, n}$, we have
  \begin{align}
    & |A_j(x_{-j}^\st, \bar x_j) - A_j(\hat x_{-j}, \bar x_j)|\nonumber\\
    \leq & \sum_{\ell=1}^{j-1} |A_j(x_{1:\ell-1}^\st, \hat x_{\ell:m, -j}, \bar x_j) - A_j(x_{1:\ell}^\st, \hat x_{\ell+1:m, -j}, \bar x_j)| + \sum_{\ell=j+1}^m |A_j(x_{1:\ell-1,-j}^\st, \hat x_{\ell:m}, \bar x_j) - A_j(x_{1:\ell,-j}^\st, \hat x_{\ell+1:m}, \bar x_j)| \nonumber\\
    \leq & (m-1) \ep_0\label{eq:aj-xstar-barx},
  \end{align}
  where the first inequality uses the triangle inequality and the second inequality uses \cref{eq:mpp-1} for values $j+1 \leq \ell \leq m$ together with the fact that $\bar x_j \times \prod_{i \in [\ell-1]\backslash j} x_i^\st \in \prod_{i=1}^{\ell-1} \MK_{\sigma, n} \subset \MK_{\sigma^{\ell-1}, [n]^{\ell-1}}$, as well as  \cref{eq:mpp-2} for values $1 \leq j \leq \ell-1$ together with the fact that $\bar x_j \times \prod_{i=1}^{\ell-1} x_i^\st \in \prod_{i=1}^\ell \MK_{\sigma, n} \subset \MK_{\sigma^\ell, [n]^\ell}$. Then we have, for each $j \in [m]$, 
  \begin{align}
    & \max_{\bar x_j \in \MK_{\sigma, n}} A_j(x_{-j}^\st, \bar x_j) - A_j(x^\st) \nonumber\\
    \leq & \max_{\bar x_j \in \MK_{\sigma, n}} A_j( \hat x_{-j}, \bar x_j) - A_j(\hat x) + \max_{\bar x_j \in \MK_{\sigma, n}} | A_j(x_{-j}^\st, \bar x_j) - A_j(\hat x_{-j}, \bar x_j)| + |A_j(x^\st) - A_j(\hat x)|\nonumber\\
    \leq & 2m\ep_0 \leq \ep\nonumber,
  \end{align}
  where the second inequality uses \cref{eq:aj-xstar-xhat,eq:aj-xstar-barx} as well as the fact that $\hat x$ is a weak $\ep_0$-approximate $\sigma$-smooth Nash equilibrium, and the final inequality uses the choice of $\ep_0$. 

  \paragraph{Bounding the running time.} Note that the outer loop of the algorithm will iterate through at most $n^{mk}$ strategy profiles $\hat x$. The argument in the proof of \cref{thm:eff-weak} establishes that whether $\hat x$ is a weak $\ep_0$-approximate $\sigma$-smooth Nash equilibrium may be checked in $\poly(1/\ep_0, \log 1/\sigma, n^m)$. Finally, we claim that feasibility of the linear program \cref{eq:multiplayer-program} can be checked in time $\poly(1/\ep_0, \log 1/\sigma, n^m)$. We do this using the ellipsoid algorithm; in particular, we will apply \cite[Theorem 2.4]{bubeck2014convex}. To do so, we must ensure that the linear program \cref{eq:multiplayer-program} has an efficient separation oracle. Even though the number of constraints is exponentially large,\footnote{At first glance, it may seem that there are infinitely many constraints, since $\MK_{\sigma^{\ell}, [n]^\ell}$ is infinite. However, we only need to check \cref{eq:mpp-1,eq:mpp-2} on a set of distributions whose convex hull is $\MK_{\sigma^\ell, [n]^\ell}$, of which there is a set of size at most $n^{\ell \cdot (\sigma n)^\ell}$. } a separation oracle may be implemented efficiently, as follows: given $x_l \in \BR^n$, clearly we may check the $O(n)$ constraints that ensure $x_\ell \in \MK_{\sigma,n}$ in $O(n)$ time. For the constraints \cref{eq:mpp-1}, given $x_\ell \in \BR^n$, for each $j \in [m]$, we find the $(\sigma n)^{\ell-1}$ tuples $(b_1, \ldots, b_{\ell-1}) \in [n]^{\ell-1}$ so that the value of $F_j(b_{1:\ell-1}) := A_j(b_{1:\ell-1}, x_\ell, \hat x_{\ell+1:m}) - A_j(b_{1:\ell-1}, \hat x_\ell, \hat x_{\ell+1:m})$ is largest (amongst all possible values in $[n]^{\ell-1}$). If the mean of $F$ over these $(\sigma n)^{\ell-1}$ tuples is greater than $\ep_0$, then the uniform distribution over these $(\sigma n)^{\ell-1}$ tuples yields a separating hyperplane. Otherwise, we repeat with $F_j(\cdot)$ replaced by $-F_j(\cdot)$, and handle the constraints \cref{eq:mpp-2} in an analogous manner. If we do not find any separating hyperplane in this manner, it is clear that $x_\ell$ is a feasible point of \cref{eq:multiplayer-program}. Summarizing, a separating hyperplane can be found in $n^{O(m)}$ time.

  Finally, to apply Theorem 2.4 of \cite{bubeck2014convex}, we need to show that if \cref{eq:multiplayer-program} is feasible, then its feasible region is (a) contained in a Euclidean ball of some radius $R$ and (b) contains a Euclidean ball of radius $r$; the running time of the ellipsoid algorithm depends linearly on $\log(R/r)$. As for (a), we may take $R = 1$ since $\MK_{\sigma, n} \subset \Delta^n$, which has $\ell_2$-diameter at most $2$. Strictly speaking, (b) does not hold, though we can correct for this by instead solving the modification of \cref{eq:multiplayer-program} where \cref{eq:mpp-0} is replaced by
  \begin{align}
0 \leq x_{\ell,i} \leq \frac{1}{\sigma n} + n^{-C}, \qquad 1- n^{-C} \leq \sum_{i=1}^n x_{\ell, i} \leq 1+ n^{-C}\nonumber,
  \end{align}
  for a sufficiently large constant $C$ (to be specified below), and the value $\ep_0$ in \cref{eq:mpp-1,eq:mpp-2} is replaced by $3\ep_0/4$. Our argument above ensures that \cref{eq:multiplayer-program} is actually feasible with the $\ep_0$ in \cref{eq:mpp-1,eq:mpp-2} replaced by $\ep_0/2$. Thus, our modification of \cref{eq:multiplayer-program} contains an $\ell_2$ ball of radius $r = n^{-2C}$, as long as $C$ is a sufficiently large constant. Therefore, we can find a feasible point $\tilde x_\ell$ of this modification of \cref{eq:multiplayer-program} using the ellipsoid method, in time $n^{O(m)}$. Finally, assuming $C$ is a sufficiently large constant, by projecting $\tilde x_\ell$ onto $\MK_{\sigma, n}$ we obtain a feasible point of \cref{eq:multiplayer-program}, as desired.

  Altogether, \GeneralStrongSmooth runs in time $n^{O(mk)} = n^{O \left( \frac{m^4 \log (m/\sigma)}{\ep^2} \right)}$, as desired.
\end{proof}

\section{Proofs from \cref{sec:hardness}}

\subsection{Primer on \PPAD} \label{sec:ppad}
In this section we formally introduce the complexity class \PPAD. Formally, a \emph{search problem} is specified by a relation $\MR \subset \{0,1\}^\st \times \{0,1\}^\st$, which is to be interpreted as follows: for $(x,y) \in \MR$, $x$ is the input instance, and $y$ is a solution to $x$. The class \TFNP consists of search problems $\MR$ so that (a) for all $x$, there is some $y$ with $|y| \leq \poly(|x|)$ so that $(x,y) \in \MR$, and (b) for all $(x,y)$, there is a $\poly(|x| + |y|)$ algorithm to determine whether $(x,y) \in \MR$. We need the following concept of polynomial-time reductions between search problems:
\begin{defn}
  \label{def:poly-reducible}
  A search problem $\MR_1 \in \TFNP$ is polynomial-time reducible to $\MR_2 \in \TFNP$ if there is a pair of polynomial-time computable functions $F, G$ so that, for every $x \in \{0,1\}^\st$, if $y \in \{0,1\}^\st$ satisfies $(F(x), y) \in \MR_2$, then $(x, G(y)) \in \MR_1$.
\end{defn}
Intuitively, given an input $x$ for $\MR_1$, we may map it to some input $F(x)$ for $\MR_2$, and given a solution $y$ for $F(x)$, we may map it to a solution $G(y)$ for $x$. \PPAD is defined to be the subclass of \TFNP of problems with a polynomial-time reduction to the \texttt{End-of-the-Line} (\texttt{EoL}) problem, defined below.

\begin{definition}[\texttt{End-of-the-Line}]
  \label{def:ppad}
  Define the \texttt{End-of-the-Line} (\texttt{EoL}) problem as the search problem where the input is two circuits $S,P : \left\{ 0,1 \right\}^n \to \left\{ 0,1 \right\}^n $ such that $S(0^n) \neq 0^n $ and $P(0^n) = 0^n $. 
  The output is a $v \in \left\{ 0,1 \right\}^n$ $v \neq 0^n$ such that $ S(P(v))  \neq v  $ or $ P(S(v))  \neq v  $. 
\end{definition}
Intuitively, the circuits $S,P$ in \cref{def:ppad} define a directed graph on $\{0,1\}^n$ where each vertex, apart from sources and sinks, is promised to have a unique successor and predecessor. Given a source (namely, $0^n$), the goal is to find a sink, or a violation of the promise. 

\PPAD contains as notable complete problems the problem of computing approximate Nash equilibria, computing approximate Brouwer fixed points, and computing approximate Arrow-Debreu equilibria.
Further, under cryptographic assumptions, \PPAD does not have polynomial-time algorithms (see \cite{DBLP:conf/tcc/BitanskyCHKLPR22,DBLP:conf/crypto/GargPS16} and references therein).
Thus, \PPAD hardness is strong evidence that a problem is intractable.

Some of our lower bounds rely on the following conjecture, which makes the stronger assumption that there is no \emph{sub-exponential} time algorithm for \PPAD:
\begin{Conjecture}[ETH for \PPAD \cite{DBLP:journals/sigecom/Rubinstein17}] \label{conj:ppad-eth}
  Any algorithm that solves $\mathrm{EoL}$ requires time $2^{\Omega(n)}$.
\end{Conjecture}
In particular, we shall need the following consequence of \cref{conj:ppad-eth}:
\begin{theorem}[Theorem 1.2 of \cite{DBLP:journals/sigecom/Rubinstein17}]
  \label{thm:quasipoly-hardness}
  Assuming \cref{conj:ppad-eth}, there is some constant $\ep_0 > 0$ so that there is no algorithm which computes $\ep_0$-approximate Nash equilibria in 2-player $n$-action games in time $n^{\log^{1-\delta}n }$ for any constant $\delta > 0$.
\end{theorem}

\subsection{Proofs from \cref{sec:harness-polynomial}}

\begin{proof}[Proof of \cref{thm:smooth-nash-nash}]
      Fix $n \in \BN$, and let $G = (A,B)$ denote a 2-player, $n$-action normal-form game. Set $k = \lceil n^{(1-c)/c} \rceil$ and $N = nk$. Moreover let $J_{k\times k} \in \BR^{k \times k}$ denote the $k \times k$ all-1s matrix. %
      We define a 2-player, $N$-action normal-form game $G' :=(A', B')$ by: $A' := J_{k \times k} \otimes A, B' := J_{k \times k} \otimes B$, so that $A', B' \in \BR^{N \times N}$. Write $\sigma := \sigma(N), \ep := \ep(N)$. 
    
      Let $(x,y)$ denote a weak $\ep$-approximate $\sigma$-smooth Nash equilibrium in $G'$, so that $x,y \in \Delta^N$. We define distributions $x', y' \in \Delta^n$, as follows: for all $ i \in [n]$,
      \begin{align}
    x'_{i} = \sum_{j=0}^{k-1} x_{i + jk}, \qquad y'_{i + \ell k } =  \sum_{j=0}^{k-1} y_{i+jk}\nonumber.
      \end{align}
      It is immediate that $x',y'$ are indeed distributions. Moreover, we claim that $(x', y')$ is a weak $\ep$-approximate $\sigma$-smooth Nash equilibrium of $G'$. To see this claim, first, for any $z \in \Delta^n$, let $J_{k,z} \in \Delta^N$ denote the distribution defined by $(J_{k,z})_{i + \ell k} = z_i/k$ for all $0 \leq \ell < k$ and $i \in [n]$. Then for all $z \in \Delta^n$, we have
      \begin{align}
        (x')^\t Bz = \sum_{0 \leq \ell < k, i \in [n]}  \sum_{t \in [n]}  x_{i+\ell k}\cdot B_{i, t} z_t = \frac 1k \sum_{0 \leq \ell < k, i \in [n]} \sum_{ 0 \leq j < k, t \in [n]}x_{i+\ell k} \cdot B'_{i+\ell k, t+jk} (J_{k,z})_{t + jk} = x^\t B' J_{k,z}\label{eq:xbt},
      \end{align}
      where the first equality uses the definition of $x'$, and the second equality uses the definition of $B' = J_{k\times k} \otimes B$ as well as the definition of $J_{k,i}$. Note that, if $z \in \MK_{\sigma, n}$, then we have $\max_{i \in [n]} z_i \leq \frac{1}{\sigma n}$, which implies that $\max_{i \in [N]} (J_{k,z})_i \leq \frac{1}{\sigma nk} = \frac{1}{\sigma N}$, and hence $J_{k,z} \in \MK_{\sigma, N}$. 
    
      Finally, we observe that
      \begin{align}
    (x')^\t B y' = \sum_{i \in [n], 0 \leq \ell < k} \sum_{t \in [n], 0 \leq j < k} x_{i+\ell k} B_{i,t} y_{t+ jk} = \sum_{i \in [n], 0 \leq \ell < k} \sum_{t \in [n], 0 \leq j < k} x_{i+\ell k} B'_{i+\ell k,t + jk} y_{t+ jk} = x^\t By \label{eq:xby}.
      \end{align}
    
      Since $(x,y)$ is a weak $\ep$-approximate $\sigma$-smooth Nash equilibrium in $G'$, we have that 
      \begin{align}
    \max_{z \in \MK_{\sigma, n}} (x')^\t B z - (x')^\t B y' = \max_{z \in \MK_{\sigma, n}} x^\t B' J_{k,z} - x^\t B' y \leq \max_{\bar y \in \MK_{\sigma, N}} x^\t B' \bar y - x^\t B' y \leq \ep\nonumber,
      \end{align}
      where the equality uses \cref{eq:xbt,eq:xby} and the first inequality uses the fact that $J_{k,z} \in \MK_{\sigma, N}$ if $z \in \MK_{\sigma, n}$. In a symmetric manner, we may obtain that
      \begin{align}
    \max_{z \in \MK_{\sigma, n}} z^\t Ay' - (x')^\t Ay' = \max_{z \in \MK_{\sigma, n}} J_{k,z}^\t A' y - x^\t A' y \leq \max_{\bar x \in \MK_{\sigma, N}} \bar x^\t A' y - x^\t A' y \leq \ep\nonumber.
      \end{align}
      It follows that $(x', y')$ is a weak $\ep$-approximate $\sigma$-smooth Nash equilibrium of $G$.
    
      Note that, since $c > 0$ is a positive constant, the size of the game $G'$ is bounded by a polynomial in the size of the game $G$. In particular, we have that $N = n \lceil n^{(1-c)/c} \rceil \geq n^{1/c}$, and also $N \leq 2n^{1/c}$.  It follows  that $\ep = \ep(N) \leq \ep(n^{1/c})$, since $\ep(\cdot)$ is a non-increasing function. Similarly, we have that $\frac{1}{\sigma n} = \frac{1}{\sigma(N) \cdot n} \geq \frac{1}{\sigma(n^{1/c}) \cdot n}$, meaning that the pair $(x', y')$ we have computed above is a weak $\ep(n^{1/c})$-approximate $\sigma(n^{1/c})$-smooth Nash equilibrium of $G$. Thus we have established a polynomial-time reduction from the problem of finding weak $\sigma(n^{1/c})$-smooth $\ep(n^{1/c})$-approximate Nash equilibria in 2-player $n$ action games to the problem of finding weak $\sigma(N)$-smooth $\ep(N)$-approximate Nash equilibria in 2-player $N$ action games. The statement of the theorem involving weak equilibria follows by replacing $\sigma$ with the function $n \mapsto \sigma(n^{1/c})$ and replacing $\ep$ with the function $n \mapsto \ep(n^{1/c})$. %
    
      As for strong equilibria, we only need to show that if $x,y \in \MK_{\sigma, N}$, then $x', y' \in \MK_{\sigma, n}$; but this holds since $\max_{i \in [n]} x_i' \leq k \cdot \max_{i \in [N]} x_i \leq \frac{k}{\sigma N} = \frac{1}{\sigma n}$, and similarly $\max_{i \in [n]} y_i' \leq \frac{1}{\sigma n}$. 
          \end{proof}

    \begin{proof}[Proof of \cref{cor:smooth-nash-eth}]
      \cref{thm:smooth-nash-nash} with $\ep(n) = \ep_0$ and $\sigma(n) = n^{-1}$ establishes the following:  an algorithm which finds weak $\ep_0$-approximate $(1/N)^c$-smooth Nash equilibrium in 2-player $N$-action games in time $N^{\log^{1-\delta} N}$ (for all sufficiently large $N \in \BN$) would imply that there is some constant $b$ so that there is an algorithm that runs in time $(n^b)^{\log^{1-\delta}(b)}$ and computes an $\ep_0$-approximate $1/n$-smooth Nash equilibrium in 2-player $n$-action games. Since the requirement of $1/n$-smoothness is vacuous, and since $(n^b)^{\log^{1-\delta}(b)} = n^{b^2 \cdot \log^{1-\delta}(b)} \leq n^{\log^{1-\delta/2}(n)}$ for sufficiently large $n$, we get a contradiction to the exponential time hypothesis for \PPAD by \cref{thm:quasipoly-hardness}. 
\end{proof}

\paragraph{PPAD-hardness.} %
In similar spirit to \cref{cor:smooth-nash-eth}, we have the following corollary of \cref{thm:smooth-nash-nash}, which shows a \PPAD-hardness result when the approximation parameter and smoothness parameter are inverse polynomials (in particular, the result does not depend on ETH for \PPAD). The result is strengthened by \cref{thm:ppad-hard-const-sig}, which allows the smoothness parameter to be an absolute constant. 
\begin{corollary}
  \label{cor:smooth-nash-ppad-hard}
For any constant $c_1 \in (0,1)$, the problem of finding weak $n^{-c_1}$-approximate $n^{-c_1}$-smooth Nash equilibrium in 2-player, $n$-action games is \PPAD-hard. 
\end{corollary}
\begin{proof}[Proof of \cref{cor:smooth-nash-ppad-hard}]
 It is known that the problem of finding a $n^{-1}$-approximate $n^{-1}$-smooth Nash equilibrium in 2-player, $n$-action games is \PPAD-hard \cite{DBLP:journals/siamcomp/DaskalakisGP09,DBLP:conf/focs/ChenD06}. 
 We apply \cref{thm:smooth-nash-nash} with $\ep(n) = n^{-1}$ and $\sigma(n) = n^{-1}$, and with $c = c_1$. In particular, the problem of finding weak $n^{-1}$-approximate $n^{-1}$-smooth Nash equilibrium in 2-player $n$-action games has a polynomial-time reduction to the problem of finding weak $n^{-c_1}$-approximate $n^{-c_1}$-smooth Nash equilibrium in 2-player $n$-action games. %
\end{proof}

    \subsection{Proofs from \cref{sec:harness-constant}}
\label{sec:const-sig-lb-proofs}

    \begin{proof}[Proof of \cref{lem:validity}]
      Let $(x,y)$ be a 1-approximate $\sigma$-smooth Nash equilibrium of $G$. %
    
      Since $(A,B)$ is an approximate GMP game, it holds that, for all strategy profiles $(x,y)$,
      \begin{align}
        M \cdot \sum_{k=1}^K \bar x_k \bar y_k \leq x^\t Ay \leq M \cdot \sum_{k=1}^K \bar x_k \bar y_k + 1\label{eq:xay-sandwich}\\
        -M \cdot \sum_{k=1}^K \bar x_k \bar y_k \leq x^\t By \leq -M \cdot \sum_{k=1}^K \bar x_k \bar y_k + 1\label{eq:xby-sandwich}.
      \end{align}
      
      We first prove the following helpful claim:
      \begin{claim}
        \label{clm:ip-ub}
        It holds that $1/K - 2/M \leq \sum_{k\in [K]} \bar x_k \bar y_k \leq 1/K + 2/M$.
      \end{claim}
      \begin{proof}
        Let $k_1, \ldots, k_{\sigma K} \in [K]$ denote indices so that $\bar y_{k_1}, \ldots, \bar y_{k_{\sigma K}}$ are the $\sigma K$ largest entries of $\bar y$. Let $x' \in \MK_{\sigma, N}$ be defined so that for all $i \in [\sigma K]$, $x'_{2k_i-1} = x'_{2k_i} = \frac{1}{2\sigma K}$, and hence $\bar x'_{k_i} = \frac{1}{\sigma K}$. By definition of the indices $k_i$, it holds that $\sum_{i=1}^{\sigma K} \bar y_{k_i} \geq \sigma$. Thus $\sum_{i=1}^{\sigma K} \bar x_{k_i}' \bar y_{k_i} \geq \frac{\sigma}{\sigma K} = 1/K$. Using \cref{eq:xay-sandwich} and the fact that $(x,y)$ is a 1-approximate $\sigma$-smooth Nash equilibrium, we have
        \begin{align}
    1 \geq (x')^\t Ay - x^\t Ay \geq M \cdot \sum_{k\in [K]} \bar x_k' \bar y_k - M \cdot \sum_{k \in [K]} \bar x_k \bar y_k - 1 \geq M/K - M \cdot \sum_{k\in[K]} \bar x_k \bar y_k - 1\nonumber.
        \end{align}
        Rearranging, we obtain that $\sum_{k \in [K]} \bar x_k \bar y_k \geq 1/K - 2/M$.
    
        The proof that $\sum_{k \in [K]} \bar x_k \bar y_k \leq 1/K + 2/M$ is symmetric.  In particular, let $k_1, \ldots, k_{\sigma K} \in [K]$ denote indices so that $\bar x_{k_1}, \ldots, \bar x_{k_{\sigma K}}$ are the $\sigma K$ smallest entries of $\bar x$. Let $y' \in \MK_{\sigma, N}$ be defined so that for all $i \in [\sigma K]$, $y'_{2k_i-1} = y'_{2k_i} = \frac{1}{2\sigma K}$, and hence $\bar y_{k_i}' = \frac{1}{\sigma K}$. By definition of the indices $k_i$, it holds that $\sum_{i=1}^{\sigma K} \bar x_{k_i} \leq \sigma$. Thus $\sum_{i=1}^{\sigma K} \bar x_{k_i} \bar y_{k_i}' \leq \frac{\sigma}{\sigma K} = 1/K$. Using \cref{eq:xby-sandwich} and the fact that $(x,y)$ is a 1-approximate $\sigma$-smooth Nash equilibrium, we have
        \begin{align}
1 \geq x^\t B y' - x^\t By \geq -M \cdot \sum_{k=1}^K \bar x_k \bar y_k' + M \cdot \sum_{k=1}^K \bar x_k \bar y_k -1 \geq - M/K + M \cdot \sum_{k=1}^K \bar x_k \bar y_k - 1\nonumber.
        \end{align}
        Rearranging, we obtain that $\sum_{k \in [K]} \bar x_k \bar y_K \leq 1/K + 2/M$. 
      \end{proof}
    
      Again let us let $\bar y_{k_1}, \ldots, \bar y_{k_{\sigma K}}$ denote the $\sigma K$ largest entries of $\bar y$. Let us write $\ep_0 = \ep/K$. Suppose for the purpose of contradiction that $\bar y_{k_1} > 1/K + \ep_0$. For simplicity let $\delta := 1/K + \ep_0$. It follows that
      \begin{align}
    \sum_{i=1}^{\sigma K} \bar y_{k_i} \geq \delta + (1-\delta) \cdot \frac{\sigma K - 1}{K-1}\nonumber.
      \end{align}
      Then we have
      \begin{align}
        \frac{1}{\sigma K} \cdot \sum_{i=1}^{\sigma K} \bar y_{k_i} \geq & \frac{1}{\sigma K} \cdot \left( \frac{K\delta + \sigma K - \delta \sigma K - 1}{K-1} \right)\nonumber\\
        = & \frac{\ep_0/\sigma + 1 - 1/K - \ep_0}{K-1}\nonumber\\
        =& \frac{1}{K} + \frac{\ep_0}{K-1} \cdot (1/\sigma - 1)\nonumber\\
        \geq & \frac{1}{K} + \frac{\ep_0}{2\sigma (K-1)}\label{eq:bary-sigmak-lb},
      \end{align}
      where the final inequality uses that $\sigma \leq 1/2$. Letting $x' \in \MK_{\sigma, N}$ be defined by $x_{2k_i-1}' = x_{2k_i}' = 1/(2\sigma K)$ for all $i \in [\sigma K]$ (i.e., so that $\bar x_{k_i}' = 1/(\sigma K)$ for $i \in [\sigma K]$), we have that
      \begin{align}
    1 \geq (x')^\t Ay - x^\t Ay \geq M \cdot \sum_{k \in [K]} \bar x_k' \bar y_k - M \cdot \sum_{k \in [K]} \bar x_k \bar y_k - 1 \geq \frac{M\ep_0}{2\sigma (K-1)} -3\nonumber,
      \end{align}
    where we have used that $(x,y)$ is a 1-approximate $\sigma$-smooth Nash equilibrium in the first inequality, \cref{eq:xay-sandwich} in the second inequality, and \cref{eq:bary-sigmak-lb} together with \cref{clm:ip-ub} in the third inequality. The above inequality leads to a contradiction since $M, \ep_0$ are chosen so that $M\ep_0 \geq 8\sigma K > 8\sigma (K-1)$ (using our assumption that $\ep \geq 8\sigma K^2/M$). 
    
    We have established that $\bar y_k \leq 1/K + \ep_0$ for all $k \in [K]$. It follows that $\bar y_k \geq 1/K - K\ep_0$ for all $k \in [K]$. Since $\ep_0 = \ep/K$, we have thus established that $\bar y_k \in [1/K - \ep, 1/K + \ep]$ for all $k \in [K]$.
    
    The proof of the analogous statement that $\bar x_k \in [1/K - \ep, 1/K + \ep]$ for all $k \in [K]$ is symmetric, using \cref{eq:xby-sandwich} instead of \cref{eq:xay-sandwich}. In particular, let $\bar x_{k_1}, \ldots, \bar x_{k_{\sigma K}}$ denote the $\sigma K$ smallest entries of $\bar x$. Suppose for the purpose of contradiction that $\bar x_{k_1} < 1/K - \ep_0$, and let us write $\eta := 1/K - \ep_0$. It follows that
    \begin{align}
\sum_{i=1}^{\sigma K} \bar x_{k_i} \leq \eta + (1-\eta) \cdot \frac{\sigma K-1}{K-1}\nonumber.
    \end{align}
    Then we have
    \begin{align}
      \frac{1}{\sigma K} \sum_{i=1}^{\sigma K} \bar x_{k_i} %
      \leq \frac{1}{K} - \frac{\ep_0}{K-1} \cdot (1/\sigma - 1) \leq \frac{1}{K} - \frac{\ep_0}{2\sigma(K-1)}\nonumber.
    \end{align}
    Letting $y' \in \MK_{\sigma, N}$ be defined by $y'_{2k_i-1} = y'_{2k_i} = 1/(2\sigma K)$ for all $i \in [\sigma K]$, we have that
    \begin{align}
1 \geq x^\t B y' - x^\t By \geq -M \cdot \sum_{k \in [K]} \bar x_k \bar y_k' + M \cdot \sum_{k \in [K]} \bar x_k \bar y_k  - 1 \geq \frac{M\ep_0}{2\sigma (K-1)} -3\nonumber.
    \end{align}
    The above inequality leads to a contradiction since $M \ep_0 / (2\sigma (K-1)) > 4$. 
    We have established that $\bar x_k \geq 1/K - \ep_0 \geq 1/K - \ep$ for all $k \in [K]$, and it follows that $\bar x_k \leq 1/K + K \ep_0 = 1/K + \ep$ for all $k \in [K]$. 
  \end{proof}

  We next recall the following result from \cite{DBLP:journals/jacm/ChenDT09}, which shows that computing approximate Nash equilibria in GMP games is \PPAD-hard. The result is not explicitly stated in the paper; below we explain how it follows as an immediate consequence of the intermediate results in \cite{DBLP:journals/jacm/ChenDT09}.
\begin{theorem}[\cite{DBLP:journals/jacm/ChenDT09}]
  \label{thm:gmp-hardness}
For $\eta \leq K^{-14}$, the problem of computing $2K^3 \cdot \eta$-approximate Nash equilibrium in approximate $K$-GMP games is \PPAD-hard. %
\end{theorem}
The proof of \cref{thm:gmp-hardness} is essentially given in \cite{DBLP:journals/jacm/ChenDT09}, but we give the details for completeness.
\begin{proof}[Proof of \cref{thm:gmp-hardness}]
  We use the notation and terminology from \cite{DBLP:journals/jacm/ChenDT09}. First, by \cite[Theorem 4.6 \& Lemma 5.7]{DBLP:journals/jacm/ChenDT09}, we see that it is \PPAD-hard to compute an $\ep$-approximate solution to a generalized circuit instance with $K$ nodes for $\ep = 1/K^{3}$.\footnote{In fact, \cite{rubinstein2015inapproximability} showed that this problem is \PPAD-hard even for $\ep = O(1/K)$, though we shall not need that stronger result.} Next,  \cite{DBLP:journals/jacm/ChenDT09} establishes that given a generalized circuit instance $\MC$ with $K$ nodes, we can efficiently construct an approximate $K$-GMP game $(A,B)$ for which the following holds. For $\eta  > 0$, a $2M \eta$-approximate Nash equilibrium of the game $(A,B)$ is a $\eta$-approximate Nash equilibrium of the game $(\frac{1}{2M} \cdot A,\frac{1}{2M} \cdot B)$ (whose payoffs are normalized to lie in $[0,1]$). Given such an equilibrium, by Lemma 2.2 of \cite{DBLP:journals/jacm/ChenDT09} we can construct a $\eta'$-well-supported Nash equilibrium $(x', y')$ of $(A,B)$, for $\eta ' = 2M \sqrt{8\eta}$.
Moreover, Lemmas 6.3 \& 6.4 of \cite{DBLP:journals/jacm/ChenDT09} and the construction of the game $(A,B)$ give that, as long as $\eta' \leq 1/K^3$, for an $\eta'$-well-supported Nash equilibrium $(x', y')$, the vector $\bar x' \in [0,1]^K$ is a $1/K^3$-approximate solution to the generalized circuit instance $\MC$. Recalling that $M = 2K^3$, we have $\eta' =2M\sqrt{8\eta} \leq 1/K^3$ as long as $\eta \leq 1/K^{14}$ and $K \geq 8$ (the latter of which is without loss of generality). 
\end{proof}

Using \cref{lem:validity,thm:gmp-hardness}, we may show that computing weak approximate smooth Nash equilibria is hard:
\begin{lemma}
  \label{lem:sigconst-ep15}
For any constant $\sigma \leq 1/6$, the problem of computing weak $n^{-15}$-approximate $\sigma$-smooth Nash equilibria in 2-player $n$-action games is \PPAD-hard.
\end{lemma}
\begin{proof}
  Fix $K \in \BN$ and consider any approxiamte $K$-GMP game $(A,B)$. 
  Set $\ep = K^{-15}$, and suppose that $(x,y)$ is a weak $\ep$-approximate $\sigma$-smooth Nash equilibrium of the normalized game $(\bar A, \bar B) := (\frac{1}{2K^3} \cdot A, \frac{1}{2K^3} \cdot B)$, which we denote by $\bar G$. Since $\ep \cdot 2K^3 \leq 1$, it follows from \cref{lem:validity} that, for any $\ep' \geq 8 \sigma K^2/(2K^3)$, we have $\bar x_k, \bar y_k \in [1/K-\ep', 1/K + \ep']$. We choose $\ep' = 1/K$, which satisfies the necessary constraint as long as $\sigma \leq 1/4$. We claim that $(x,y)$ is an $\ep K$-approximate Nash equilibrium of $\bar G$. If this were not the case, then one player (say the $x$-player) can increase their utility by more than $\ep K $ by deviating, and in particular, for some $i \in [N]$, $ (e_i - x)^\t Ay > \ep K $. Thus, defining $x' := (1-1/K) \cdot x + 1/K \cdot e_i$, we have $(x' - x)^\t Ay > K$, and $x' \in \MK_{\sigma, n}$ since
  \[\max_{i \in [N]} x_i' \leq 1/K + \max_{i \in [N]} x_i \leq 1/K + \max_{k \in [K]} \bar x_k \leq 3/K = 6/N \leq \frac{1}{\sigma N}
  \]
  as we have assumed $\sigma \leq 1/6$. This contradicts the fact that $(x,y)$ is a weak $\ep$-approximate $\sigma$-smooth Nash equilibrium of $\bar G$.  Since it is \PPAD-hard to compute a $K^{-14}$-approximate Nash equilibrium of normalized $K$-GMP games $\bar G$ (by \cref{thm:gmp-hardness}), it follows by our choice of $\ep = K^{-15}$ above that it is \PPAD-hard to compute a weak $K^{-15}$-approximate $\sigma$-smooth Nash equilibrium in 2-player normal-form games (whose payoffs are normalized to $[0,1]$).
\end{proof}

Finally, to prove \cref{thm:ppad-hard-const-sig}, we simply apply the padding result of \cref{thm:smooth-nash-nash} to the result of \cref{lem:sigconst-ep15} above.
\begin{proof}[Proof of \cref{thm:ppad-hard-const-sig}]
The proof follows by combining \cref{lem:sigconst-ep15,thm:smooth-nash-nash}. In particular, we apply \cref{thm:smooth-nash-nash} with $\ep(n) = n^{-15}$ and $\sigma(n) = 1/6$. The lemma then establishes the following: for any $c \in (0,1)$, the problem of finding a weak $n^{-15}$-approximate $1/6$-smooth Nash equilibrium in 2-player, $n$-action games has a polynomial-time reduction to the problem of finding a weak $n^{-15c}$-approximate $1/6$-smooth Nash equilibrium in 2-player, $n$-action games. Taking $c = c_1/15$ and applying \cref{lem:sigconst-ep15} yields the desired conclusion.
\end{proof}

\section{Fast Rates for Zero Sum Games} \label{sec:zero-sum} 

In this section, we will discuss the special case of two player zero-sum games.
As discussed in \cref{sec:other-equilibria}, smooth equilibria in zero sum games have been studied in the literature in connection to boosting and hard-core lemmas \cite{v008a006,10.5555/2207821,DBLP:journals/eccc/Kale07,DBLP:conf/soda/BarakHK09}. 

Formally, we are in the setting where there is a game matrix $A \in \mathbb{R}^{n \times n} $ and the equilibrium corresponds to a pair of strategies $x$ and $y$ such that 
\begin{align}
    \min_{x \in K_{ \sigma , n } } \max_{ y \in K_{\sigma , n}  }  x^{\top} A y.  
\end{align} 

In this setting, it is a classical result that even in the unrestricted case i.e. $x, y \in \Delta^n$, there are efficient algorithms for (approximately) finding the equilibrium.
This follows from a close connection between the minimax equilibrium problem and the notion of no-regret learning. 
In particular, using this connection one can show that an $\epsilon$-approximate minimax equilibrium can be found by running $ O(  \epsilon^{-2}  \log n     )  $ iterations of a no-regret learning algorithm playing against another algorithm best responding. 
In addition, one can design no-regret learning algorithms such that if both players play according to the algorithm, then they converge to an $\epsilon$-approximate minimax equilibrium in $ O(  \epsilon^{-1}  \log n     )  $ iterations \cite{OLG}. 

In this section, we note that we can achieve a faster converge to equilibrium if we restrict to the setting of $x, y \in K_{\sigma , n} $.
As a warmup, we will first show that a simple modification of the multiplicative weights update algorithm gives us an algorithm with rate of convergence $ O (  \epsilon^{-2}  \log(1/\sigma))$.
Note that \cref{alg:zero-sum} specifies the algorithm for one player. 
The bound below is stated for both players using the algorithm using the action of the other player as the losses.
Similar bounds have been observed in literature \cite{v008a006,DBLP:journals/eccc/Kale07,DBLP:journals/eccc/Kale07,DBLP:conf/nips/HaghtalabRS20} but we include it here for completeness.

\begin{theorem}[Regret for Multiplicative Weights Update]
    Let $G$ be a two player zero-sum game with game matrix $A \in \mathbb{R}^{n \times n} $, and $\sigma \in (0,1)$ be given.
    Let both players use \cref{alg:zero-sum} to produce a sequence of strategies $x_t$ and $y_t$ respectively. 
    Then, we have that $ \bar{x} = 1/T \sum_i x_i  $ and $ \bar{y} = 1/T \sum_i y_i  $ is a strong $ \epsilon$-approximate $\sigma$-smooth Nash equilibrium,  where   
     \begin{align}
        \epsilon = O \left(  \sqrt{ \frac{ \log \left( 1 / \sigma \right) }{T} } \right). 
     \end{align}
\end{theorem}

This bound follows from the general regret bound for multiplicative weights update in terms of KL divergence \cite[Theorem 2.4]{v008a006} and the standard connection between no-regret learning and minimax equilibrium \cite[Section 3.2]{v008a006}. 
  The key idea that leads to the improved bounds is the fact that the set $ K_{\sigma , n} $, we have that $ \max_{x \in K_{\sigma , n} }  \mathrm{D}_{\mathrm{KL}} \left( x , \mu_n \right)    \leq \log(1 / \sigma)   $ where $ \mu_n $ is the uniform distribution on $ [n] $.
  The regret bound and the corresponding rate of convergence would remain true even if one of the players best responded but the projection step is necessary to ensure that both players are playing in the set $ K_{\sigma , n} $ as would be required for strong smooth equilibria.

\begin{algorithm}[t]
    \caption{Projected Multiplicative Weights for 2-player Zero-sum games}
    \label{alg:zero-sum}
    \begin{algorithmic}[1]\onehalfspacing
    \State Set $ x_{i,0} = \frac{1}{n} $ for all $i$.
      \For{ For $t \leq T $  }
      \State Play according to the distribution $ x_t =  x_{\cdot , t} $.
      \State Receive loss vector $\ell_{ \cdot,  t} = A y_t   $ corresponding to the actions of the adversary $y_t$  
      \State Update $ \tilde{x}_{i,t+1}  \propto x_{i,t} \left( 1 - \eta \ell_{i,t} \right)$ where $\ell_{i,t}$ is the loss suffered by action $i$ with respect to the adversary's action at time $t$. 
      \State Update $ x_{t+1} =  x_{ \cdot , t+1 } = \argmin_{p \in K_{\sigma , n } } \mathrm{D}_{\mathrm{KL}} \left( p || \tilde{x}_{\cdot , t+1}  \right)   $
      :%
        \EndFor
    \end{algorithmic}
  \end{algorithm}

Additionally, in the setting where we have access to the game matrix $A$, it is possible to get a faster rate of convergence to equilibrium using the additional fact that each player is using a no-regret learning algorithm (the usual regret bounds are designed for the case that the loss sequence is chosen adversarially but in the case of game solving we can use the fact that the loss sequence is chosen by the other player who playing a no regret learning algorithm).
In particular, it is known that several algorithms for no-regret learning when used together in a two player zero-sum game can converge to equilibrium at a rate of $ O(  \epsilon^{-1}  \log(n))$. 
We use the optimistic mirror descent algorithm \cite{OLG} with the relative entropy regularization which is a variant of the multiplicative weights update algorithm encorporating the additional stability of the loss sequence.
As stated before, the main idea that leads to $ \log (1/ \sigma )  $ is the bound on the KL divergence between the uniform distribution and elements in the set $ K_{\sigma , n} $.
The proof of this fact is analogous to \cite{OLG} and we include a proof in \cref{sec:proof_zero_sum} for completeness.

\begin{algorithm}[t]
    \caption{Optimistic Mirror Descent for 2-player Zero-sum games}
    \label{alg:zero-sum-opt}
    \begin{algorithmic}[1]\onehalfspacing
        \State Set $ x_{i,0} = \frac{1}{n} $ for all $i$.  
    \For{ For $t \leq T $  }
      \State Play according to the distribution $ x_{t} $.
      \State Receive loss vector $\ell_{ \cdot,  t} = A y_t   $ corresponding to the actions of the adversary $y_t$
      \State Update $  \tilde{x}_{t}  = \argmin_{p   \in  K_{\sigma,n} }  \eta_t \ip{p}{  \ell_t } +  \mathrm{D}_{\mathrm{KL}} \left( p || \tilde{x}_{\cdot , t-1}  \right)      $
      \State Update $ x_{ t+1 } = \argmin_{p \in K_{\sigma , n } }  \eta_t \ip{ p }{ \ell_t } +  \mathrm{D}_{\mathrm{KL}} \left( p || \tilde{x}_{ t}  \right)   $
      :%
        \EndFor
    \end{algorithmic}
  \end{algorithm}

\newcommand{\cR}{\mathcal{R}}
\newcommand{\inner}[1]{ \left\langle #1 \right\rangle }
\newcommand{\cD}{\mathcal{D}}

\begin{theorem} \label{lem:zero_sum_fast_rates}
    Let $A \in \mathbb{R}^{n \times n} $ be a game matrix and $\sigma \in (0,1)$ be given.  
    Assume both players play $ x_t $ and $ y_t $ respectively according to \cref{alg:zero-sum-opt} for an appropriate choice of step sizes $ \eta_t $.
    Then, we have that $  \bar{x} = 1 / T \sum_i x_i  $ and $   \bar{y}  = 1/ T \sum_i y_i $ is a strong $\epsilon$-approximate $\sigma$-smooth Nash equilibrium where
    \begin{align}
        \epsilon =  O \left( \frac{\log \left( 1 / \sigma \right)  }{T}     \right).   
    \end{align} 
\end{theorem}

\section{Proof of \cref{lem:zero_sum_fast_rates}} \label{sec:proof_zero_sum} 

We prove the following stronger version of \cref{lem:zero_sum_fast_rates}, which specifies the step sizes $\eta_t, \eta_t'$ for the two players; it is immediate that \cref{lem:zero_sum_fast_rates_app} below implies \cref{lem:zero_sum_fast_rates}. 

\begin{theorem} \label{lem:zero_sum_fast_rates_app}
    Let $A \in \mathbb{R}^{n \times n} $ be a game matrix and $\sigma \in (0,1)$ be given.  
    Assume both players play $ x_t $ and $ y_t $ respectively according to \cref{alg:zero-sum-opt} for an appropriate choice of step sizes $ \eta_t $ and $ \eta'_t $.
    Then, for
    \begin{align}
        \eta_t = \min\left\{  R_1^2  \left(\sqrt{\sum_{i=1}^{t-1} \norm{Ay_i - Ay_{i-1}}_*^2} + \sqrt{\sum_{i=1}^{t-2} \norm{Ay_i - Ay_{i-1}}_*^2} \right)^{-1} , \frac{1}{11}\right\}\nonumber 
    \end{align}
    and 
    \begin{align}
        \eta'_t = \min\left\{ R_2^2\left(\sqrt{\sum_{i=1}^{t-1} \norm{x_i^{\top} A - x_{i-1}^{\top} A}_*^2} + \sqrt{\sum_{i=1}^{t-2} \norm{x_i^{\top} A - x_{i-1}^{\top} A}_*^2} \right)^{-1} , \frac{1}{11}\right\} \nonumber
    \end{align}
    we have that $ \left(1 / T \sum_i f_i , 1 / T \sum_i x_i  \right) $ is a strong 
    $
        O \left(  \frac{ \log(1/\sigma)   }{T} \right)
$-approximate $\sigma$-smooth Nash equilibrium, where $ R_{i}^2 = \max_{f \in K_{\sigma, n}  }  \mathrm{D}_{\mathrm{KL}} \left( p, x_0  \right)   $
\end{theorem}

\begin{proof}
    Let $\cR(f) = \sum_{i=1}^n f(i)\ln f(i)$. 
    These functions are  strongly convex with respect to $\|\cdot\|_1$ norm on $ K_{\sigma , n} $. 
     We first upper bound regret of Player I, writing $\nabla_t$ as a generic observation vector, later to be chosen as $Ay_t$, and $M_t$ as a generic predictable sequence, later chosen to be $Ay_{t-1}$. 
     Let $x^*  \in K_{  \sigma , n  } $ be a comparator strategy. 
     Then,
\begin{align}
    \inner{x_t-x^*, \nabla_t} = \inner{x_t- \tilde{x}_{t},\nabla_t-M_{t}} + \inner{x_t- \tilde{x}_{t}, M_t} + \inner{ \tilde{x}_{t}-x^*, \nabla_t}\nonumber
\end{align}
By the update rule,
\begin{align}
    \inner{x_t- \tilde{x}_{t}, M_t} \leq \frac{1}{\eta_t} \left( -     \mathrm{D}_{\mathrm{KL}}( \tilde{x}_{t},x_t) - \mathrm{D}_{\mathrm{KL}}(x_t, \tilde{x}_{t-1})  \right) \nonumber
\end{align}
and
\begin{align}
    \inner{ \tilde{x}_{t}-x^*, \nabla_t} & \leq \frac{1}{\eta_t}\left(\mathrm{D}_{\mathrm{KL}}(x^*, \tilde{x}_{t-1})  - \mathrm{D}_{\mathrm{KL}}(x^*, \tilde{x}_{t})  \right)  \ .\nonumber
\end{align}
We conclude that $\inner{x_t-x^*, \nabla_t} $ is upper bounded by
\begin{align}
    &\norm{\nabla_t - M_t}_* \norm{x_t - \tilde{x}_{t}} +\frac{1}{\eta_t} \left( - \mathrm{D}_{\mathrm{KL}}( \tilde{x}_{t},x_t) - \mathrm{D}_{\mathrm{KL}}(x_t, \tilde{x}_{t-1}) \right) 
    + \frac{1}{\eta_t}\left(\mathrm{D}_{\mathrm{KL}}( x^*, \tilde{x}_{t-1})  - \mathrm{D}_{\mathrm{KL}}( x^*, \tilde{x}_{t})  \right)\notag \\
    &= \norm{\nabla_t - M_t}_* \norm{ x_t - \tilde{x}_{t}} 
 + \frac{1}{\eta_t}\left(\mathrm{D}_{\mathrm{KL}}(x^*, \tilde{x}_{t-1})  - \mathrm{D}_{\mathrm{KL}}( x^*, \tilde{x}_{t})  - \mathrm{D}_{\mathrm{KL}}( \tilde{x}_{t},x_t) - \mathrm{D}_{\mathrm{KL}}(x_t, \tilde{x}_{t-1}) \right) \notag
\end{align}
Using strong convexity, the term involving the four divergences can be further upper bounded by
\begin{align}
     &\frac{1}{\eta_t}\left(\mathrm{D}_{\mathrm{KL}}(x^*, \tilde{x}_{t-1})  - \mathrm{D}_{\mathrm{KL}}(x^*, \tilde{x}_{t})  - \frac{1}{2}\norm{ \tilde{x}_{t} - x_t}^2 - \frac{1}{2}\norm{ \tilde{x}_{t-1} - x_t}^2 \right)\nonumber \\ 
    &=  \frac{1}{\eta_t}\left(\mathrm{D}_{\mathrm{KL}}( x^*, \tilde{x}_{t-1})  - \mathrm{D}_{\mathrm{KL}}(x^*, \tilde{x}_{t})  - \frac{1}{2}\norm{ \tilde{x}_{t} - x_t}^2 - \frac{1}{2}\norm{ \tilde{x}_{t-1} - x_t}^2 \right) \nonumber\\
    &= \frac{1}{\eta_t}\left(\mathrm{D}_{\mathrm{KL}}(x^*, \tilde{x}_{t-1})  - \mathrm{D}_{\mathrm{KL}}( x^*, \tilde{x}_{t})  - \frac{1}{2}\norm{ \tilde{x}_{t} - x_t}^2 - \frac{1}{2}\norm{ \tilde{x}_{t-1} - x_t}^2 \right)\nonumber. 
\end{align}

Using the above in the bound on $\inner{x_t-x^*, \nabla_t} $ and summing over $t=1,\ldots,T$, and using the fact that the step size are non-increasing, we conclude that 
\begin{align}
\sum_{t=1}^T \inner{x_t-x^*, \nabla_t} & \le \eta_1^{-1} \mathrm{D}_{\mathrm{KL}}(x^*,\tilde{x}_0) +   \sum_{t=2}^T  \mathrm{D}_{\mathrm{KL}}(x^*,\tilde{x}_{t-1})\left(\frac{1}{\eta_{t}}  - \frac{1}{\eta_{t-1}}\right)  +  \sum_{t=1}^T \norm{\nabla_t - M_t}_* \norm{ \tilde{x}_{t} - x_t} \notag\\
& ~~~~~  - \sum_{t=1}^T \frac{1}{2 \eta_t} \left(\norm{ \tilde{x}_{t} - x_t}^2 +  \norm{\tilde{x}_{t-1} - x_t}^2 \right) \notag \\
& \leq   (\eta_1^{-1} + \eta_T^{-1}) R_{1}^2 +  \sum_{t=1}^T \norm{\nabla_t - M_t}_* \norm{\tilde{x}_{t} - x_t}   - \frac{1}{2 } \sum_{t=1}^T  \eta_t^{-1}\left(\norm{\tilde{x}_{t} - x_t}^2 +  \norm{\tilde{x}_{t-1} - x_t}^2 \right) \notag\\
\label{eq:regbndet} \ .
\end{align}
where $R_{1}^2$ is an upper bound on the largest KL divergence between $x^*$ and any $\tilde{x}$. 
Hence we conclude that a bound on regret of Player I is given by
\begin{align} 
    \label{eq:player_1_regret} 
\sum_{t=1}^T \inner{x_t-\tilde{x}^*, \nabla_t} & \le		(\eta_1^{-1} + \eta_T^{-1}) R_{1}^2 +  \sum_{t=1}^T \norm{Ay_t - Ay_{t-1}}_* \norm{\tilde{x}_{t} - x_t}   - \frac{1}{2 } \sum_{t=1}^T  \eta_t^{-1}\left(\norm{\tilde{x}_{t} - x_t}^2 +  \norm{\tilde{x}_{t-1} - x_t}^2 \right) + 1  
\end{align}
Observe that
$$\eta_t = \min\left\{R^2_{1} \frac{\sqrt{\sum_{i=1}^{t-1} \norm{Ay_i - Ay_{i-1}}_*^2} - \sqrt{\sum_{i=1}^{t-2} \norm{Ay_i - Ay_{i-1}}_*^2}}{\norm{Ay_{t-1} - Ay_{t-2}}_*^2} , \frac{1}{11}\right\}$$
and 
$$
\textstyle 11\leq \eta_t^{-1} \le \max\left\{2R^{-2}_{1} \sqrt{\sum_{i=1}^{t-1} \norm{Ay_{i} - Ay_{i-1}}_*^2} , 11\right\}
$$
With this, the upper bound on Player I's unnormalized regret is
\begin{align}
    1 + 22 R_{1}^2 &+ 2\sqrt{\sum_{t=1}^{T-1} \norm{Ay_t - Ay_{t-1}}_*^2} + \sum_{t=1}^T \norm{Ay_t - Ay_{t-1}}_* \norm{\tilde{x}_{t} - x_t}   - \frac{11}{2} \sum_{t=1}^T  \left(\norm{ \tilde{x}_{t} - x_t}^2 +  \norm{\tilde{x}_{t-1} - x_t}^2 \right)\nonumber.  
\end{align}

Adding the regret of the second player who uses step size $\eta'_t$, the overall bound on the suboptimality is
\begin{align}
& 2 + 22 R_{1}^2 +  2 \sqrt{\sum_{t=1}^{T-1} \norm{Ay_t - Ay_{t-1}}^2_*} +  \sum_{t=1}^T \norm{Ay_t - Ay_{t-1}}_* \norm{ \tilde{x}_{t} - x_t} \nonumber \\
& + 22 R_{2}^2 +  2 \sqrt{\sum_{t=1}^{T-1} \norm{ x_t^{\top} A - x_{t-1}^{\top} A}^2_*}  +  \sum_{t=1}^T \norm{x_t^{\top} A - x_{t-1}^{\top} A}_* \norm{y_{t} - \tilde{y}_t} \nonumber \\ 
& - \frac{11}{2} \sum_{t=1}^T \left(\norm{ \tilde{x}_{t} - x_t}^2 +  \norm{\tilde{x}_{t-1} - x_t}^2 \right) - \frac{11}{2} \sum_{t=1}^T   \left(\norm{ y_{t} - \tilde{y}_t}^2 +  \norm{\tilde{y}_{t-1} - y_t}^2 \right)\nonumber.
\end{align}

Using $\sqrt{c}\leq c+1$ for $c\geq 0$, we obtain an upper bound
\begin{align}
\sum_{t=1}^T \inner{ x_t-x^*, \nabla_t} \le & \ 6 + 22 R_{1}^2 +  2\sum_{t=1}^{T-1} \norm{Ay_t - Ay_{t-1}}^2_* +  \sum_{t=1}^T \norm{Ay_t - Ay_{t-1}}_* \norm{ \tilde{x}_{t} - x_t}\nonumber  \\
& + 22 R_{2}^2 +  2\sum_{t=1}^{T-1} \norm{ x_t^{\top} A - x_{t-1}^{\top} A}^2_*  +  \sum_{t=1}^T \norm{ x_t^{\top} A - x_{t-1}^{\top} A}_* \norm{ \tilde{y}_{t} - y_t} \nonumber \\
& - \frac{11}{2} \sum_{t=1}^T \left(\norm{ \tilde{x}_{t} - x_t}^2 +  \norm{ \tilde{x}_{t-1} - x_t}^2 \right) - \frac{11}{2} \sum_{t=1}^T   \left(\norm{ \tilde{y}_{t} - y_t}^2 +  \norm{ \tilde{y}_{t-1} - y_t}^2 \right)\nonumber\\
& \le 6 + 22 R_{1}^2 +  \frac{5}{2} \sum_{t=1}^T \norm{ Ay_t - Ay_{t-1}}^2_*  +  \frac{1}{2}\sum_{t=1}^T \norm{ \tilde{x}_{t} - x_t}^2 \nonumber \\
& + 22 R_{2}^2 + \frac{5}{2} \sum_{t=1}^T \norm{ x_t^{\top} A - x_{t-1}^{\top} A}^2_*    +  \frac{1}{2} \sum_{t=1}^T \norm{ \tilde{y}_{t} - y_t}^2 \nonumber \\
& - \frac{11}{2} \sum_{t=1}^T \left(\norm{ \tilde{x}_{t} - x_t}^2 +  \norm{ \tilde{x}_{t-1} - x_t}^2 \right) - \frac{11}{2} \sum_{t=1}^T \left(\norm{ \tilde{y}_{t} - y_t}^2 +  \norm{ \tilde{y}_{t-1} - y_t}^2 \right).\nonumber
\end{align}
Since each entry of the matrix is bounded by $1$, 
$$
\norm{Ay_t - Ay_{t-1}}^2_* \le \norm{y_t - y_{t-1}}^2 \le 2 \norm{x_t - \tilde{x}_{t-1}}^2 + 2 \norm{x_{t-1} - \tilde{x}_{t-1}}^2.
$$
A similar inequality holds for the other player too. This leads to an upper bound of
\begin{align}
& 6 + 22 R_{1}^2 + 22 R_{2}^2 +  \frac{1}{2} \sum_{t=1}^T \norm{y_{t} - x_t}^2  +  \frac{1}{2}\sum_{t=1}^T \norm{ \tilde{x}_{t} - x_t}^2  \notag \\
& ~~~+ 5 \sum_{t=1}^T \left(\norm{ \tilde{x}_{t} - x_t}^2 +  \norm{ \tilde{x}_{t-1} - x_t}^2 \right)   + 5 \sum_{t=1}^T \left(\norm{ \tilde{y}_{t} - y_t}^2 +  \norm{ \tilde{y}_{t-1} - y_t}^2 \right)  \notag\\
& ~~~ - \frac{11}{2} \sum_{t=1}^T \left(\norm{ \tilde{x}_{t} - x_t}^2 +  \norm{ \tilde{x}_{t-1} - x_t}^2 \right) -  \frac{11}{2} \sum_{t=1}^T \left(\norm{ \tilde{y}_{t} - y_t}^2 +  \norm{ \tilde{y}_{t-1} - y_t}^2 \right)  \notag \\
& \le 6 + 22 R_{1}^2 + 22 R_{2}^2 \notag,
\end{align}
as required. The desired result about the approximate equilibrium follows from the standard bound connecting the regret and the minimax equilibrium. See for example \cite[Section 3.2]{v008a006}. 
\end{proof}

\section{Quantal Response Equilibria \& Smooth Nash Equilibria}
\label{sec:qre-se}
In this section, we discuss the relationship between smooth Nash equilibria and quantal response equilibria. %
We first formally define quantal response equilibria. Let $P$ denote some real-valued probability distribution, which has an integrable density $f : \BR \ra \BR$, and whose cumulative distribution function is denoted by $F(x) := \Pr_{X \sim P}(X \leq x)$. Fix $n \in \BN$, denoting the number of actions; we let $P_n$ denote the distribution on $\BR^n$ consisting of $n$ independent copies of $P$, and $f_n$ denote the corresponding density. For any vector $u \in \BR^n$ and $i \in [n]$, define the \emph{$i$-response set} $\MR_{i}(u)$ as follows:
\begin{align}
  \MR_{i}(u) := \left\{ \vep \in \BR^n \ : \ u_i + \vep_i \geq u_{i'} + \vep_{i'} \ \forall i' \in [n] \right\}\nonumber.
\end{align}
In words, $\MR_i(u)$ denotes the set of error vectors which, for (true) utility vector $u$, will lead an agent to choose action $i$. The \emph{quantal response function} of an agent, for noise distribution $P$, is then defined as
\begin{align}
\tau_{P,i}(u) := \Pr_{\vep \sim P_n} \left( \vep \in \MR_i(u) \right) = \int_{\MR_i(u)} f_n(\vep) d\vep\nonumber,
\end{align}
for $i \in [n]$. We let $\tau_P(u) \in \Delta^n$ denote the vector $(\tau_{P,1}(u), \ldots, \tau_{P,n}(u))$. 
The quantal response function $\tau_{P,i}(u)$ describes the probability that an agent will choose action $i$, given noise distribution $P$ and (true) utility vector $u$. A quantal response equilibrium is then defined as a strategy profile in which each agent chooses actions with probabilities given by its quantal response function:
\begin{definition}[Quantal response equilibrium]
  \label{def:qre}
  Let $A_1, \ldots, A_m : [n]^m \ra [0,1]$ be the payoff matrices of an $m$-player normal-form game. A strategy profile $x = (x_1, \ldots, x_j) \in (\Delta^n)^m$ is a \emph{quantal response requilibrium} corresponding to the noise distribution $P$ if, for each $j \in [m]$,
  \begin{align}
x_j = \tau_P(\bar A_j(x_{-j}))\nonumber,
  \end{align}
  where $\bar A_j(x_{-j}) \in \BR^n$ denotes the vector whose $i$th entry is $A_j(i, x_{-j})$. %
\end{definition}
We remark that we have assumed for simplicitly that the noise vectors $\vep$ are product distributions, and are independently and identically distributed across players. While this assumption is satisfied by essentially all instantiations of quantal response equilibrium studied in the literature, it is possible to define quantal response equilibria for more general noise distributions \cite{mckelvey1995quantal}. 

In \cref{prop:qre-sne}, we show that even in the case of a single player ($m=1$), for any noise distribution $P$, the set of quantal response equilibria and strong smooth Nash equilibrium can be disjoint, and in fact separated in total variation distance: %
\begin{proposition}
  \label{prop:qre-sne}
  Fix any $\sigma \in [16/n, 1]$. 
For any real-valued probability distribution $P$ which has an integrable density, there is some 1-player game $G$ so that the total variation distance between any quantal response equilibrium of $G$ with respect to $P$ and any  strong $\sigma$-smooth Nash equilibrium of $G$ is at least $\Omega(1/n)$. %
\end{proposition}
\begin{proof}
  Let the density of $P$ be denoted by $f$, the cumulative distribution function of $P$ be denoted by $F$, and let the payoffs of the player, which can be expressed as a $n$-dimensional vector, be denoted by $A : [n] \ra [0,1]$. There is a unique quantal response equilibrium $x^\st \in \Delta^n$ corresponding to $P$, defined by $x^\st = \tau_P(A)$. Our goal is to show that $x^\st$ is not a strong $\sigma$-smooth Nash equilibrium. Let $y \in \BR$ be defined so that $F(y) = 1-1/n$; such $y$ exists because $P$ has a density. We consider the following two cases:

  \paragraph{Case 1: $F(y-1) \leq 1-\frac{8}{\sigma n}$.} Suppose the payoffs are given as follows: $A(1) = 1$ and $A(i) = 0$ for $i > 1$. Then, using \cref{def:qre}, the quantal response equilibrium $x^\st$ satisfies:
  \begin{align}
x^\st_1 = \int_{\BR} f(z) \cdot F(z+1)^{n-1} dz \geq \int_{[y-1, \infty)} f(z) \cdot F(z+1)^{n-1} dz \geq \int_{[y-1,\infty)} f(z) \cdot (1-1/n)^{n-1} dz \geq \frac{2}{\sigma n}\nonumber,
  \end{align}
  where the second inequality uses that $F(z+1) \geq F(y) \geq 1-1/n$ for $z \geq y-1$, and the final inequality uses $(1-1/n)^{n-1} \geq 1/4$ for $n \geq 2$ as well as the fact that $\int_{[y-1,\infty)} f(z) dz = 1- F(y-1) \geq 8/(\sigma n)$. Since any strong $\sigma$-smooth Nash equilibrium $x$ puts mass at most $1/(\sigma n)$ on each action, we must have that the total variation distance between $x$ and $x^\st$ is at least $1/(\sigma n) > 1/n$. %

  \paragraph{Case 2: $F(y-1) > 1-\frac{8}{\sigma n}$.} Suppose the payoffs are given by $A(i) = 1$ for $i \leq \lceil \sigma n \rceil$ and $A(i) = 0$ for $i > \lceil \sigma n \rceil$. Let us write $s := \lceil \sigma n \rceil$. Any strong $\sigma$-smooth Nash equilibrium puts 0 mass on actions $i > s$. However, %
  \begin{align}
   \sum_{i =s+1}^n x_i^\st =& \sum_{i=s+1}^n \Pr_{\vep \sim P_n} \left( \vep_i > \vep_{i'} + 1 \ \forall i' \in [s], \quad \vep_i > \vep_{i'} \ \forall i' \not \in \{ 1, \ldots, s, i \} \right)\nonumber\\
    \geq & \Pr_{\vep \sim P_n} \left( \vep_{s+1} > \vep_i + 1 \ \forall i \in [s] \right) \nonumber\\
    = & \int_{\BR} f(z) \cdot F(z-1)^s dz \nonumber\\
    \geq & \int_{[y,\infty)} f(z) \cdot F(z-1)^s dz \nonumber\\
    \geq & \int_{[y,\infty)} f(z) \cdot \left(1-\frac{16}{s}\right)^s dz \geq c/ n\nonumber,
  \end{align}
  where $c > 0$ is a positive constant, the second-to-last inequality above uses the fact that $F(z-1) \geq F(y-1) > 1-8/(\sigma n) \geq 1-16/s > 0$ for $z \geq y$, and the final inequality uses that $\int_{[y,\infty)} f(z) dz = 1-F(y) = 1/n$. It follows that the total variation distance between $x^\st$ and any strong $\sigma$-smooth Nash equilibrium is at least $c/n$. 
\end{proof}
One cannot replace ``strong'' with ``weak'' in the statement of \cref{prop:qre-sne}: if $P$ is the point mass at 0, then the set of quantal response equilibria of any $G$ with respect to $P$ is exactly the set of Nash equilibria, which is a subset of all weak $\sigma$-smooth Nash equilibria.\footnote{Technically, the point mass at 0 does not have a density, so \cref{prop:qre-sne} would not apply regardless. But we can correct for  this by considering a sequence of distributions which approaches the point mass at 0, and recover essentially the same conclusion.} We do have, however, the following weaker fact pertaining to weak smooth Nash equilibria:
\begin{proposition}
  \label{prop:weak-qre}
For any $\sigma \in (2/n, 1)$, and any real-valued probability distribution $P$, there is some 1-player game $G$ so that the quantal response equilibrium of $G$ with respect to $P$ is $\Omega(1)$-far (in total variation distance) from \emph{some} $\sigma$-smooth Nash equilibrium of $G$.
\end{proposition}
\begin{proof}
Consider the game whose payoffs are defined by $A : [n] \ra [0,1]$ with $A(1) = 1$ and $A(i) = 0$ for $i > 1$. Note that the distributions $x := (1, 0, \ldots, 0)$ and $x' := (1/2, 1/2, 0, \ldots, 0)$ are both weak $\sigma$-smooth Nash equilibria of $G$, since $\sigma \geq 2/n$. Since $x,x'$ are $1/2$-far in total variation distance, at least one of them must be $1/4$-far from the unique quantal response equilibrium. 
\end{proof}

\subsection{Logit equilibrium}
We consider the special case that the perturbation distribution is given by the extreme value distribution $P_\lambda$ whose cumulative distribution function is $F(y) := \exp(-\exp(-\lambda y ))$, for some parameter $\lambda > 0$. In this case, it may be seen \cite[Lemma 1]{mcfadden1973conditional} that the quantal response function is given by:
\begin{align}
\tau_{P_\lambda,i}(u) = \frac{\exp(\lambda u_i)}{\sum_{i' \in [n]} \exp(\lambda u_{i'})}\label{eq:logit-response}.
\end{align}
\cref{prop:logit-smooth} below shows that, for a given utility vector $u \in [0,1]^n$, the value obtained by the logistic response function $\tau_{P_\lambda}$ is no better than that of the best $e^{-\lambda}$-smooth distribution. In light of the extensive body of theoretical and empirical work analyzing and justifying logit response functions \cite{mcfadden1976quantal,anderson2002logit,haile2008empirical,goeree2002quantal,goeree2000asymmetric}, the class of smooth distributions should be seen as a reasonably strong deviation set, allowing one to be at least as well off as the logit response permits. 
\begin{proposition}
  \label{prop:logit-smooth}
  For any $u \in [0,1]^n$ and $\sigma \leq e^{-\lambda}$, it holds that
  \begin{align}
\lng \tau_{P_\lambda}(u), u \rng \leq \max_{x \in \MK_{\sigma ,n}} \lng x, u \rng\nonumber.
  \end{align}
\end{proposition}
\begin{proof}
  The definition of $\tau_{P_\lambda}(u)$ in \cref{eq:logit-response} gives that, for all $i,j \in [n]$,
  \begin{align}
\frac{\tau_{P_\lambda, i}(u)}{\tau_{P_\lambda, j}(u)} = \frac{\exp(\lambda u_i)}{\exp(\lambda u_j)} \leq e^\lambda\nonumber.
  \end{align}
  Thus, for any $i \in [n]$, 
  \begin{align}
1 = \sum_{j \in [n]} \tau_{P_\lambda,j}(u) \geq \sum_{j \in [n]}e^{-\lambda} \cdot  \tau_{P_\lambda, i}(u) = ne^{-\lambda} \cdot \tau_{P_\lambda, i}(u)\nonumber,
  \end{align}
  and rearranging gives that $\tau_{P_\lambda, i}(u) \leq \frac{e^\lambda}{n}$, i.e., $\tau_{P_\lambda}(u)$ is $e^{-\lambda}$-smooth. 
\end{proof}

\end{document}